\documentclass[11pt, letterpaper]{article}
\usepackage[utf8]{inputenc}
\usepackage{amsmath, amsthm, amssymb, mathrsfs, latexsym}
\usepackage{textcomp}
\usepackage[round]{natbib}
\usepackage[margin=1in]{geometry}
\usepackage{color}
\usepackage{graphicx}
\usepackage{enumitem}

\usepackage[normalem]{ulem} 

\usepackage{setspace}
\usepackage{pdflscape}

\usepackage{xr}
\usepackage{hyperref}
\hypersetup{colorlinks=true, linkcolor=black, urlcolor=black, citecolor=black}

\theoremstyle{plain}
\newtheorem{thm}{Theorem}[section]

\newtheorem{cor}[thm]{Corollary}
\newtheorem{lem}[thm]{Lemma}

\theoremstyle{remark}

\theoremstyle{definition}
\newtheorem{defn}[thm]{Definition}

\newtheorem{example}{Example}

\newenvironment{exm}
{\pushQED{\qed}\begin{example}}
{\popQED\end{example}}

\newtheorem*{continueexample}{Example \continuation}
\newcommand{\continuation}{??}
\newenvironment{excont}[2]
 {\pushQED{\qed}\renewcommand{\continuation}{\ref{#1}} \continueexample[{#2} continued]}
 {\popQED\endcontinueexample}

\newcommand{\ud}{\textnormal{d}}

\newcommand{\prob}[1]{P \left\{ #1 \right\}}
\newcommand{\ex}[1]{\textnormal{E} \left[ #1 \right]}
\newcommand{\cw}{\leadsto}
\newcommand{\cp}{\stackrel{p}{\rightarrow}}
\newcommand{\R}{\mathbb{R}}
\DeclareMathOperator*{\argmax}{argmax}
\DeclareMathOperator*{\argmin}{argmin}
\DeclareMathOperator*{\sgn}{sgn}
\newcommand{\gr}{\textnormal{gr}}

\newcommand{\calA}{\mathcal{A}}
\newcommand{\calB}{\mathcal{B}}
\newcommand{\calC}{\mathcal{C}}

\newcommand{\calG}{\mathcal{G}}
\newcommand{\calP}{\mathcal{P}}
\newcommand{\calT}{\mathcal{T}}

\newcommand{\bbD}{\mathbb{D}}
\newcommand{\bbE}{\mathbb{E}}
\newcommand{\bbF}{\mathbb{F}}
\newcommand{\bbG}{\mathbb{G}}
\newcommand{\bbL}{\mathbb{L}}
\newcommand{\bbU}{\mathbb{U}}

\title{Uniform inference for value functions}

\author{Sergio Firpo\footnote{Insper, Sao Paulo, Brazil.  \href{mailto:firpo@insper.edu.br}\texttt{firpo@insper.edu.br}} \and Antonio F.  Galvao\footnote{Department of Economics, Michigan State University, East Lansing, USA.  \href{mailto:agalvao@msu.edu}{\texttt{agalvao@msu.edu}}} \and Thomas Parker\footnote{Corresponding author.  Department of Economics, University of Waterloo, Waterloo, Canada.  \href{mailto:tmparker@uwaterloo.ca}{\texttt{tmparker@uwaterloo.ca}}}}

\begin{document}
\maketitle
\thispagestyle{empty}

\begin{abstract}
  \begin{spacing}{1}
    We propose a method to conduct uniform inference for the (optimal) value function, that is, the function that results from optimizing an objective function marginally over one of its arguments. Marginal optimization is not Hadamard differentiable (that is, compactly differentiable) as a map between the spaces of objective and value functions, which is problematic because standard inference methods for nonlinear maps usually rely on Hadamard differentiability.  However, we show that the map from objective function to an $L_p$ functional of a value function, for $1 \leq p \leq \infty$, are Hadamard directionally differentiable.  As a result, we establish consistency and weak convergence of nonparametric plug-in estimates of Cram\'er-von Mises and Kolmogorov-Smirnov test statistics applied to value functions.  For practical inference, we develop detailed resampling techniques that combine a bootstrap procedure with estimates of the directional derivatives.  In addition, we establish local size control of tests which use the resampling procedure.  Monte Carlo simulations assess the finite-sample properties of the proposed methods and show accurate empirical size and nontrivial power of the procedures. Finally, we apply our methods to the evaluation of a job training program using bounds for the distribution function of treatment effects.
  \end{spacing}
\end{abstract}

\textbf{Keywords:} uniform inference, value function, bootstrap, delta method,
directional differentiability

\textbf{JEL codes:} C12, C14, C21

\newpage
\section{Introduction}
\setcounter{page}{1}

The (optimal) value function, that is, a function defined by optimizing an objective function over one of multiple arguments, is used widely in economics, finance, statistics and other fields.  The statistical behavior of estimated value functions appears, for instance, in the study of distribution and quantile functions of treatment effects under partial identification (see, e.g., \citet{FirpoRidder08, FirpoRidder19} and \citet{FanPark10,FanPark12}).  More recently, investigating counterfactual sensitivity, \citet{ChristensenConnault19} construct interval-identified counterfactual outcomes as the distribution of unobservables varies over a nonparametric neighborhood of an assumed model specification.  Statistical inference for value functions can be analyzed at a point as in \citet{Roemisch06} or \citet{CarcamoCuevasRodriguez19} and inference can be carried out using the techniques of \citet{HongLi18} or \citet{FangSantos19}.

This paper extends these pointwise inference methods to uniform inference for value functions that are estimated with nonparametric plug-in estimators.  Uniform inference has an important role for the analysis of statistical models.\footnote{For example, \cite{AngristChernozhukovVal06} discuss the importance of uniform inference \textit{vis-\`a-vis} pointwise inference in a quantile regression framework.  Uniform testing methods allow for multiple comparisons across the entire distribution without compromising confidence levels, and uniform confidence bands differ from corresponding pointwise confidence intervals because they guarantee coverage over the family of confidence intervals simultaneously.} However, uniform inference for the value function can be quite challenging.  In order to use standard statistical inference procedures, one must show that the map from objective function to value function is Hadamard (i.e., compactly) differentiable as a map between two functional spaces, which may require restrictive regularity conditions.  As a result, the delta method, as the conventional tool used to analyze the distribution of nonlinear transformations of sample data, may not apply to marginal optimization as a map between function spaces due to a lack of differentiability.  We propose a solution to this problem that allows for the use of the delta method to conduct uniform inference for value functions.

Because in general the map of objective function to value function is not differentiable, our solution is to directly analyze statistics that are applied to the value function.  This is in place of a more conventional analysis that would first establish that the value function is well-behaved and as a second step use the continuous mapping theorem to determine the distribution of test statistics.  In settings involving a chain of maps, this can be seen as choosing how to define the ``links'' in the chain to which the chain rule applies.  The results are presented for $L_p$ functionals for $1 \leq p \leq \infty$, that are applied to a value function, which should cover many cases of interest.  In particular, this family of functionals includes Kolmogorov-Smirnov and Cram\'er-von Mises tests.

By considering $L_p$ functionals of the value function, we may bypass the most serious impediment to uniform inference.  However, these $L_p$ functionals are only Hadamard directionally differentiable.  A directional generalization of Hadamard (or compact) differentiability was developed for statistical applications in \citet{Shapiro90} and \citet{Duembgen93}.  There is a more recent growing literature of applications using Hadamard directional differentiability.  See, for example, \citet{SommerfeldMunk17}, \citet{ChetverikovSantosShaikh18}, \citet{ChoWhite18}, \citet{HongLi18}, \citet{FangSantos19}, \citet{ChristensenConnault19}, \cite{MastenPoirier20}, and \citet{CattaneoJanssonNagasawa20}.  We separately analyze $L_p$ functionals for $1 \leq p < \infty$ and the supremum norm, as well as one-sided variants used for tests that have a sense of direction.  As an intermediate step towards showing the differentiability of supremum-norm statistics, we show directional differentiability of a minimax map applied to functions that are not necessarily convex-concave.

Directional differentiability of $L_p$ functionals provides the minimal amount of regularity needed for feasible statistical inference.  We use directional differentiability to find the asymptotic distributions of test statistics estimated from sample data.  The distributions generally depend on features of the objective function, and we show how to combine estimation with resampling to estimate the limiting distributions, as was proposed by \citet{FangSantos19}.  This technique can be tailored to impose a null hypothesis and is simple to implement in practice. We provide three examples throughout the paper to illustrate the proposed methods --- dependency bounds, stochastic dominance, and cost and profit functions.  Moreover, we establish local size control of the resampling procedure in general and uniform size control for some functionals.

As a practical illustration of the developed framework, we consider bounds for the distribution function of a treatment effect. We use this model in both the Monte Carlo simulations and empirical application. Treatment effects models have provided a valuable analytic framework for the evaluation of causal effects in program evaluation studies. When the treatment effect is heterogeneous, its distribution function and related features, for example its quantile function, are often of interest. Nevertheless, it is well known (e.g.  \citet{HeckmanSmithClements97, AbbringHeckman07}) that features of the distribution of treatment effects beyond the average are not point identified unless one imposes significant theoretical restrictions on the joint distribution of potential outcomes. Therefore it has became common to use partial identification of the distribution of such features, and to derive bounds on their distribution without making any unwarranted assumptions about the copula that connects the marginal outcome distributions.\footnote{Bounds for the CDF and quantile functions at one point in the support of the treatment effects distribution were developed in \citet{Makarov82, Rueschendorf82, FrankNelsenSchweizer87, WilliamsonDowns90}.  These bounds are minimal in the sense that they rely only on the marginal distribution functions of the observed outcomes, not the joint distribution of (potential) outcomes.}

Monte Carlo simulations (reported in a supplementary appendix) assess the finite-sample properties of the proposed methods. The simulations suggest that Kolmogorov-Smirnov statistics used to construct uniform confidence bands have accurate empirical size and power against local alternatives.  A Cram\'er-von Mises type statistic is used to test a stochastic dominance hypothesis using bound functions, similar to the breakdown frontier tests of \citet{MastenPoirier20}.  This test has accurate size at the breakdown frontier and power against alternatives that violate the dominance relationship.  All results are improved when the sample size increases and are powerful using a modest number of bootstrap repetitions.

We illustrate the methods empirically using job training program data set first analyzed by \citet{LaLonde86} and subsequently by many others, including \citet{HeckmanHotz89}, \citet{DehejiaWahba99}, \citet{SmithTodd01, SmithTodd05}, \citet{Imbens03}, and \citet{Firpo07}, without making any assumptions on dependence between potential outcomes in the experiment.  This experimental data set has information from the National Supported Work Program used by \citet{DehejiaWahba99}. We document strong heterogeneity of the job training treatment across the distribution of earnings. The uniform confidence bands for the treatment effect distribution function are imprecisely estimated at some parts of the earnings distribution. These large uniform confidence bands may in part be attributed to the large number of zeros contained in the data set, but are also inherent to the fact that the distribution function is everywhere only partially identified.

The remainder of the paper is organized as follows. Section \ref{sec:model} defines the statistical model of interest for the value function. Section \ref{sec:differentiability} establishes directional differentiability for the functionals of interest. Inference procedures are established in Section \ref{sec:inference}.  An empirical application to job training is discussed in Section \ref{sec:application}.  Finally, Section \ref{sec:conclusion} concludes.  Monte Carlo simulation descriptions and results and all proofs are relegated to the supplementary appendix.

\subsection*{Notation} For any set $T \subseteq \R^d$ let $\ell^\infty(T)$ denote the space of bounded functions $f: T \rightarrow \R$ and let $\calC(T)$ denote the space of continuous functions $f: T \rightarrow \R$, both equipped with the uniform norm $\| f \|_\infty = \sup_{t \in T} |f(t)|$.  Given a sequence $\{f_n\}_n \subset \ell^\infty(T)$ and limiting random element $f$ we write $f_n \cw f$ to denote weak convergence in $(\ell^\infty(T), \| \cdot \|_\infty)$ in the sense of Hoffmann-J\o rgensen \citep{vanderVaartWellner96}.  Let $[x]_+ = \max\{x, 0\}$ and $[x]_- = \min\{x, 0\}$.  A set-valued map (or correspondence) $S$ that maps elements of $X$ to the collection of subsets of $Y$ is denoted $S: X \rightrightarrows Y$.  For set-valued map $S: X \rightrightarrows Y$, let $\gr S$ denote the graph of $S$ as a set in $X \times
Y$.

\section{The model}\label{sec:model}

Consider two sets $U \subseteq \R^{d_U}$ and $X \subseteq \R^{d_X}$, a set-valued map $A: X \rightrightarrows U$ that serves as a choice set, and an objective function $f \in \ell^\infty(\gr A)$.  The objective function and the value function are linked by a marginal optimization step that maps one functional space to another.  Let $\psi: \ell^\infty(\gr A) \rightarrow \ell^\infty(X)$ map the objective function to the value function obtained by marginally optimizing the objective $f$ with respect to $u \in A(x)$ for each $x \in X$.  Without loss of generality, we consider only marginal maximization with respect to $u$ for each $x$:
\begin{equation}
  \psi(f)(x) = \sup_{u \in A(x)} f(u, x). \label{m_def}
\end{equation}
The value function $\psi(f)$ is the object of interest and we would like to conduct uniform inference using a plug-in estimator $\psi(f_n)$.  First we introduce examples of the value function model represented in equation \eqref{m_def}.

\begin{exm}[Dependency bounds] \label{ex:bounds}
  Dependency bounds describe bounds on the distribution function of the sum of two random variables $X$ and $Y$, $F_{X+Y}$, using only the marginal distribution functions $F_X$ and $F_Y$ and no information from the joint distribution $F_{XY}$.  \citet{FrankNelsenSchweizer87} showed that bounds derived for $F_{X+Y}$ by \citet{Makarov82} are related to copula bounds and can be extended to other binary operations.  Upper and lower dependency bounds for the CDFs $F_{X + Y}$, $F_{X - Y}$, $F_{X \cdot Y}$ and $F_{X / Y}$ and their associated quantile functions are provided in \citet{WilliamsonDowns90}.

  \citet{WilliamsonDowns90} show that for any $x$, $L_{X+Y}(x) \leq F_{X+Y}(x) \leq U_{X+Y}(x)$, where
  \begin{align}
    L_{X+Y}(x) &= \left[ \sup_{u \in \R} (F_X(u) + F_Y(x - u) - 1) \right]_+ \label{fbound_lo} \\
    U_{X+Y}(x) &= 1 - \left[ \inf_{u \in \R} (F_X(u) + F_Y(x - u)) \right]_-. \label{fbound_hi}
  \end{align}
  Similarly, letting $f^{-1}$ denote the generalized inverse of an increasing function, bounds for the quantile function $F_{X+Y}^{-1}(\tau)$ may be considered for $0 < \tau < 1$.  The functions above may be inverted to find that for each quantile level $\tau$, $U_{X+Y}^{-1}(\tau) \leq F_{X+Y}^{-1}(\tau) \leq L_{X+Y}^{-1}(\tau)$, and the quantile function bounds are
  \begin{align}
    U_{X+Y}^{-1}(\tau) &= \sup_{u \in (0, \tau)} \left\{ F_X^{-1}(u) + F_Y^{-1}(\tau - u) \right\} \label{qbound_lo}\\
    L_{X+Y}^{-1}(\tau) &= \inf_{u \in (\tau, 1)} \left\{ F_X^{-1}(u) + F_Y^{-1}(\tau - u + 1) \right\}. \label{qbound_hi}
  \end{align}
  The functions $U_{X+Y}^{-1}$ and $L_{X+Y}^{-1}$ are value functions where marginal optimization takes place over a set-valued map that varies with $\tau$.

  All these related functions may be written as maps that depend on a value function in some part of a chain of maps.  For example, take $L_{X+Y}$.  For a pair of functions $f = (f_x, f_y) \in (\ell^\infty(\R))^2$, let $\Pi^+: (\ell^\infty(\R))^2 \rightarrow \ell^\infty(\R^2)$
  \begin{equation} \label{pi_def}
    \Pi^+(f)(s, t) = f_x(s) + f_y(t - s) - 1
  \end{equation}
  and let $\psi$ be as in~\eqref{m_def} with $A(x) = \R$ for all $x$.  Then $L_{X+Y} = \psi(\Pi^+(F) \vee 0)$.  Similarly, consider $U_{X+Y}^{-1}$.  For a pair of CDFs $f = (f_x, f_y)$ where $f_x, f_y: \R \rightarrow [0, 1]$, define $\tilde{\Pi}^+: (\ell^\infty(\R))^2 \rightarrow \ell^\infty([0, 1]^2)$ by\footnote{We have simplified the problem by excluding $\{0, 1\}$ from the definition of $\tilde{\psi}$ for ease of exposition~--- see p.115-117 of \citet{WilliamsonDowns90} for more complex cases which include these points.}
  \begin{equation} \label{pi_til_def}
    \tilde{\Pi}^+(f)(s, t) = f_x^{-1}(s) + f_y^{-1}(t - s)
  \end{equation}
  and let $\psi$ use $A(\tau) = (0, \tau)$ for each $\tau \in (0, 1)$.  Then $U_{X+Y}^{-1} = \psi(\tilde{\Pi}^+(F))$.  We will focus on constructing uniform confidence bands for $L_{X+Y}$ and $U_{X+Y}^{-1}$, since the calculations for the other functions are analogous.
\end{exm}

\begin{exm}[Stochastic dominance] \label{ex:stoc_dominance}
  Inference for features of $F_{X-Y}$ at a given $x$, when $X - Y$ may be considered a treatment effect distribution, has been considered in \citet{FanGuerreZhu17} and the references cited therein.  Consider testing for first order stochastic dominance without assuming point identification.  Suppose that $(X_A, X_B, X_0)$ have marginal CDFs $(G_A, G_B, G_0) \in (\ell^\infty(\R))^3$, and let $F_{\Delta_A}$ and $F_{\Delta_B}$ be the distribution functions of the treatment effects $\Delta_A = X_A - X_0$ and $\Delta_B = X_B - X_0$.  Without point identification we cannot test the hypothesis $H_0: F_{\Delta_A} \succeq_{FOSD} F_{\Delta_B}$, which is equivalent to the condition that $F_{\Delta_A}(x) - F_{\Delta_B}(x) \leq 0$ for all $x$.  However, the distributions of $\Delta_A$ and $\Delta_B$ may be bounded.  For example, using references from the previous example it can be shown that for each $x$,
  \begin{align*}
    F_{\Delta_A}(x) \geq L_A(x) &= \left[ \sup_{u \in \R} (G_A(u) - G_0(u - x)) \right]_+, \\
    F_{\Delta_B}(x) \leq U_B(x) &= 1 + \left[ \inf_{u \in \R} (G_B(u) - G_0(u - x)) \right]_-.
  \end{align*}
  Under the null hypothesis, $L_A \leq F_{\Delta_A} \leq F_{\Delta_B} \leq U_B$, so that $L_A(x) - U_B(x) \leq F_{\Delta_A} - F_{\Delta_B} \leq 0$ for all $x$, and the null hypothesis may be rejected whenever there exists an $x'$ such that $L_A(x') > U_B(x')$.  This holds regardless of the correlations between treatment outcomes. Let $\Pi^-: (\ell^\infty(\R))^2 \rightarrow \ell^\infty(\R^2)$ be, for $f = (f_0, f_1)$
  \begin{equation*}
    \Pi^-(f)(u, x) = f_1(u) - f_0(u - x)
  \end{equation*}
  and let $\psi$ be defined by~\eqref{m_def} with $A(x) = \R$ for all $x$.  The function used to indicate a violation of dominance maps three distribution functions into a function in $\ell^\infty(\R)$: for each $x \in \R$,
  \begin{align}
    L_A(x) - U_B(x) &= \sup_{u \in \R} (G_A(u) - G_0(u - x)) - 1 - \inf_{u \in \R} (G_B(u) - G_0(u - x)) \\
    {} &= \psi(\Pi^-(G_0, G_A))(x) - 1 + \psi(-\Pi^-(G_0, G_B))(x).
  \end{align}
  We would like to test the hypothesis that $F_A$ dominates $F_B$, without knowledge of the dependence between $X_0$, $X_A$ and $X_B$, by looking for $x$ where $L_A(x) - U_B(x) > 0$, a case in which a one-sided functional $f \mapsto \|[f]_+\|_p$ is preferred for testing.
\end{exm}

\begin{exm}[Cost and profit functions] \label{ex:legendre}
  Suppose a firm produces outputs $x \in \R^n_+$ for $n \geq 1$.  Given cost function $c: \R^n \rightarrow \R$, the firm's profit at given product prices $p \in \R^n$ is $\pi(p) = \sup_{x \in \R^n_+} (\langle p, x \rangle - c(x))$.  Suppose that we have an accurate model of costs and would like to conduct inference for the profit curve, for example by putting a confidence band around an estimated profit curve.

  The Legendre-Fenchel transform can be used to find the dual conjugate function of a convex function $f$ or the convex hull of a non-convex function $f$ \citep[Chapter 11]{RockafellarWets98}.  For convex $f: \R^n \rightarrow \R$, the conjugate $f^\ast: \R^n \rightarrow \R$ is \(f^\ast(x) = \sup_{u \in \R^n} \left( \langle u, x \rangle - f(u) \right)\), that is, the smallest function such that $f^\ast(x) + f(u) \geq \langle u, x \rangle$.  We can recast this as an optimization map by defining the Legendre transform as a map $\mathcal{L}: \ell^\infty(\R^n) \rightarrow \ell^\infty(\R^n)$ defined by
  \begin{equation*}
    \mathcal{L}(f)(x) = \sup_{u \in \R^n} \left( \langle u, x \rangle - f(u) \right).
  \end{equation*}
  Then the profit function can be seen as the conjugate dual function of the cost function: $\pi(p) = \mathcal{L}(c)(p) = \sup_{x \in \R^n} \left( \langle p, x \rangle - f(x) \right)$.  As discussed in the text, it is possible (i.e., when it is not possible to directly calculate $\pi$ from the cost function) that $\pi$ is not Hadamard differentiable.  Instead, consider the supremum of the absolute value of transformed functions $\check{\lambda}(f) = \|\mathcal{L}(f)\|_\infty$.  Since $\pi = \mathcal{L}(c)$, we can find a confidence band for $\pi$ using $\check{\lambda}(c)$ using a sample analog $\check{\lambda}(c_n)$.
\end{exm}

In Example~\ref{ex:bounds}, the bounds on the quantile function are an example where $A(x)$ varies with $x$.  It would be notationally simpler to consider optimization problems over the rectangle $\mathcal{X} \times \mathcal{U}$.  However, given an $x$-varying set of constraints we would need to use ``indicator functions'' as they are used in the optimization literature, e.g. we would change the optimization problem to $\tilde{f}(u, x) = f(u, x) + \delta_A(u, x)$, where $\delta_A(u, x) = 0$ for $u \in A(x)$ and (for maximization) $\delta_A(u, x) = -\infty$ for $u \notin A(x)$.  This takes us out of the realm of bounded functions and towards the consideration of epi- or hypographs, which may be more complicated.  For example, see \citet{BuecherSegersVolgushev14}.  The literature on weak convergence in the space of bounded functions is voluminous and more familiar to researchers in economics.

To analyze the asymptotic properties of the above examples with plug-in estimates, one would typically rely on the delta method because $\psi$ is a nonlinear map from objective function to value function.  In the next section we discuss the difficulties encountered with the delta method for value functions and a solution for the purposes of uniform inference.

\section{Differentiability properties} \label{sec:differentiability}

We use the delta method to establish uniform inference methods for value functions.  This depends on the notion of Hadamard (or compact) differentiability to linearize the map near the population distribution (see, for example, \citet[Section 3.9.3]{vanderVaartWellner96}), and has the advantage of dividing the analysis into a deterministic part and a statistical part.  We first consider Hadamard derivatives without regard to sample data, before explicitly considering the behavior of sample statistics in the next section.

\subsection{Directional differentiability}

The delta method between two metric spaces usually relies on (full) Hadamard differentiability \citep[Section 3.9]{vanderVaartWellner96}.  \citet{Shapiro90}, \citet{Duembgen93} and \citet{FangSantos19} discuss Hadamard directional differentiability and show that this weaker notion also allows for the application of the delta method.\footnote{The derivative $\phi_f'$ in the definition is also called a \emph{semiderivative} \citep[Definition 7.20]{RockafellarWets98}.}

\begin{defn}[Hadamard directional differentiability] \label{def:diff}
  Let $\bbD$ and $\bbE$ be Banach spaces and consider a map $\phi: \bbD_\phi \subseteq \bbD \rightarrow \bbE$.  $\phi$ is \emph{Hadamard directionally differentiable} at $f \in \bbD_\phi$ tangentially to a set $\bbD_0 \subseteq \bbD$ if there is a continuous map $\phi'_f: \bbD_0 \rightarrow \bbE$ such that
  \begin{equation*}
    \lim_{n \rightarrow \infty} \left\| \frac{\phi(f + t_n h_n) - \phi(f)}{t_n} - \phi_f'(h) \right\|_\bbE = 0
  \end{equation*}
  for all sequences $\{h_n\} \subset \bbD$ and $\{t_n\} \subset \R_+$ such that $h_n \rightarrow h \in \bbD_0$ and $t_n \rightarrow 0^+$ as $n \rightarrow \infty$ and $f + t_n h_n \in \bbD_\phi$ for all $n$.
\end{defn}

In case the derivative map $\phi'_f$ is linear and $t_n$ may approach $0$ from both negative and positive sides, the map is fully Hadamard differentiable.  The delta method and chain rule can be applied to maps that are either Hadamard differentiable or Hadamard directionally differentiable.

To discuss the differentiability of $\psi$ in equation \eqref{m_def}, for any $\epsilon \geq 0$ let $U_f: X \rightrightarrows U$ define the set-valued map of $\epsilon$-maximizers of $f(\cdot, x)$ in $u$ for each $x \in X$:
\begin{equation} \label{marginal_epsmax_def}
  U_f(x, \epsilon) = \left\{ u \in A(x): f(u, x) \geq \psi(f)(x) - \epsilon \right\}.
\end{equation}
Results on the directional derivatives of the optimization map $f \mapsto \sup_X f$ date to \citet{Danskin67}.  For example, \citet[Theorem 2.1]{CarcamoCuevasRodriguez19} show that, in our notation, for fixed $x \in X$, $\psi$ is directionally differentiable in $\ell^\infty(A(x) \times \{x\})$, and for directions $h(\cdot, x)$
\begin{equation} \label{psiprime_def}
  \psi'_f(h)(x) = \lim_{\epsilon \rightarrow 0^+} \sup_{u \in U_f(x, \epsilon)} h(u, x).
\end{equation}
Assuming the stronger conditions that $A$ is continuous and compact-valued and that $f$ is continuous on the graph of $A$ ($\gr A$), the maximum theorem implies that $U_f$ is non-empty, compact-valued and upper hemicontinuous \citep[Theorem 17.31]{AliprantisBorder06}, and \citet[Corollary 2.2]{CarcamoCuevasRodriguez19} show that tangentially to $\calC(U \times \{x\})$, the derivative of $\psi(f)(x)$ simplifies to $\psi'_f(h)(x) = \sup_{u \in U_f(x, 0)} h(u, x)$.  Further, when (for fixed $x$) $U_f(x, 0)$ is a singleton set, $\psi(f)(x)$ is fully differentiable, not just directionally so.

The pointwise directional differentiability of $\psi(f)(x)$ at each $x$ might lead one to suspect that $\psi$ is differentiable more generally as a map from $\ell^\infty(\gr A)$ to $\ell^\infty(X)$.  However, this is not true.  The following example illustrates a case where $\psi$ is not Hadamard directionally differentiable as a map from $\ell^\infty(\gr A)$ to $\ell^\infty(X)$, although $f$ and $h$ are both continuous functions.

\begin{exm}
  Let $U \times X = [0, 1] \times [-1, 1]$ and define $f: U \times X \rightarrow \R$ by
  \begin{equation*}
    f(u, x) = \begin{cases} 0 & x \in [-1, 0] \\ x(u^2 - u) & x \in (0, 1] \end{cases}.
  \end{equation*}
  Let $h(u, x) = -u^2 + u$.  Then $\psi(f)(x) = 0$ for all $x$ and for any $t > 0$,
  \begin{equation*}
    \psi(f + th)(x) = \begin{cases} t / 4 & x \in [-1, 0] \\ (t - x) / 4 & x \in (0, t] \\ 0 & x \in (t, 1] \end{cases}.
  \end{equation*}
  Therefore for each $x$,
  \begin{equation*}
    \lim_{t \rightarrow 0^+} \frac{\psi(f + th)(x) - \psi(f)(x)}{t} = \frac{1}{4} I(x \leq 0).
  \end{equation*}
  However, for any $t > 0$,
  \begin{equation*}
    \sup_{x \in [-1, 1]} \left| (\psi(f + th)(x) - \psi(f)(x)) / t - (1/4) I(x \leq 0) \right| = 1 / 4.
  \end{equation*}
  This implies that no derivative exists as an element of $\ell^\infty([-1, 1])$.  The functions $f$ and $h$ are well behaved, but because the candidate derivative $I(x \leq 0) / 4$ is not continuous in $x$, uniform convergence of the quotients to the candidate fails.
\end{exm}

\subsection{Directional differentiability of norms of value functions} \label{sec:stats}

The lack of uniformity of the convergence to $\psi_f'(h)(\cdot)$ appears to preclude the delta method for uniform inference, because the typical path of analysis would use a well-behaved limit of $r_n(\psi(f_n) - \psi(f))$ and the continuous mapping theorem to find the distribution of statistics applied to the limit.  However, real-valued statistics are used to find uniform inference results, so we examine maps that include not only the marginal optimization step but also a functional applied to the resulting value function.

Consider $L_p$ norms (for $1 \leq p \leq \infty$) applied to $\psi(f)$ or $[\psi(f)]_+$.  Letting $m$ denote the dominating measure, define $\lambda_j: \ell^\infty(\gr A) \rightarrow \R$ for $j = 1, \ldots, 4$ by
\begin{align}
  \begin{aligned}
    \lambda_1(f) &= \sup_{x \in X} \left| \sup_{u \in A(x)} f(u, x) \right|, &\lambda_2(f) &= \sup_{x \in X} \left[ \sup_{u \in A(x)} f(u, x) \right]_+, \\
    \lambda_3(f) &= \left( \int_X \left| \sup_{u \in A(x)} f(u, x) \right|^p \ud m(x) \right)^{1/p}, &\lambda_4(f) &= \left( \int_X \left[ \sup_{u \in A(x)} f(u, x) \right]_+^p \ud m(x) \right)^{1/p}. \label{lambdas_abstract}
  \end{aligned}
\end{align}

The maps $\lambda_2$ and $\lambda_4$ are included because one-sided comparisons may also be of interest, and the map $f \mapsto [f]_+$ is pointwise directionally differentiable but not differentiable as a map from $\ell^\infty(X)$ to $\ell^\infty(X)$.  We illustrate its use in a stochastic dominance example below.

We make the following definitions in order to construct well-defined derivatives.  First, let $\mu(f)$ be the supremum of $f$ over $\gr A$ and let $\sigma(f)$ be the maximin value of $f$ over $\gr A$:
\begin{align}
  \mu(f) &= \sup_{(u, x) \in \gr A} f(u, x) \label{mu_def} \\
  \sigma(f) &= \sup_{x \in X} \inf_{u \in A(x)} f(u, x). \label{sigma_def}
\end{align}
Define $A_f^\mu \subseteq \gr A$ for any $\epsilon > 0$ to be the set of $\epsilon$-maximizers of $f$ in $(u, x)$ over $\gr A$.  Next, let $X_f^\sigma \subseteq X$ collect the $\epsilon$-maximinimizers of $f$ in $x$ over $X$.  That is, define
\begin{align}
  A_f^\mu(\epsilon) &= \left\{ (u, x) \in \gr A : f(u, x) \geq \mu(f) - \epsilon \right\} \label{Afmu_def} \\
  X_f^\sigma(\epsilon) &= \left\{ x \in X : \inf_{u \in A(x)} f(u, x) \geq \sigma(f) - \epsilon \right\}. \label{Xfsigma_def}
\end{align}
Also, note that $U_{(-f)}(x, \epsilon)$ is the set of $\epsilon$-minimizers of $f(u, x)$ in $u$ for given $x$.  

\citet{CarcamoCuevasRodriguez19} show that the map~\eqref{mu_def} has directional derivative
\begin{equation} \label{muprime_def}
  \mu'_f(h) = \lim_{\epsilon \rightarrow 0^+} \sup_{x \in A_f^\mu(\epsilon)} h(u, x)
\end{equation}
and Lemma~\ref{lem:saddle} in the \hyperref[appn]{Appendix} shows that~\eqref{sigma_def} has directional derivative
\begin{equation} \label{sigmaprime_def}
  \sigma_f'(h) = \lim_{\delta \rightarrow 0^+} \sup_{x \in X_f^\sigma(\delta)} \lim_{\epsilon \rightarrow 0^+} \inf_{u \in U_{(-f)}(x, \epsilon)} h(u, x).
\end{equation}
To analyze derivatives of the one-sided $\lambda_2(f)$ and $\lambda_4(f)$, define the contact set
\begin{equation} \label{contact_def}
  X_0 = \left\{ x \in X : \psi(f)(x) = 0 \right\}.
\end{equation}

The following two theorems describe the form of the four directional derivatives generally and under the condition that the value function is zero everywhere or on some subset of $X$.
\begin{thm} \label{thm:supnorm_stats}
  Suppose $f, h \in \ell^\infty(\gr A)$, where $A: X \rightrightarrows U$ is non-empty-valued for all $x \in X$.  Define $\mu$ by~\eqref{mu_def} and $\sigma$ by~\eqref{sigma_def}, and their derivatives $\mu_f'$ and $\sigma_f'$ as in~\eqref{muprime_def} and~\eqref{sigmaprime_def}.  Then:
  \begin{enumerate}
    \item $\lambda_1(f)$ is Hadamard directionally differentiable and
      \begin{equation*}
        \lambda'_{1f}(h) = \begin{cases} \mu_f'(h), & \mu(f) > \sigma(-f) \\ \max \left\{ \mu_f'(h), \sigma_{(-f)}'(-h) \right\}, & \mu(f) = \sigma(-f) \\ \sigma_{(-f)}'(-h), & \mu(f) < \sigma(-f) \end{cases}.
      \end{equation*}
    \item If $\|\psi(f)\|_\infty = 0$, $\lambda_1(f)$ is Hadamard directionally differentiable and
      \begin{align*}
        \lambda'_{1f}(h) &= \max \left\{ \lim_{\epsilon \rightarrow 0^+} \sup_{(u, x) \in U_f(X, \epsilon)} h(u, x), \lim_{\epsilon \rightarrow 0^+} \sup_{x \in X} \inf_{u \in U_f(x, \epsilon)} (-h(u, x)) \right\} \\
        {} &= \lim_{\epsilon \rightarrow 0^+} \sup_{x \in X} \left| \sup_{u \in U_f(x, \epsilon)} h(u, x) \right|.
      \end{align*}
    \item $\lambda_2(f)$ is Hadamard directionally differentiable and
      \begin{equation*}
        \lambda'_{2f}(h) = \begin{cases} \mu_f'(h), & \mu(f) > 0 \\ \left[ \mu_f'(h) \right]_+, & \mu(f) = 0 \\ 0 & \mu(f) < 0 \end{cases}.
      \end{equation*}
    \item If $\|[\psi(f)]_+\|_\infty = 0$, $\lambda_2(f)$ is Hadamard directionally differentiable and
      \begin{equation*}
        \lambda'_{2f}(h) = \begin{cases} \lim_{\epsilon \rightarrow 0^+} \sup_{(u, x) \in U_f(X_0, \epsilon)} \left[ h(u, x) \right]_+, & \mu(f) = 0 \\ 0, & \mu(f) < 0 \end{cases}.
      \end{equation*}
  \end{enumerate}
\end{thm}

The first part of Theorem~\ref{thm:supnorm_stats} essentially takes advantage of the identity $|x| = \max\{x, -x\}$, and breaks the evaluation of the supremum of the absolute value of the value function into two parts that are well-defined, a maximization and a minimization problem.  The second part of the theorem is simpler because the $x \mapsto [x]_+$ map simplifies the evaluation of the supremum.

The next result shows that the $L_p$ functionals for $1 \leq p < \infty$ defined in~\eqref{lambdas_abstract}, are Hadamard directionally differentiable.  The form of the derivatives for these maps also change depending on whether $\lambda_j(f) = 0$ or $\lambda_j(f) \neq 0$.

\begin{thm} \label{thm:lpnorm_stats}
  Suppose that $f, h \in \ell^\infty(\gr A)$ where $A: X \rightrightarrows U$ is non-empty-valued for all $x \in X$ and $m(X) < \infty$.  Recall $\psi'_f$ defined in~\eqref{psiprime_def}.  Then:
  \begin{enumerate}
    \item If $\|\psi(f)\|_p \neq 0$, then $\lambda_3$ is Hadamard directionally differentiable and
      \begin{equation*}
        \lambda_{3f}'(h) = \| \psi(f) \|_p^{1-p} \int_X \sgn(\psi(f)(x)) |\psi(f)(x)|^{p-1} \psi_f'(h)(x) \ud m(x).
      \end{equation*}
    \item If $\|\psi(f)\|_p = 0$, then $\lambda_3$ is Hadamard directionally differentiable and
      \begin{equation*}
        \lambda_{3f}'(h) = \left( \int_X \left| \psi_f'(h)(x) \right|^p \ud m(x) \right)^{1/p}.
      \end{equation*}
    \item If $\|[\psi(f)]_+\|_p \neq 0$, then $\lambda_4$ is Hadamard directionally differentiable and
      \begin{equation*}
        \lambda_{4f}'(h) = \| [\psi(f)]_+ \|_p^{1-p} \int_X [\psi(f)(x)]_+^{p-1} \psi_f'(h)(x) \ud m(x).
      \end{equation*}
    \item If $\|[\psi(f)]_+\|_p = 0$, then $\lambda_4$ is Hadamard directionally differentiable and
      \begin{equation*}
        \lambda_{4f}'(h) = \left( \int_{X_0} \left[ \psi_f'(h)(x) \right]_+^p \ud m(x) \right)^{1/p}.
      \end{equation*}
  \end{enumerate}
\end{thm}

Theorem~\ref{thm:lpnorm_stats} is related to the results of \citet{ChenFang19}, who dealt with a squared $L_2$ statistic.  It is interesting to note here that using an $L_2$ statistic instead of its square results in first-order (directional) differentiability of the map, unlike the squared $L_2$ norm, which has a degenerate first-order derivative but nondegenerate second-order derivative.

We conclude this section by revisiting Examples \ref{ex:bounds}, \ref{ex:stoc_dominance}, and \ref{ex:legendre} and describing the computation of the directional derivatives in the corresponding cases.

\begin{excont}{ex:bounds}{Dependency bounds}
  To construct a uniform confidence band for $L_{X+Y}$ we can invert a test for the hypothesis that $L_{X+Y} = L_0$.  Recalling the map $\Pi^+$ defined in~\eqref{pi_def}, consider the statistic $\lambda: (\ell^\infty(\R))^2 \rightarrow \R$ defined by
  \begin{equation} \label{dep_stat_def}
    \lambda(F) = \sup_{x \in \R} |\max\{ \sup_{u \in \R} \Pi^+(F)(u, x), 0\} - L_0(x)|.
  \end{equation}
  This is challenging because of the pointwise maximum inside the absolute value.  Lemma~\ref{lem:other_sup} in the appendix, which uses calculations similar to those in the proof of Theorem~\ref{thm:supnorm_stats}, shows that, letting
  \begin{align*}
    X_0 &= \{x \in \R : \sup_{u \in \R} (F_X(u) + F_Y(x - u)) = 1\} \\
    X_+ &= \{x \in \R : \sup_{u \in \R} (F_X(u) + F_Y(x - u)) > 1\} \\
    U_F(x, \epsilon) &= \{u \in \R : F_X(u) + F_Y(x - u) \geq \sup_{u \in \R} (F_X(u) + F_Y(x - u)) - \epsilon\},
  \end{align*}
  $\lambda$ is Hadamard directionally differentiable and its derivative is, for direction $h = (h_X, h_Y)$,
  \begin{multline} \label{cband_derivative}
    \lambda_F'(h) = \max \Bigg\{ \lim_{\epsilon \rightarrow 0^+} \sup_{x \in X_0} \sup_{u \in U_F(x, \epsilon)} [h_X(u) + h_Y(x - u)]_+, \\
    \lim_{\epsilon \rightarrow 0^+} \sup_{x \in X_+} \left| \sup_{u \in U_F(x, \epsilon)} (h_X(u) - h_Y(x - u)) \right| \Bigg\}.
  \end{multline}

  Similarly, we may work with the bound $U_{X+Y}^{-1}$ on the quantile function $F_{X+Y}^{-1}$.  When a CDF $F$ is continuously differentiable on a set $\calT = [F^{-1}(p) - \epsilon, F^{-1}(q) + \epsilon]$ for $0<p<q<1$ and $\epsilon > 0$, it is fully Hadamard differentiable in $\ell^\infty(\calT)$ with derivative in direction $h$ at any $\tau \in \calT$ given by $-h(F^{-1}(\tau)) / f(F^{-1}(\tau)$, where $f$ is the density associated with $F$ \citep[Lemma 3.9.23]{vanderVaartWellner96}.  Recall the definition of $\tilde{\Pi}^+$ in~\eqref{pi_til_def} and consider, for $F = (F_X, F_Y)$,
  \begin{equation}
    \tilde{\lambda}(F) = \sup_{\tau \in \calT} |\sup_{u \in (0, \tau)} \tilde{\Pi}^+(F)(u, x) - U_0^{-1}(x)|.
  \end{equation}
  Under differentiability of the CDFs, we have
  \begin{equation*}
    (\tilde{\Pi}^+)'_F(h)(s, t) = - \frac{h_X(F_X^{-1}(s))}{f_X(F_X^{-1}(s))} - \frac{h_Y(F_Y^{-1}(t-s))}{f_Y(F_Y^{-1}(t-s))}.
  \end{equation*}
  Then this fact, Theorem~\ref{thm:supnorm_stats} and the chain rule imply that, letting $\tilde{U}(\tau, \epsilon) = \{u \in (0, \tau) : F_X^{-1}(u) + F_Y^{-1}(\tau - u) \geq \sup_{u \in (0, \tau)} (F_X^{-1}(u) - F_Y^{-1}(\tau - u)) - \epsilon \}$,
  \begin{equation*}
    \tilde{\lambda}'_F(h) = \lim_{\epsilon \rightarrow 0} \sup_{\tau \in \calT} \left| \sup_{u \in \tilde{U}(\tau, \epsilon)} \left( - \frac{h_X(F_X^{-1}(u))}{f_X(F_X^{-1}(u))} - \frac{h_Y(F_Y^{-1}(\tau-u))}{f_Y(F_Y^{-1}(\tau-u))} \right) \right|.
  \end{equation*}
\end{excont}

\begin{excont}{ex:stoc_dominance}{Stochastic dominance}
  Similarly to the previous example, we may test the necessary condition for stochastic dominance using the test statistic
  \begin{equation}
    \Lambda = \left( \int [ (L_{A}(x) - U_{B}(x))]_+^2 \ud m(x) \right)^{1/2}.
  \end{equation}
  This test statistic is zero when $L_A \leq U_B$ and greater than zero otherwise.  The map $\Pi^-$ is fully differentiable with derivative $(\Pi)_f'(h) = h_1(u) - h_0(u - x)$.  Then $\Lambda$ is directionally differentiable and, extending Theorem~\ref{thm:lpnorm_stats} to two functions inside the integral,
  \begin{multline} \label{biglambdaprime_def}
    \Lambda'_G(h) = \bigg( \int_{L_A = U_B} \Big[ \lim_{\epsilon \rightarrow 0^+} \sup_{u \in U_{\Pi^-(G_0, G_A)}(\epsilon, x)} \Pi^-(h_0, h_A) (u, x) \\
    - \lim_{\epsilon \rightarrow 0^+} \inf_{u \in U_{-\Pi^-(G_0, G_B)}(\epsilon, x)} \Pi^-(h_0, h_B)(u, x) \Big]_+^2 \ud m(x) \bigg)^{1/2}.
  \end{multline}
\end{excont}


\begin{excont}{ex:legendre}{Cost and profit functions}
  This example shows how some of the technical difficulties encountered in the previous two examples may be alleviated with assumptions on the economic model.  Recall that a uniform confidence band for $\pi = \mathcal{L}(c)$ depends on the derivative of $\check{\lambda}(f) = \sup_x |\mathcal{L}(f)(x)|$.  Suppose that $(u, x) \in \mathcal{U} \times \mathcal{X}$, a compact set.  More generally, the assumptions that firms will not operate with negative profits and that the cost function is 1-coercive (that is, that $\lim_{\|x\| \rightarrow \infty} c(x) / \|x\| = +\infty$) would ensure that $\pi$ is continuous.  Given the existence of optimizers and a maximum, write
  \begin{align*}
    \check{\lambda}(c) &= \max_p \left| \max_x (p \cdot x - c(x)) \right| \\
    {} &= \max \left\{ \max_{p,x} (p \cdot x - c(x)), \max_p \min_x (c(x) - p \cdot x) \right\}.
  \end{align*}
  Since $c$ is estimated and $c$ only depends on $x$, it is sufficient to perturb the above statistic in a direction that does not depend on $p$, i.e., $h(x)$, to calculate derivatives.

  Define the derivative of $\check{\mu}(c) = \max_{p,x} (p \cdot x - c(x))$ by
  \begin{equation*}
    \check{\mu}'_c(h) = \max_{\argmax (p \cdot x - c(x))} (-h(x)),
  \end{equation*}
  and the derivative of $\check{\sigma}(c) = \max_p \min_x (c(x) - p \cdot x)$ by
  \begin{equation*}
    \check{\sigma}'_c(h) = \max_{\argmax_p (\argmin_x (c(x) - p \cdot x))} \min_{\argmin_x (c(x) - p \cdot x)} h(x).
  \end{equation*}
  With strict convexity of $c$ the argmin expressions in $\check{\mu}'$ and $\check{\sigma}'$ would be single-valued, and these would be simpler.  The directional derivative of $\lambda(c)$ is
  \begin{equation} \label{leg_derivative}
    \lambda_c'(h) = \begin{cases} \check{\mu}'_c(h), & \check{\mu}(c) > \check{\sigma}(c), \\ \check{\mu}'_c(h) \vee \check{\sigma}'_c(h), & \check{\mu}(c) = \check{\sigma}(c), \\ \check{\sigma}'_c(h) & \check{\mu}(c) < \check{\sigma}(c) \end{cases}.
  \end{equation}
\end{excont}

\section{Inference}\label{sec:inference}

In this section we derive asymptotic distributions for uniform statistics applied to value functions, and propose bootstrap estimators of their distributions that can be used for practical inference.

\subsection{Asymptotic distributions}
The derivatives developed in the previous subsections can be used in a functional delta method.  We make a high-level assumption on the way that observations are related to estimated functions, but conditions under which such convergence holds are well-known, and examples will be shown below.

\begin{enumerate}[label=\textbf{A\arabic*.}, ref=\textbf{A\arabic*}]
  \item \label{A:estimator} Assume that for each $n$ there is a random sample $\{Z_i\}_{i=1}^n$  with $P$ denoting their joint probability distribution.  Assume there is a map $\{Z_i\}_{i=1}^n \mapsto f_n$ where $f_n \in \ell^\infty(U \times X)$.  In case $\lambda_3$ or $\lambda_4$ is used, also assume $m(X) < \infty$.  Furthermore, for some sequence $r_n \rightarrow \infty$, $r_n (f_n - f) \cw \calG_P$, where $\calG_P$ is a tight random element of $\ell^\infty(\gr A)$.
\end{enumerate}

In case an $L_p$ statistic is used (with $p < \infty$), we restrict the space of allowable functions.  In order to ensure that $L_p$ statistics are well-defined we make the assumption that the measure of $X$ is finite.  While stronger than the assumption used for the supremum statistic, it is sufficient to ensure that the $L_p$ statistics of the value function are finite.  This assumption also has the advantage of providing an easy-to-verify condition based on the objective function $f$, rather than the more direct, but less obvious condition $\psi(f) \in L_p(m)$.  Lifting this restriction would require some other restriction on the objective function $f$ that is sufficient to ensure the value function is $p$-integrable.\footnote{We attempted to show $p$-integrability by assuming that $f$ is bounded and integrable (note that this would imply $f$ is $p$-integrable in $\gr A$), but were unable to show that this implies that the value function is integrable (in $X$).  Note that if one were able to show that $f$ is bounded and integrable, then the elegant results in \citet[Section 3]{Kaji19} would apply for the purposes of verifying Assumption~\ref{A:estimator}.}

Given assumption~\ref{A:estimator} and the derivatives of the last section, the asymptotic distribution of test statistics applied to the value function is straightforward.
\begin{thm} \label{thm:asymptotic}
  Under Assumption~\ref{A:estimator},
  \begin{equation*}
    r_n \left(\lambda_j(f_n) - \lambda_j(f) \right) \cw \lambda_{jf}'(\calG_P),
  \end{equation*}
  where for $j \in \{1, \ldots 4\}$, $\lambda_j$ are defined in~\eqref{lambdas_abstract} and $\lambda_{jf}'$ are defined in Theorems~\ref{thm:supnorm_stats} and~\ref{thm:lpnorm_stats}.
\end{thm}
This theorem is abstract and hides the fact that the limiting distributions may depend on features of $\lambda_j$ and $f$.  Therefore it is only indirectly useful for inference.  A resampling scheme is the subject of the next section.

The delta method as used in Theorem~\ref{thm:asymptotic} is tailored to arguments that have weak limits in the space $\ell^\infty(\gr A)$, which is an intrinsic feature of the delta method.  That implies some limitations on the type of functions that may be considered.  For example, functions estimated using kernel methods do not converge to a weak limit as assumed in Assumption~\ref{A:estimator}.  One can in general only show the stochastic order of the sequence $r_n(f_n - f)$, and methods for uniform inference must shown using something other than the delta method.  We leave uniform inference for value functions based on less regular estimators to future research.

Our bootstrap technique ahead is designed for resampling under the null hypothesis that $\lambda_j(f) = 0$. We assume that when $j = 1$ or $j = 3$ (that is, conventional two-sided statistics are used), the statistics are used to test the null and alternative hypotheses
\begin{align}
  H_0&: \psi(f)(x) = 0 \quad \text{for all } x \in X \label{eq_null} \\
  H_1&: \psi(f)(x) \neq 0 \quad \text{for some } x \in X.
\end{align}
Meanwhile, when $j = 2$ or $j = 4$ (the one-sided cases) we assume the hypotheses are
\begin{align}
  H_0&: \psi(f)(x) \leq 0 \quad \text{for all } x \in X \label{ineq_null} \\
  H_1&: \psi(f)(x) > 0 \quad \text{for some } x \in X.
\end{align}

The probability measure $P$ is assumed to belong to a collection $\calP$ which describes the set of probability measures allowed by the model.  A few subcollections of $\calP$ serve to organize the asymptotic results below.  When $P$ is such that $\psi(f) \equiv 0$, we label $P \in \calP_{00}^E$.  For functional inequalities, the behavior of test statistics under the null is more complicated.  If $P$ is such that $\psi(f) \leq 0$ everywhere, then label $P \in \calP_0^I$.  When $P \in \calP_0^I$ makes $\psi(f)(x) = 0$ for at least one $x \in X$, then we label $P \in \calP_{00}^I$.

The following corollary combines the derivatives from Theorem~\ref{thm:supnorm_stats} and Theorem~\ref{thm:lpnorm_stats} with the result of Theorem~\ref{thm:asymptotic} for these distribution classes.  Recall that the set $X_0$, defined in~\eqref{contact_def}, represents the subset of $X$ where $\psi(f)$ is zero.

\begin{cor} \label{cor:null_asymptotics}
  Under Assumption~\ref{A:estimator}, for $j \in \{1, \ldots 4\}$,
  \begin{enumerate}
    \item For each $P \in \calP_{00}^E$, $r_n (\lambda_1(f_n) - \lambda_1(f)) \cw \lim_{\epsilon \rightarrow 0^+} \sup_{x \in X} \left| \sup_{u \in U_f(x, \epsilon)} \calG_P(u, x) \right|$.
    \item
      \begin{enumerate}
        \item For each $P \in \calP_{00}^I$, $r_n(\lambda_2(f_n) - \lambda_2(f)) \cw \lim_{\epsilon \rightarrow 0^+} \sup_{(u, x) \in U_f(X_0, \epsilon)} \left[ \calG_P(u, x) \right]_+.$
        \item For each $P \in \calP_0^I \backslash \calP_{00}^I$, $r_n(\lambda_2(f_n) - \lambda_2(f)) \cp 0$.
      \end{enumerate}
    \item For each $P \in \calP_{00}^E$, $r_n(\lambda_3(f_n) - \lambda_3(f)) \cw \left( \int_X \left| \psi_f'(\calG_P)(x) \right|^p \ud m(x) \right)^{1/p}$.
    \item
      \begin{enumerate}
        \item For each $P \in \calP_{00}^I$, $r_n(\lambda_4(f_n) - \lambda_4(f)) \cw \left( \int_{X_0} \left[ \psi_f'(\calG_P)(x) \right]_+^p \ud m(x) \right)^{1/p}$.
        \item For each $P \in \calP_0^I \backslash \calP_{00}^I$, $r_n(\lambda_4(f_n) - \lambda_4(f)) \cp 0$.
      \end{enumerate}
  \end{enumerate}
\end{cor}

The distributions described in the above corollary under the assumption that $P \in \calP_{00}^E$ or $\calP_{00}^I$ are those that we emulate using resampling methods in the next section.

\subsection{Resampling}

In the previous section we established asymptotic distributions for the test statistics of interest.  However, for practical inference we turn to resampling to estimate their distributions.  This section suggests the use of a resampling strategy that was proposed by \citet{FangSantos19}, and it combines a standard bootstrap procedure with estimates of the directional derivatives $\lambda_{jf}'$.  Hadamard directional differentiability also implies that the numerical approximation methods of \citet{HongLi18} may be used to estimate these directional derivatives, an approach that we address in simulations (in the supplementary appendix).

The resampling scheme described below is designed to reproduce the null distribution under the assumption that the statistic is equal to zero, or in other terms, that $P \in \calP_{00}^E$ or $P \in \calP_{00}^I$.  This is achieved by restricting the form of the estimates $\hat{\lambda}_{jn}'$.  The bootstrap routine is consistent under more general conditions, as in Theorem 3.1 of \citet{FangSantos19}.  However, we discuss behavior of bootstrap-based tests under the assumption that one of the null conditions described in the previous section holds.

All the derivative formulas in Corollary~\ref{cor:null_asymptotics} require some form of an estimate of the near-maximizers of $f$ in $u$, that is, of the set $U_f(x, \epsilon)$ defined in~\eqref{marginal_epsmax_def} for various $x$.  The set estimators we use are similar to those used in \citet{LintonSongWhang10}, \citet{ChernozhukovLeeRosen13} and \citet{LeeSongWhang18} and depend on slowly decreasing sequences of constants to estimate the relevant sets in the derivative formulas.  For a sequence $a_n \rightarrow 0^+$ we estimate $U_f(x, \epsilon)$ with the plug-in estimator $U_{f_n}(x, a_n)$.  The sequence $a_n$ should decrease more slowly than the rate at which $f_n$ converges uniformly to $f$, an assumption which will be formalized below.

The one-sided estimates $\hat{\lambda}_{2n}'$ and $\hat{\lambda}_{4n}'$ use a second estimate.  For another sequence $b_n \rightarrow 0^+$, estimate the contact set $X_0$ defined in~\eqref{contact_def} with
\begin{equation*}
  \hat{X}_0 = \begin{cases} \{x \in X : |\psi(f_n)(x)| \leq b_n\} & \text{if } \{x \in X : |\psi(f_n)(x)| \leq b_n\} \neq \varnothing \\ X & \text{if } \{x \in X : |\psi(f_n)(x)| \leq b_n\} = \varnothing \end{cases}.
\end{equation*}
This is a method used by~\citet{LintonSongWhang10} of achieving uniformity in the convergence of estimators that depend on contact set estimation under the null hypothesis $P \in \calP_{00}^I$.

Define estimators of $\lambda_{1f}'$ and $\lambda_{2f}'$ by
\begin{align}
  \hat{\lambda}_{1n}'(h) &= \max \left\{ \sup_{(u, x) \in U_{f_n}(X, a_n)} h(u, x), \sup_{x \in X} \inf_{u \in U_{f_n}(x, a_n)} (-h(u, x)) \right\} \label{est1_def} \\
  {} &= \sup_{x \in X} \left| \sup_{u \in U_{f_n}(x, a_n)} h(u, x) \right| \notag
\end{align}
and
\begin{equation}
  \hat{\lambda}_{2n}'(h) = \sup_{(u, x) \in U_{f_n}(\hat{X}_0, a_n)} \left[ h(u, x) \right]_+.
\end{equation}

These formulas impose the condition that $P \in \calP_{00}^E$ or $\calP_{00}^I$ on the derivative estimates, following the forms shown in Corollary~\ref{cor:null_asymptotics}.  For more general estimation we would need to devise more complex estimators to deal with the variety of limits discussed in Theorems~\ref{thm:supnorm_stats} and~\ref{thm:lpnorm_stats}, and also the technique of \citet{HongLi18} could be used in that case.  However, inference is usually of primary interest, and generating a bootstrap distribution that respects the null hypothesis can improve the performance of bootstrap inference procedures.  When $P \in \calP_{00}^E$ the imposition amounts to the assumption that $\psi(f)(x) = 0$ for each $x$, so that the set of population $\epsilon$-maximizers in $\gr A$ is $U_f(X, \epsilon)$, and the estimator $\hat{\lambda}_{1n}'$ emulates that.  The condition $P \in \calP_{00}^I$ implies that $X_0$ is not empty, and $\hat{X}_0$ is meant to estimate this contact set, while using all of $X$ when the distribution is not as is hypothesized.

The estimators $\hat{\lambda}_{3n}'$ and $\hat{\lambda}_{4n}'$ are defined similarly to the supremum-norm estimators: let
\begin{equation}
  \hat{\lambda}_{3n}'(h) = \left( \int_X \left| \sup_{u \in U_{f_n}(x, a_n)} h(u, x) \right|^p \ud m(x) \right)^{1/p}
\end{equation}
and
\begin{equation} \label{est4_def}
  \hat{\lambda}_{4n}'(h) = \left( \int_{\hat{X}_0} \sup_{u \in U_{f_n}(x, a_n)} \left[ h(u, x) \right]_+^p \ud m(x) \right)^{1/p}.
\end{equation}

We use an exchangeable bootstrap, which depends on a set of weights $\{W_i\}_{i=1}^n$ that put probability mass $W_i$ at each observation $Z_i$.  This type of bootstrap describes many well-known bootstrap techniques including sampling with replacement from the observations \citep[Section 3.6.2]{vanderVaartWellner96}.

We make the following assumptions to ensure the bootstrap is well-behaved.
\begin{enumerate}[label=\textbf{A\arabic*.}, ref=\textbf{A\arabic*}]
    \setcounter{enumi}{1}
  \item \label{A:sequences} Assume $a_n, b_n \rightarrow 0^+$, $a_n r_n \rightarrow \infty$ and $b_n r_n \rightarrow \infty$.
  \item \label{A:bootstrapf} Suppose that for each $n$, $W$ is independent of the data $Z$ and there is a map $\{Z_i, W_i\}_{i=1}^n \mapsto f_n^*$ where $f_n^* \in \ell^\infty(\gr A)$.  Suppose $r_n(f_n^* - f_n)$ is asymptotically measurable; for all continuous and bounded $g$, $g(r_n(f_n^* - f_n))$ is a measurable function of $\{W_i\}$ outer almost surely in $\{Z_i\}$; and $r_n(f_n^* - f_n) \cw \calG_P$, where $\calG_P$ was defined in Assumption~\ref{A:estimator}.
\end{enumerate}

We remark that it is difficult to propose a general method for choosing how to set $a_n$ and $b_n$ in practice.  We numerically investigated their effect in simulations. We apply similar rules to those used in \citet{LintonSongWhang10} and \citet{LeeSongWhang18}.  However, in general situations, where the variance of the objective function changes significantly over its domain, one would need to estimate the pointwise standard deviation of the function.  For example, \citet{FanPark10} provide pointwise asymptotic distributions for estimated dependency bound functions for $F_{X-Y}$. Nevertheless, we leave this extension for future research.

Assumption~\ref{A:sequences} is used to ensure consistency of the $\epsilon$-maximizer set estimators used in the bootstrap algorithm.  Assumption~\ref{A:bootstrapf} is a condensed version of Assumption 3 of~\citet{FangSantos19} to ensure that $r_n(f_n^* - f_n)$ behaves asymptotically like $r_n(f_n - f)$.  

\vspace{1em}
\noindent \textbf{Resampling routine to estimate the distribution of $r_n(\lambda_j(f_n) - \lambda_j(f))$}
\begin{enumerate}
  \item Estimate $\hat{\lambda}_{jn}'$ using sample data and formulas~\eqref{est1_def}-\eqref{est4_def} above.
\end{enumerate}
Then repeat steps \ref{resample12}-\ref{resample22} for $r = 1, \ldots R$:
\begin{enumerate}
    \setcounter{enumi}{1}
  \item \label{resample12} Use an exchangeable bootstrap to construct $r_n(f_n^* - f_n)$.
  \item \label{resample22} Calculate the resampled test statistic $\lambda^*_r = \hat{\lambda}_{jn}'(r_n(f_n^* - f_n))$ using the estimate from step 1.
\end{enumerate}
Finally,
\begin{enumerate}
    \setcounter{enumi}{3}
  \item Let $\hat{q}_{\lambda^{*}}(1-\alpha)$ be the $(1-\alpha)$-th sample quantile from the bootstrap distribution of $\{\lambda_r^*\}_{r=1}^R$, where $\alpha \in (0, 1)$ is the nominal size of the test. Reject the null hypothesis if $r_n \lambda_j(f_n)$ is larger than $\hat{q}_{\lambda^{*}}(1-\alpha)$.
\end{enumerate}

The consistency of this resampling procedure under the null hypothesis is summarized in the following theorem.  To discuss weak convergence it is easiest to use the space of bounded Lipschitz functions, which indicates weak convergence.  That is, defining $BL_1(\R) = \{g \in C_b(\R) : \sup_x |g(x)| \leq 1, \sup_{x \neq y} |g(x) - g(y)| \leq |x - y|\}$, $X_n$ converges weakly to $X$ if and only if $\sup_{g \in BL_1(\R)} |\ex{g(X_n)} - \ex{g(x)}| \rightarrow 0$ \citep[\S 1.12]{vanderVaartWellner96}.

\begin{thm} \label{thm:bootstrap_consistency_teststats}
  Under Assumptions~\ref{A:estimator}-\ref{A:bootstrapf}, if either $j \in \{1, 3\}$ and $P \in \calP_{00}^E$, or $j \in \{2, 4\}$ and $P \in \calP_{00}^I$, then for any $\epsilon > 0$,
  \begin{equation*}
    \limsup_{n \rightarrow \infty} \prob{ \sup_{g \in BL_1(\R)} \left| \ex{g \left( \hat{\lambda}_{jn}' \left( r_n (f_n^* - f_n) \right) \right) \big| \{Z_i\}_{i=1}^n} - \ex{g \left( \lambda_{jf}'(\calG_P) \right)} \right| > \epsilon } = 0.
  \end{equation*}
\end{thm}

We conclude this section by revisiting Examples \ref{ex:bounds}, \ref{ex:stoc_dominance}, and \ref{ex:legendre} and describing the resampling procedures in the corresponding cases.

\begin{excont}{ex:bounds}{Dependency bounds}
  Recall that tests using the statistic $\lambda(F)$ defined in~\eqref{dep_stat_def} may be inverted to find a uniform confidence band for $L_{X+Y}$ (here we focus on the lower dependency bound for the CDF $F_{X+Y}$ only).  Suppose that we observe two independent random samples, $\{X_i\}_{i=1}^{n_X}$ and $\{Y_i\}_{i=1}^{n_Y}$ and $\bbF_n = (\bbF_{Xn}, \bbF_{Yn})$ be their empirical CDFs.  Then
  \begin{equation} \label{twosided_ex}
    \lambda(\bbF_n) = \sup_{x \in \R} | \bbL_n(x) - L_0(x) |,
  \end{equation}
  where $\bbL_n(x) = \bbL_n(x, \bbF_n) = \max \left\{ \sup_{u \in \R} (\bbF_{Xn}(u) + \bbF_{Yn}(x - u)) - 1, 0 \right\}$.

  Standard conditions ensure that, with $n = n_X + n_Y$, $\sqrt{n}(\bbF_n - F) \cw \calG_F$, where $\calG_F$ is a Gaussian process in $(\ell^\infty(\R))^2$.  Under the hypothesis that $L_0$ is the true lower bound function, $\sqrt{n} \lambda(\bbF_n) \cw \lambda'_F(\calG_F)$, where $\lambda_F'$ is defined in~\eqref{cband_derivative}.  This derivative needs to be estimated as in~\eqref{est1_def}, estimating the set-valued map $U_F(x, \epsilon)$ with the plug-in estimate
  \begin{equation}
    U_{\bbF_n}(x, a_n) = \left\{ u \in \R : \max\{ \bbF_{Xn}(u) + \bbF_{Yn}(x - u), 0 \} \geq \bbL_n(x) - a_n \right\},
  \end{equation}
  where $a_n$ is a sequence that converges slowly to zero, and with an additional slowly-decreasing sequence $b_n$ that is used to determine the regions where $\Pi^+(\bbF_n)$ falls in the interval $(-b_n, b_n)$ or is above $b_n$ to estimate $\hat{X}_0$ and $\hat{X}_+$.  Labeling the estimated derivative $\hat{\lambda}_F'$, slight modifications of Corollary~\ref{cor:null_asymptotics} and Theorem~\ref{thm:bootstrap_consistency_teststats} imply that for all $q \geq 0$,
  \begin{equation*}
    \prob{ \hat{\lambda}_F' (\sqrt{n} (\bbF_n^* - \bbF_n)) \leq q \; \Big| \{X_i\}_{i=1}^{n_X}, \{Y_i\}_{i=1}^{n_Y}  } \cp \prob{ \lambda_F' (\calG_F) \leq q}.
  \end{equation*}
  We can find a critical value of the asymptotic distribution using the bootstrap by estimating
  \begin{equation*}
    \hat{q}_{\lambda^*}(1-\alpha) = \min \left\{ q : \prob{ \hat{\lambda}_F' (\sqrt{n} (\bbF_n^* - \bbF_n)) \leq c \; \Big| \{X_i\}_{i=1}^{n_X}, \{Y_i\}_{i=1}^{n_Y} } \geq 1 - \alpha \right\}
  \end{equation*}
  by simulation.  For $r = 1, \ldots R$, let $\lambda_r^* = \hat{\lambda}_F' (\sqrt{n} (\bbF_n^* - \bbF_n))$, and let $\hat{q}_{\lambda^*}(1-\alpha)$ be the $(1-\alpha)$-th quantile of the bootstrap sample $\{\lambda^*_r\}_{r=1}^R$.  Simulation evidence presented in the supplementary appendix verifies that the coverage probability of these intervals is accurate in moderate sample sizes by examining the empirical rejection probabilities of the associated tests.
\end{excont}

\begin{excont}{ex:stoc_dominance}{Stochastic dominance}
Suppose that three independent random samples $\{X_{0i}\}_{i=1}^n$, $\{X_{Ai}\}_{i=1}^n$ and $\{X_{Bi}\}_{i=1}^n$ are observed.  Let $\bbG_n = (\bbG_{0n}, \bbG_{An}, \bbG_{Bn})$ be their empirical distribution functions, and define the sample statistic
  \begin{equation} \label{eq:sd_test_stat}
    \hat{\Lambda}(\bbG_n) = \left( \int_\R \left[ \bbL_{An}(x) - \bbU_{Bn}(x) \right]_+^2 \ud m(x) \right)^{1/2},
  \end{equation}
  where $\bbL_{An}$ is the plug-in estimate of $L_A$ using $\bbG_{0n}$ and $\bbG_{An}$ and $\bbU_{Bn}$ is the plug-in estimate of $U_B$ using $\bbG_{0n}$ and $\bbG_{Bn}$.  Under standard conditions, with $n = n_0 + n_A + n_B$, $\sqrt{n}(\bbG_n - G) \cw \calG_G$, where $\calG_G$ is a Gaussian process in $(\ell^\infty(\R))^3$.  The theory above can be extended in a straightforward way to show that $\hat{\Lambda}(\sqrt{n} \bbG_n) \cw \Lambda_G'(\calG_G)$, where $\Lambda_G'(h)$ was defined in~\eqref{biglambdaprime_def}.  To estimate the distribution of $\Lambda_G'(\calG_G)$, some estimates of the derivative are required.  Given sequence $\{b_n\}$, let
  \begin{equation*}
    \hat{X}_0 = \{x \in \R: |\psi(\Pi^-(\bbG_0, \bbG_A))(x) - \psi(-\Pi^-(\bbG_0, \bbG_B))(x)| \leq b_n\},
  \end{equation*}
  and estimate the near-maximizers in $u$ for each dependency bound and all $x$ as in the previous example.  Estimate the distribution by calculating resampled statistics, for $r = 1, \ldots, R$,
  \begin{multline*}
    \Lambda^*_r = \sqrt{n} \Bigg( \int_{\hat{X}_0} \bigg[ \sup_{u \in U_{\Pi^-(\bbG_{0n}, \bbG_{An})}(x, a_n)} \Pi(\bbG^*_{0n} - \bbG_{0n}, \bbG^*_{An} - \bbG_{An})(u, x) \\
    + \sup_{u \in U_{-\Pi^-(\bbG_{0n}, \bbG_{Bn})}(x, a_n)} (-\Pi)(\bbG_{0n}^* - \bbG_{0n}, \bbG_{Bn}^* - \bbG_{Bn}) (u, x) \bigg]_+^2 \ud m(x) \Bigg)^{1/2}.
  \end{multline*}
  A test can be conducted by comparing $\sqrt{n}\hat{\Lambda}(\bbG_n)$ to the $(1-\alpha)$-th quantile of the bootstrap distribution $\{\Lambda^*_r\}_{r=1}^R$.  A simulation in the supplementary appendix illustrates the accurate size and power of this testing strategy.
\end{excont}

\begin{excont}{ex:legendre}{Cost and profit functions}
  Now consider estimating a confidence band for a conjugate function.  Suppose that we observe $\{X_i\}_{i=1}^n$ and have a parametric model $c(x, \theta)$, estimated using $c(x, \hat{\theta})$.  Supposing that $\theta \in \R^p$, this is a VC-subgraph class of functions and $\sqrt{n}(c(\cdot, \hat{\theta}) - c(\cdot, \theta)) \cw \calG_c$ where $\calG_c$ is a Gaussian process \citep[Example 19.7]{vanderVaart98}.  To estimate the directional derivative $\lambda'(\calG_c)$ where $\lambda'$ is defined in~\eqref{leg_derivative}, we need to estimate several quantities.  Let $\widehat{\check{\mu}} = \max_{p, x} (px - c(x, \hat{\theta}))$ and $\widehat{\check{\sigma}} = \max_p \min_x (c(x) - px)$.  For a slowly decreasing sequence $a_n$, decide whether $\widehat{\check{\mu}} + a_n < \widehat{\check{\sigma}}$, $|\widehat{\check{\mu}} - \widehat{\check{\sigma}}| \leq a_n$ or $\widehat{\check{\mu}} > \widehat{\check{\sigma}} + a_n$.  Given this determination, estimate the derivative $\check{\lambda}'(\calG_c)$ by $\widehat{\check{\sigma}}'(\calG_c^*)$, $\max\{ \widehat{\check{\mu}}'(\calG_c^*), \widehat{\check{\sigma}}'(\calG_c^*)\}$ or $\widehat{\check{\mu}}'(\calG_c^*)$ respectively, where $\calG_c^* = \sqrt{n}(c(\cdot, \hat{\theta}^*) - c(\cdot, \hat{\theta}))$.  Given $\check{q}_{\lambda^*}(1-\alpha/2)$, a quantile from the bootstrap sample $\{\check{\lambda}^*_r\}_{r=1}^R$ calculated as just described, we could estimate a symmetric level $(1-\alpha)$ uniform confidence band for $\mathcal{L}(c)$ using $\mathcal{L}(c(\cdot, \hat{\theta})) \pm \check{q}_{\lambda^*}(1-\alpha/2)$.
\end{excont}

\subsection{Local size control}
It is also of interest to examine how these tests behave under sequences of distributions local to distributions that satisfy the null hypothesis.  We consider sequences of local alternative distributions $\{P_n\}$ such that for each $n$, $\{Z_i\}_{i=1}^n$ are distributed according to $P_n$, and $P_n$ converges towards a limit $P$ that satisfies the null hypothesis.  To describe this process, for $t \geq 0$ define a path $t \mapsto P_t$, where $P_t$ is an element of the space of distribution functions $\calP$, such that
\begin{equation} \label{local_sqrt}
  \lim_{t \rightarrow 0} \int \left( ((\ud P_t)^{1/2} - (\ud P_0)^{1/2}) / t - \frac{1}{2} h (\ud P_0)^{1/2} \right)^2 = 0,
\end{equation}
where the score function $h \in L_2(F)$ satisfies $\ex{h} = 0$ and $P_0$ is a distribution that satisfies the null hypothesis.  The direction that the sequence approaches the null is described asymptotically by the score $h$.  We assume that by letting $t = c / r_n$ for $c \in \R$, we can parameterize distributions that are local to $P_0$ and for $t \geq 0$ denote $f(P_t)$ as the function $f$ under distribution $f(P_t)$ so that the unmarked $f$ described above can be rewritten $f = f(P_0)$.  See, e.g., \citet[Section 3.10.1]{vanderVaartWellner96} for more details.

The following assumption ensures that $f$ remains suitably regular under such local perturbations to the null distribution.
\begin{enumerate}[label=\textbf{A\arabic*.}, ref=\textbf{A\arabic*}]
    \setcounter{enumi}{3}
  \item \label{A:local_alts} For all $c \in \R$,
    \begin{enumerate}
      \item For given $c$, there exists some $f'(c) \in \ell^\infty(\gr A)$ such that $\|r_n (f(P_{c/r_n}) - f(P_0)) - f'(c)\|_\infty \rightarrow 0$, where $P_{c/r_n}$ satisfy~\eqref{local_sqrt}.
      \item $r_n(f_n - f(P_{c/r_n})) \cw \calG_{P_0}$ in $\ell^\infty(\gr A)$, where for each $n$, $\{Z_i\}_{i=1}^n \sim P_{c/r_n}$.
    \end{enumerate}
\end{enumerate}
Both parts of Assumption~\ref{A:local_alts} ensure that $f_n$ behave regularly as distributions drift towards $P_0 \in \calP_{00}^E$ or $P_0 \in \calP_{00}^I$.  These additional regularity conditions allow us to describe the size and local power properties of the test statistics.

\begin{thm} \label{thm:size_control}
  Under Assumptions~\ref{A:estimator}-\ref{A:local_alts}:
  \begin{enumerate}
    \item Suppose $j = 1$ or $3$, $P_0 \in \calP_{00}^E$ and for each $n$, the observations $\{Z_i\}_{i=1}^n$ are distributed according to $P_n = P_{c / r_n} \in \calP_{00}^E$.  Suppose that the CDF of $\lambda_{jf}'(\calG_{P_0})$ is continuous and increasing at $q_{1-\alpha}$, its $(1-\alpha)$-th quantile.  Then
      \begin{equation*}
        \limsup_{n \rightarrow \infty} P_n \left\{ r_n \lambda_j(f_n) > \hat{q}_{\lambda^*}(1-\alpha) \right\} \geq P_0 \left\{ \lambda_{jf}'( \calG_{P_0} + f'(c) ) > q_{1-\alpha} \right\}.
      \end{equation*}
      This holds with equality if $q_{1-\alpha}$ is a continuity point of the CDF of $\lambda_{jf}'( \calG_{P_0} + f'(c) )$.  If $c = 0$ then the limiting rejection probability is equal to $\alpha$.
    \item Suppose $j = 2$ or $4$, $P_0 \in \calP_{00}^I$ and for each $n$, the observations $\{Z_i\}_{i=1}^n$ are distributed according to $P_n = P_{c/r_n} \in \calP_0^I$.  Suppose that the CDF of $\lambda_{jf}'(\calG_{P_0})$ is continuous and increasing at its $(1-\alpha)$-th quantile.  Then
      \begin{equation*}
        \limsup_{n \rightarrow \infty} P_n \left\{ r_n \lambda_j(f_n) > \hat{q}_{\lambda^*}(1-\alpha) \right\} \leq \alpha.
      \end{equation*}
  \end{enumerate}
\end{thm}

Theorem~\ref{thm:size_control} shows that the size of tests can be controlled locally to the null region (and the nominal rejection probability matches the intended probability) in only some cases.  In particular, the tests for the null that $P_0 \in \calP_{00}^E$ cannot be shown to control local size without more information about the direction of the local alternatives.  This is because we must make assumptions about the underlying objective functions without being able to make similar assumptions about the corresponding value functions~--- note that Assumption~\ref{A:local_alts} is made with regard to $f(P)$, not $\psi(f(P))$.  To see where the problem lies we may take $\lambda_1$ as an example.  The result of Theorem~\ref{thm:size_control} shows that when $P_n \in \calP_{00}^E$ for each $n$,
\begin{equation*}
  \sqrt{r_n} \lambda_1(f_n) \cw \lambda_{1f}'(\calG_{P_0} + f'(c)). 
\end{equation*}
Although $P \in \calP_{00}^E$ implies $\lambda_f'(f'(c)) = 0$, as shown in the proof of Theorem~\ref{thm:size_control}, it does not imply that $f'(c) \equiv 0$.  For example, if $f$ are CDFs corresponding to a location-shift family of distributions, i.e., for all $c \in \R$, $G_{c}(x) = G_0(x - c)$ with differentiable densities that have square-integrable derivatives.  Then since the $h$ in~\eqref{local_sqrt} is $cg'_0(x) / g_0(x)$, the $f'(c)$ of assumption~\ref{A:local_alts} would be $-cg_0(x)$.  This non-zero derivative can cause problems when optimizing it interacts with the absolute value function.  Generally $|\sup_u (f(u) + g(u))| \not\leq |\sup_u f(u)| + |\sup_u g(u)|$ (take $f$ and $g$ strictly less than zero), and we do not achieve a simplification that would imply the probability on the right-hand side is less than or equal to $\alpha$.  We cannot assume regularity of $\psi(f)$ because this would require the existence of a derivative $\psi_f'(\cdot) \in \ell^\infty(X)$, but this derivative cannot generally exist.  It may be that size control could be shown for special cases (for example when theory implies a shape for the value function), but not in general.  On the other hand, for the same reason, the size of one-sided test statistics is as intended.  That is, when the value function is found through maximization, the positive-part map in the definitions of $\lambda_2$ and $\lambda_4$ interacts with it in a way that maintains size control.  In contrast, the negative-part map would not guarantee size control.  A stronger result in this vein will be shown in the next section.

The inequality in the second part of this theorem results from the one-sidedness of the test statistics $\lambda_2$ and $\lambda_4$.  This is related to a literature in econometrics on moment inequality testing.  Tests may exhibit size that is lower than nominal for local alternatives that are from the interior of the null region.  A few possible solutions to this problem have been proposed.  For example, one might evaluate the region where the moments appear to hold with equality, which leads to contact set estimates like in the bootstrap routine described above \citep{LintonSongWhang10}.  Alternatively, we may alter the reference distribution by shifting it in the regions where equality does not seem to hold \citep{AndrewsShi17}.

\subsection{Uniform size control}

Uniformity of inference procedures in the data distribution was introduced in~\citet{GineZinn91} and~\citet{SheehyWellner92}, and uniformity has become a topic of great interest in the econometrics literature (see, for example, \citet{AndrewsGuggenberger10}, \citet{LintonSongWhang10}, \citet{RomanoShaikh12} or \citet{HongLi18} for an application like ours).  Under stronger assumptions that ensure the underlying $f_n$ converge weakly to a limiting process uniformly over the set of possible data distributions, some of the above results can be extended from size control under local alternative distributions to size control that holds uniformly over the class of distributions that satisfy the null.  To ensure uniformity we assume the following regularity conditions are satisfied.

\begin{enumerate}[label=\textbf{A\arabic*.}, ref=\textbf{A\arabic*}]
    \setcounter{enumi}{4}
  \item \label{A:unif_convf} Suppose that $r_n (f_n - f) \cw \calG_P$ uniformly over $P \in \calP$:
    \begin{equation*}
      \limsup_{n \rightarrow \infty} \sup_{P \in \calP} \sup_{g \in BL_1(\ell^\infty(\gr A))} \left| \ex{g(r_n(f_n - f(P)))} - \ex{g(\calG_P)} \right| = 0.
    \end{equation*}
  \item \label{A:unif_tight} Suppose that the limiting process $\calG_P$ is tight uniformly over $P \in \calP$:
    \begin{equation*}
      \lim_{M \rightarrow \infty} \sup_{P \in \calP} \prob{ \|\calG_P\|_\infty \geq M } = 0.
    \end{equation*}
  \item \label{A:regular} Suppose that on $\calP_{00}^I$, the processes $\calG_P$ are regular for $\lambda_{2f}'$ or $\lambda_{4f}'$ in the following sense: for any $0 < \alpha < 1/2$, there exists a $\underline{q} > 0$ depending only on $\alpha$ such that
    \begin{equation*}
      \sup_{P \in \calP_{00}^I} \prob{ \lambda_{jf}'(\calG_P) < \underline{q} } < 1 - \alpha
    \end{equation*}
    and for any $s > 0$,
    \begin{equation*}
      \limsup_{\eta \rightarrow 0^+} \sup_{P \in \calP_{00}^I} \prob{ \left| \lambda_{jf}'(\calG_P) - s \right| \leq \eta } = 0.
    \end{equation*}
    In addition, suppose that over $\calP_0^I \backslash \calP_{00}^I$, the processes $\calG_P$ are regular for statistics related to the entire domain $X$ instead of $X_0$~--- letting $\tilde{\lambda}_{2f}'(h) = \lim_{\epsilon \rightarrow 0^+} \sup_{(u, x) \in U_f(X, \epsilon)} [h(u, x)]_+$ and $\tilde{\lambda}_{4f}'(h) = (\int_X \psi_f'(h)(x) \ud m(x))^{1/p}$, assume that the above two displays are satisfied with $\tilde{\lambda}_{jf}'(\calG_P)$ in the place of $\lambda_{jf}'(\calG_P)$.
  \item \label{A:unif_boot} Suppose that conditional on the observations $\{Z_i\}_{i=1}^n$, $r_n(f_n^* - f_n) \cw \calG_P$ uniformly over $P \in \calP$: for all $\epsilon > 0$,
    \begin{equation*}
      \limsup_{n \rightarrow \infty} \sup_{P \in \calP} \prob{ \sup_{g \in BL_1(\ell^\infty(\gr A))} \left| \ex{g(r_n(f_n^* - f_n)) | \{Z_i\}_{i=1}^n } - \ex{g(\calG_{P})} \right| > \epsilon } = 0.
    \end{equation*}
\end{enumerate}

The above assumptions help define the collection $\calP$ on which we may assert that inference is valid uniformly for $P \in \calP$.  Assumptions~\ref{A:unif_convf} and~\ref{A:unif_tight} require that the sample objective functions converge to well-behaved limits uniformly on $\calP$.  Assumption~\ref{A:regular} requires that also the limiting distributions of the sample statistics $r_n(\lambda_j(f_n) - \lambda_j(f))$ are regular enough over $\calP_0^I$ that their CDFs do not have extremely large atoms (i.e., extending above the $(1-\alpha)$-th quantile of the distribution) and are strictly increasing at the relevant critical value.  This assumption may be simplified when the $\calG_P$ are Gaussian since the functionals $\lambda_2$ and $\lambda_4$ are convex \citep[Theorem 11.1]{DavydovLifshitsSmorodina98}.  Assumption~\ref{A:unif_boot} requires that the bootstrap objective functions also converge uniformly, as the original sample functions do.  This last assumption may be implied by Assumption~\ref{A:unif_convf}, see Lemma~A.2 of~\citet{LintonSongWhang10}, for example.

\begin{thm} \label{thm:uniform_size_control}
  Under Assumptions~\ref{A:estimator}-\ref{A:bootstrapf} and~\ref{A:unif_convf}-\ref{A:unif_boot}, for $j = 2$ or $4$,
  \begin{equation*}
    \limsup_{n \rightarrow \infty} \sup_{P \in \calP_0^I} \prob{ r_n \lambda_j(f_n) > \hat{q}_{\lambda^*}(1-\alpha) } \leq \alpha.
  \end{equation*}
\end{thm}

Uniformity for $P \in \calP$ is maintained by the one-sided statistics when optimization matches the ``direction'' of the test.  That is, when the value function matches the positive-part map in $\lambda_2$ and $\lambda_4$, we have uniformity over $P \in \calP$, and analogously uniform inference would be possible for minimization and the negative-part map.  Uniformity can be lost when optimization and the direction of the test do not match, which may be the case with two-sided tests.

In the next section we illustrate the usefulness of our results by providing details on the construction of uniform confidence bands around bound functions for the CDF of a treatment effect distribution.  See the first section of the supplemental appendix for numerical simulations evaluating the finite-sample performance of the tests developed in Examples~\ref{ex:bounds} and~\ref{ex:stoc_dominance}.

\section{The treatment effect distribution of job training on wages}\label{sec:application}

This section illustrates the inference methods with an evaluation of a job training program.  We construct upper and lower bounds for both the distribution and quantile function of the treatment effects, confidence bands for these bound function estimates, and describe a few inference results. This application uses an experimental job training program data set from the National Supported Work (NSW) Program, which was first analyzed by \citet{LaLonde86} and later by many others, including \citet{HeckmanHotz89}, \citet{DehejiaWahba99}, \citet{SmithTodd01, SmithTodd05}, \citet{Imbens03}, and \citet{Firpo07}.

Recent studies in statistical inference for features of the treatment effects distribution in the presence of partial identification include, among others, \citet{FirpoRidder08, FirpoRidder19, FanPark09, FanPark10, FanWu10, FanPark12, FanShermanShum14, FanGuerreZhu17}. These studies have concentrated on distributions of finite-dimensional functionals of the distribution and quantile functions, including these functions themselves evaluated at a point.  Additional work includes, among others, \citet{GautierHoderlein11}, \citet{ChernozhukovLeeRosen13}, \citet{Kim14} and \citet{ChesherRosen15}.  Each of these papers provides pointwise inference methods for bounds on the distribution or quantile functions, and often for more complex objects.  We contribute to this literature by applying the general results in this paper to provide uniform inference methods for the bounds developed by Makarov and others, and hope that it may indicate the direction that pointwise inference for bounds in more involved models may be extended to be uniformly valid.

The data set we use is described in detail in \citet{LaLonde86}.  We use the publicly available subset of the NSW study used by \citet{DehejiaWahba99}. The program was designed as an experiment where applicants were randomly assigned into treatment. The treatment was work experience in a wide range of possible activities, such as learning to operating a restaurant or a child care center, for a period not exceeding 12 months. Eligible participants were targeted from recipients of Aid to Families With Dependent Children, former addicts, former offenders, and young school dropouts. The NSW data set consists of information on earnings and employment (outcome variables), whether treated or not.\footnote{The data set also contains background characteristics, such as education, ethnicity, age, and employment variables before treatment.  Nevertheless, since we only use the experimental part of the data we refrain from using this portion of the data.} We consider male workers only and focus on earnings in 1978 as the outcome variable of interest. There are a total of 445 observations, where 260 are control observations and 185 are treatment observations. Summary statistics for the two parts of the data are presented in Table \ref{tb:summary}.

\begin{table}[!ht]
\linespread{1.5}
\par
\begin{center}
{\small \setlength{\tabcolsep}{4pt}
\begin{tabular}{ccccccccccc}
\cline{1-11}
\multicolumn{1}{c}{} & \multicolumn{1}{c}{} & \multicolumn{4}{c}{Treatment Group} & \multicolumn{1}{c}{} & \multicolumn{4}{c}{Control Group}\\
\cline{3-6}\cline{8-11}
\multicolumn{1}{c}{} & \multicolumn{1}{c}{} & \multicolumn{1}{c}{Mean} & \multicolumn{1}{c}{Median} & \multicolumn{1}{c}{Min.} & \multicolumn{1}{c}{Max.} & \multicolumn{1}{c}{} & \multicolumn{1}{c}{Mean} & \multicolumn{1}{c}{Median} & \multicolumn{1}{c}{Min.} & \multicolumn{1}{c}{Max.}\\
\hline
Earnings (1978) &  & 6,349.1 & 4,232.3 & 0.0 & 60,307.9 &   & 4,554.8 & 3,138.8 & 0.0 & 39,483.5 \\
&  & (7,867.4) &  &  &  &   & (5,483.8) &  &  &  \\
\hline
\end{tabular}
}
\end{center}
  \caption{Summary statistics for the experimental National Supported Work (NSW) program data.}
\label{tb:summary}
\end{table}

To provide a more complete overview of the data, we also compute the empirical CDFs of the treatment and control groups in Figure~\ref{fig:lalonde_cdfs}. From this figure we note that the empirical treatment CDF stochastically dominates the empirical control CDF, and that there are a large number of zeros in each sample. In particular, $\bbF_{tr,n}(0) \approx 0.24$ and $\bbF_{co,n}(0) \approx 0.35$.

\begin{figure}[!ht]
\begin{center}
  \includegraphics[height=3in]{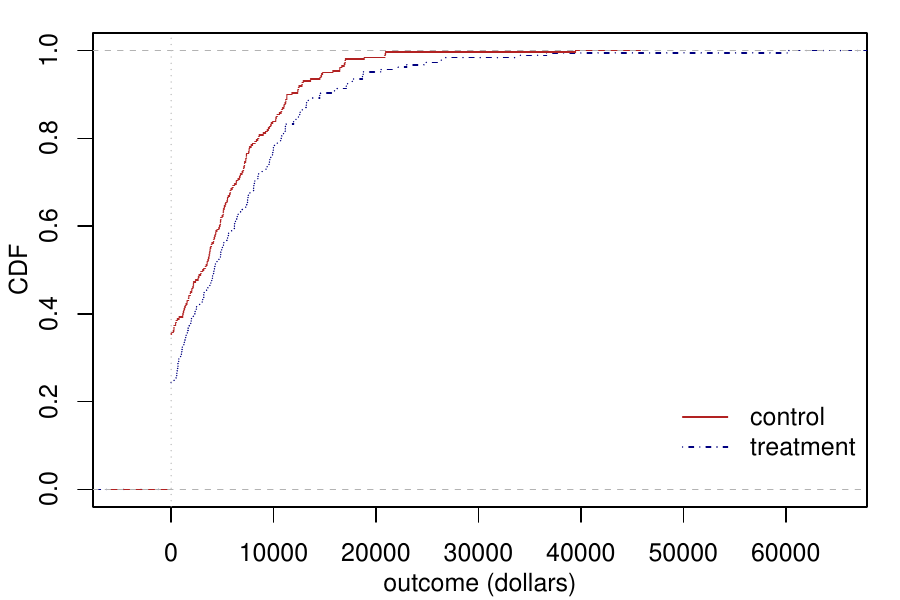}
\end{center}
  \caption{Empirical CDFs of treatment and control observations in the experimental NSW data.  There are 260 control group observations and 185 treatment group observations, many outcomes equal to zero and the treatment group outcomes stochastically dominate the control group outcomes at first order.} \label{fig:lalonde_cdfs}
\end{figure}

Suppose a binary treatment is independent of two potential outcomes $(X_{co}, X_{tr})$, where $X_{co}$ denotes outcomes under a control regime and $X_{tr}$ denotes outcomes under a treatment, and $X_{co}$ and $X_{tr}$ have marginal distribution functions $F_{co}$ and $F_{tr}$ respectively.  Suppose that interest is in the distribution of the treatment effect $\Delta = X_{tr} - X_{co}$ but we are unwilling to make any assumptions regarding the dependence between $X_{co}$ and $X_{tr}$.  In this section we study the relationship between the identifiable functions $F_{co}$ and $F_{tr}$ and functions that bound the distribution function of the unobservable random variable $\Delta$.

$F_\Delta(\cdot)$ is not point-identified because the full bivariate distribution of $(X_{co}, X_{tr})$ is unidentified and the analyst has no knowledge of the dependence between potential outcomes.  However, $F_\Delta$ can be bounded.  Suppose we observe samples $\{X_{ki}\}_{i=1}^{n_k}$ for $k \in \{co, tr\}$.  Define the empirical lower and upper bound functions
\begin{align}
  \bbL_{\Delta,n}(x) &= \sup_{u \in \R} \{ \bbF_{tr,n}(u) - \bbF_{co,n}(u - x) \} \label{eq:SDLB} \\
  \bbU_{\Delta,n}(x) &= \inf_{u \in \R} \{ 1 + \bbF_{tr,n}(u) - \bbF_{co,n}(u - x) \}. \label{eq:SDUB}
\end{align}
These are plug-in estimates of analogous population bounds $L_\Delta$ and $U_\Delta$, and are similar to Example~\ref{ex:bounds} because $F_\Delta = F_{(X_{tr}-X_{co})}$.  We make the following assumptions on the observed samples of treatment and control observations.

\begin{enumerate}[label=\textbf{B\arabic*.}, ref=\textbf{B\arabic*}]
  \item \label{assume:first} The observations $\{X_{co,i}\}_{i=1}^{n_{co}}$ and $\{X_{tr,i}\}_{i=1}^{n_{tr}}$ are iid samples and independent of each other and are distributed with marginal distribution functions $F_{co}$ and $F_{tr}$ respectively.  Refer to the pair of distribution functions and their empirical distribution function estimates as $F = (F_{co}, F_{tr})$ and $\bbF_n = (\bbF_{co,n}, \bbF_{tr,n})$.
  \item \label{assume:last} The sample sizes $n_{co}$ and $n_{tr}$ increase in such a way that $n_k / (n_{co} + n_{tr}) \rightarrow \nu_k$ as $n_{co}, n_{tr} \rightarrow \infty$, where $0 < \nu_k < 1$ for $k \in \{co, tr\}$.  Define $n = n_{co} + n_{tr}$.
\end{enumerate}

Under Assumptions~\ref{assume:first} and~\ref{assume:last}, it is a standard result \citep[Example 19.6]{vanderVaart98} that for $k \in \{co, tr\}$, $\sqrt{n_k} (\bbF_{kn} - F_k) \cw \calG_k$, where $\calG_{co}$ and $\calG_{tr}$ are independent $F_{co}$- and $F_{tr}$-Brownian bridges, that is, mean-zero Gaussian processes with covariance functions $\rho_k(x, y) = F_k(x \wedge y) - F_k(x) F_k(y)$.  This implies in turn that \( \sqrt{n} (\bbF_n - F) \cw \calG_F = ( \calG_{co} / \sqrt{\nu_{co}}, \calG_{tr} / \sqrt{\nu_{tr}} )\), where $\calG_F$ is a mean-zero Gaussian process with covariance process $\rho_F(x, y) = \text{diag}\{\rho_k(x, y) / \nu_k\}$.

Now we focus on the calculation of a uniform confidence band for $L$ only.  Because the bootstrap algorithm was described previously we only verify that the regularity conditions for this plan hold.  Assumptions~\ref{assume:first} and~\ref{assume:last}, along with the above discussion of the weak convergence of $\sqrt{n}(\bbF_n - F)$ imply that Assumption~\ref{A:estimator} is satisfied.  The class of functions $\{ (I(X \leq x), I(Y \leq y), x,y \in \R \}$ is uniform Donsker and the sequences $a_n$ and $b_n$ can be chosen to satisfy Assumption~\ref{A:sequences} by making them converge to zero more slowly than $n^{-1/2}$.  Assuming the weights are independent of the observations, assumption~\ref{A:bootstrapf} is satisfied by Lemma A.2 of \citet{LintonSongWhang10}, which implies that the bootstrap algorithm described above is consistent, as described in Theorem~\ref{thm:bootstrap_consistency_teststats}.  It is also straightforward to verify that, under the high-level assumption~\eqref{local_sqrt}, both parts of Assumption~\ref{A:local_alts} are satisfied~--- in the language of the assumption, $f(P_{c/\sqrt{n}})$ are the pair $(F_{co}^n, F_{tr}^n)$ under the local probability distribution $P_{c/\sqrt{n}}$, and $f' = (\int_{-\infty}^\cdot h_{co}^c, \int_{-\infty}^\cdot h_{tr}^c)$) for direction $h^c$ indexed by $c$ and $\sqrt{n}(\bbF_n - F^n) \cw \calG_F$ \citep[Theorem 3.10.12]{vanderVaartWellner96}.

The main objective is to provide uniform confidence bands for the CDF of the treatment effects distribution.  We calculate the lower and upper bounds for the distribution function using the control and treatment samples as in equations \eqref{eq:SDLB}--\eqref{eq:SDUB}.  The bounds are computed on a grid with increments of \$100 dollars along the range of the common support of the bounds, which is roughly from -\$40,000 to \$60,000. We focus on the region between -\$40,000 and \$40,000, which contains almost all the observations (there is a single \$60K observation in the treated sample). The results are presented in Figure~\ref{fig:uniform_bounds_lalonde} and given by the black solid lines in the picture. The most prominent feature is that, as expected, the upper bound for the CDF of treatment effects stochastically dominates the corresponding lower bound.

Next, we compute the uniform confidence bands as described in the text.  They are shown in Figure~\ref{fig:uniform_bounds_lalonde} as the dashed lines around the corresponding solid lines.
\begin{figure}[!ht]
\begin{center}
  \includegraphics[height=3in]{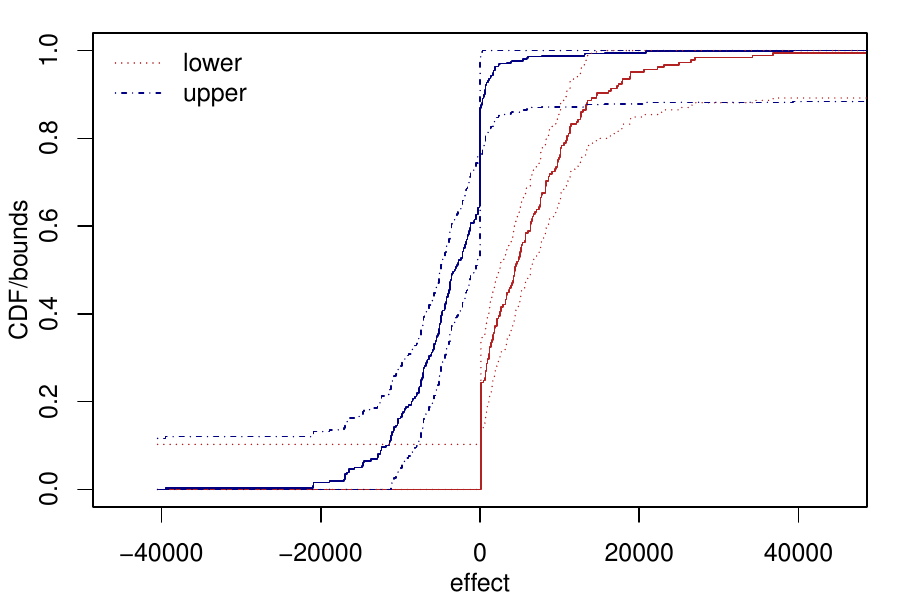}
\end{center}
  \caption{Bound functions and their uniform confidence bands.  These confidence bands were constructed by inverting Kolmogorov-Smirnov-type test statistics as described in the text.} \label{fig:uniform_bounds_lalonde}
\end{figure}
Due to the large number of zero outcomes in both samples, these bounds have some interesting features that we  discuss further. First, we note that for any $\epsilon$ greater than zero but smaller than the next smallest outcome (about \$45), $\bbF_{tr,n}(0) - \bbF_{co,n}(-\epsilon) = 0.24$, which explains the jump in the lower bound estimate near zero (it is really for a point in the grid just above zero).  Likewise, for the same $\epsilon$, $\bbF_{tr,n}(-\epsilon) - \bbF_{co,n}(0) = -0.35$, which explains the jump in the upper bound just below zero.  Without these point masses at zero, both bounds would more smoothly tend towards 0 or 1.  Second, the point masses at zero imply another feature of the bounds that can be discerned in the picture.  The upper bound to the left of 0 is the same as $1 - \bbF_{co,n}(-x)$ and the lower bound to the right of zero is the same as $\bbF_{tr,n}(x)$.  Taking the lower bound as an example, for each $x > 0$ find the closest observation from the control sample $y_{co,i^*}$, and set $x^*(x) = X_{co,i^*} + \epsilon$, leading to the supremum $\bbF_{tr,n}(X_{co,i} + \epsilon)$ at every point where $X_{co,i} + \epsilon < X_{tr,j}$ for all $j$ in the treated sample.  It is identical to the empirical treatment CDF for the entire positive part of the support because the treatment first-order stochastically dominates the control.  The situation would be different if there were a jump in the empirical control CDF at least as large as the jump in the empirical treatment CDF at zero.  Because the opposite is the case for the upper bound, it does change slightly above the zero mark, tending from $1 - \bbF_{tr,n}(0) + \bbF_{co,n}(0)$ to 1 as $x$ goes from 0 to the right.

We can also use the confidence bands for the bound functions to construct a confidence band for the true distribution function of treatment effects.  This is shown in Figure~\ref{fig:cdf_confband_lalonde}.  A $1 - \alpha$ confidence band can be constructed by using the upper $1 - \alpha / 2$ limit of the upper bound confidence band and the lower $\alpha / 2$ limit of the lower confidence band.  This band is a uniform asymptotic confidence band for the true CDF, and uniform over correlation between the potential outcomes between samples.  In other words, if $\mathcal{P}$ is the collection of bivariate distributions that have marginal distributions $F_{tr}$ and $F_{co}$, then
\begin{equation*}
  \liminf_{n \rightarrow \infty} \inf_{\textnormal{P} \in \mathcal{P}} \prob{F_\Delta(x) \in CB(x) \text{ for all } x} \geq 1 - \alpha.
\end{equation*}

This confidence band is likely conservative, since
\begin{align*}
  \text{P} \big\{ \exists x &: F_\Delta(x) \not\in CB(x) \big\} \\ {} &= \prob{ \{\exists x : F_\Delta(x) < \bbL_{\Delta,n}(x) - q^*_{L, 1-\alpha/2}/\sqrt{n} \} \cup \{\exists x : F_\Delta(x) > \bbU_{\Delta,n}(x) + q^*_{U, 1-\alpha/2} / \sqrt{n}\} } \\
  {} &\leq \prob{ \{\exists x : L_\Delta(x) < \bbL_{\Delta,n}(x) - q^*_{L, 1-\alpha/2}/\sqrt{n} \} \cup \{\exists x : U_\Delta(x) > \bbU_{\Delta,n}(x) + q^*_{U, 1-\alpha/2}/\sqrt{n} \} }.
\end{align*}
See \citet{Kim14} and \citet{FirpoRidder19} for a more thorough discussion of the sense in which these bounds are not uniformly sharp for the treatment effect distribution function.  We leave more sophisticated, potentially tighter confidence bands for future research.  Note that the technique of \citet{ImbensManski04} cannot be used to tighten these bounds, because the parameter, a function, could violate the null hypothesis at both sides of the confidence band.

\begin{figure}[!ht]
\begin{center}
  \includegraphics[height=3in]{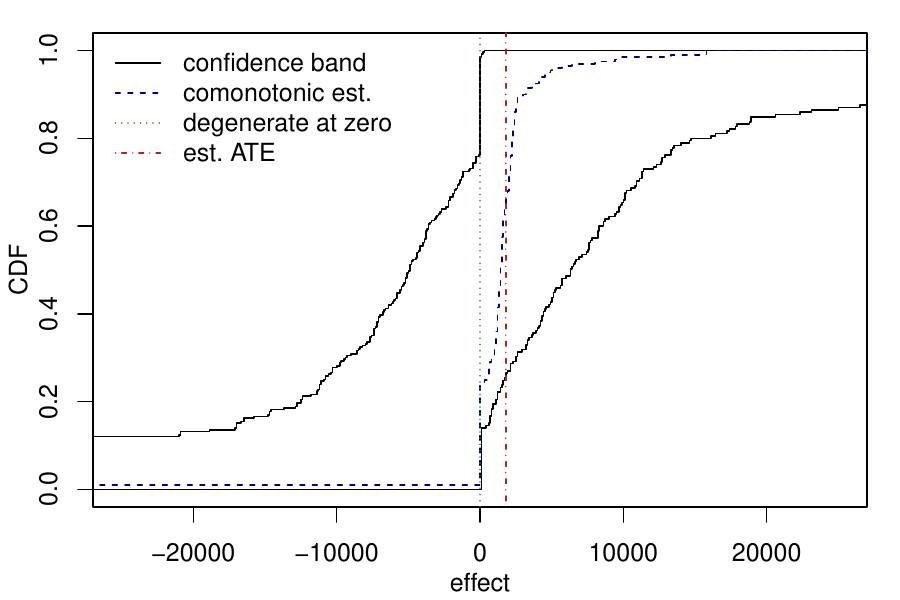}
\end{center}
  \caption{A uniform confidence band for the true treatment effect CDF.  This bound is constructed by using the lower $\alpha/2$ limit of the lower bound and the upper $\alpha/2$ limit of the upper bound.  A few other estimates of the treatment effect distribution are made: the vertical line at zero represents an informal hypothesis test that the effect is zero across the entire distribution, and is rejected (see the calculations at the beginning of this section about the nontrivial parts of the bounds to see why this is so).  The vertical line to the right of zero is the average treatment effect, and it can be seen that this average ignores some variation in treatment effect outcomes.  The dotted curve in between the bounds is the same as the (inverted) quantile treatment effect, which is equivalent to assuming rank invariance between potential treatment and control outcomes.} \label{fig:cdf_confband_lalonde}
\end{figure}

We plot some other features in this figure for context.  First, the dotted vertical line is positioned at $y = 0$, and it can be seen that we (just) reject the null hypothesis $H_0: \prob{\Delta = 0} = 1$.  This hypothesis is closest to non-rejection, and it is clearer that one should reject the null that the treatment effect distribution is degenerate at any other point besides zero.  This supports the notion that treatment effect heterogeneity is an important feature of these observations, especially because this band is completely agnostic about the form of the joint distribution.  On the other hand, by examining the bands at horizontal levels, it can be seen that for the median effect and a wide interval in the center of the distribution, the hypothesis of zero treatment effect cannot be rejected (although these are uniform bounds and not tests of individual quantile levels).

The final feature in the figure is the dashed curve that represents the estimate that one would make under the assumption of comonotonicity (or rank invariance)~--- the assumption that, had an individual been moved from the treatment to the control group, their rank in the control would be the same as their observed rank.  Under this strong assumption the quantile treatment effects are the quantiles of the treatment distribution and they can be inverted to make an estimate.  Clearly, the estimate under this assumption is just one point-identified treatment effect distribution function of many.
\begin{figure}[!ht]
\begin{center}
  \includegraphics[height=3in]{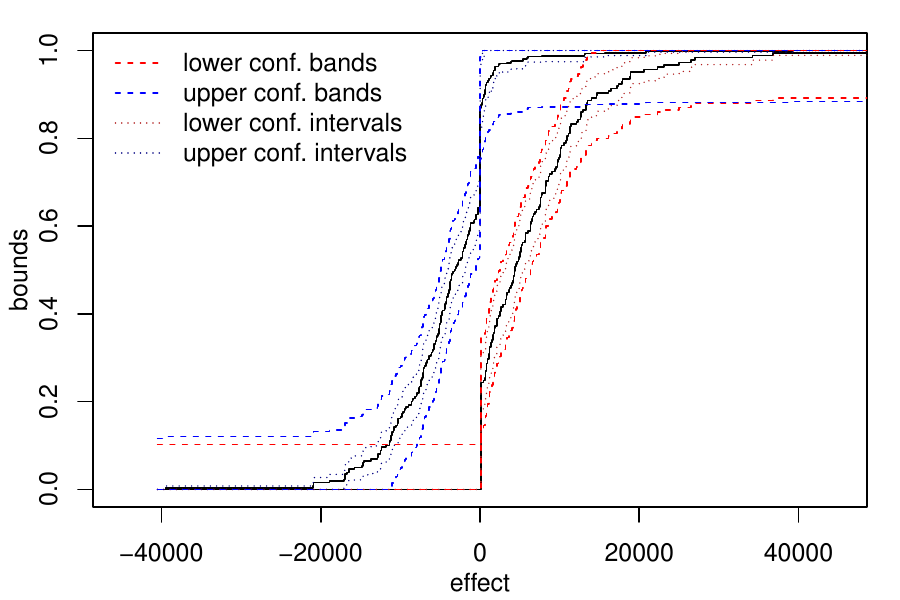}
\end{center}
  \caption{Uniform confidence bands for the CDF of treatment effects plotted with a collection of confidence intervals at treatment outcome levels.  The confidence intervals are narrower because they represent confidence statements at each treatment effect individually, while the uniform bands give confidence statements about where the entire bound function lies.  Pointwise confidence intervals were calculated using the method proposed in \citet{FanPark10}.} \label{fig:bounds_compare}
\end{figure}

Finally, to provide context for the uniform confidence band results within the literature on inference for bounds, we compare the proposed uniform bands to the pointwise confidence intervals suggested by \citet{FanPark10}. We used \citet{BickelSakov08}'s automatic procedure to choose subsample size and constructed confidence intervals for each individual point in the grid of the treatment effect support. This collection of pointwise confidence intervals are plotted along with uniform confidence bands in Figure~\ref{fig:bounds_compare}.  The results show that the uniform bands are farther from the bound estimates than the set of pointwise confidence intervals.

\section{Conclusion}\label{sec:conclusion}

This paper develops uniform statistical inference methods for optimal value functions, that is, functions constructed by optimizing an objective function in one argument over all values of the remaining arguments.  Value functions can be seen as a nonlinear transformation of the objective function.  The map from objective function to value function is not Hadamard differentiable, but statistics used to conduct uniform inference are Hadamard directionally differentiable.  We establish the asymptotic properties of nonparametric plug-in estimators of these uniform test statistics and develop a resampling technique to conduct practical inference.  Examples involving dependency bounds for treatment effects distributions are used for illustration.  Finally, we provide an application to the evaluation of a job training program, estimating a conservative uniform confidence band for the distribution function of the program's treatment effect without making any assumptions about the dependence between potential outcomes.

\section*{Acknowledgements}
We would like to thank Timothy Christensen, Zheng Fang, Hidehiko Ichimura, Lynda Khalaf, Alexandre Poirier, Pedro H.C. Sant'Anna, Andres Santos, Rami Tabri, Brennan Thompson, Tiemen Woutersen, and seminar participants at the University of Kentucky and New York Camp Econometrics XIV for helpful comments and discussions.  We would also like to thank the referees who helped improve the paper, especially the comments of one referee that improved subsection 4.4 considerably.  All the remaining errors are ours.  This research was enabled in part by support from the Shared Hierarchical Academic Research Computing Network (SHARCNET: \url{https://www.sharcnet.ca}) and Compute Canada (\url{http://www.computecanada.ca}).  Code used to implement the methods described in this paper and to replicate the simulations can be found at~\url{https://github.com/tmparker/uniform_inference_for_value_functions}.

\newpage
\appendix
\setcounter{equation}{0}
\renewcommand{\theequation}{S.\arabic{equation}}

\begin{center}
  { \Large \bf Supplementary appendix to ``Uniform inference for value functions''}
\end{center}

This supplemental appendix contains the results of simulation experiments designed to test the size and power of tests described in the main text, and proofs of theorems and lemmas stated in the main text.

\section{Simulation study} \label{sec:bounds}
We investigate the finite-sample behavior of the tests using examples A and B the main text.  Specifically, we consider uniform inference for a lower bound function $L_{X+Y}$ for the distribution function $F_{X+Y}$ without knowledge of the dependence between $X$ and $Y$ by looking at empirical rejection probabilities for a two-sided test that $L_{X+Y} = L_0$ as in Example~\ref{ex:bounds}. Then we investigate the performance of one-sided testing procedures through the stochastic dominance test as outlined in Example~\ref{ex:stoc_dominance}.  Code to reproduce these simulation experiments may be found through the third author's personal website.\footnote{\url{https://github.com/tmparker/uniform_inference_for_value_functions}.}  In all cases, we test size and power in tests with nominal size $5\%$ against local alternatives for samples of size 100, 500 and 1000 with respectively 499, 999 and 1999 bootstrap repetitions for each sample size, and use 1,000 simulation repetitions to calculate empirical rejection probabilities.

\subsection{Tuning parameter selection} \label{sec:tune}

In this experiment, simulated observations are normally distributed and we shift the location of one distribution to examine size and power properties of tests using statistic~\eqref{twosided_ex} from Example~\ref{ex:bounds}.  The bound function $L_{X+Y}$ was estimated using sample data and standard empirical distribution functions.  The bounds must be equal to zero for small enough arguments and equal to one for large enough arguments, and are monotonically increasing in between.  However, $\bbL_n$ may change at all possible $X_i - Y_j$ combinations, and computing the function on all $n_X \times n_Y$ points could be computationally prohibitive.  Therefore we compute the bound functions on a reasonably fine grid $\{x_k\}_{k=1}^K$ (we used increments of $0.05$ in the experiment).

In order to make calculations as efficient as possible, it is helpful to know some features of the support of the bounds given sample data, where the support of a bound is the region where it is strictly inside the unit interval.  The empirical lower bound at a given $x$ may be written as the maximum of the function $\bbF_{Xn}(\cdot) - \bbF_{(x - Y)n}(\cdot)$, which takes steps of size $+1/n_X$ at sample observations $\{X_i\}_{i=1}^{n_X}$ and, given a value of $x$, steps of size $-1/n_Y$ at shifted observations $\{x - Y_j\}_{hj=1}^{n_Y}$.  This function is not right-continuous, but this is not a problem when $Y$ is distributed continuously and calculations are done on a grid.  The upper endpoint of the support of the lower bound is the smallest value for which it is equal to 1.  If for some (large) $x$, $x - Y_j \geq X_i$ for all $i$ and $j$, then the shape of $\bbF_{Xn}(u) - \bbF_{(x - Y)n}(u)$ rises monotonically to one, then falls, as $u$ increases.  The $x$ that satisfy this condition are $x \geq \max_{i,j} \{X_i + Y_j\} = \max_i X_i + \max_j Y_j$.  To find the minimum of the support, the maximum value of $x$ such that $\bbL_n(x) = 0$, note that the function $\bbF_{Xn}(u) - \bbF_{(x - Y)n}(u)$ is always equal to zero for some $u$, but zero is the maximum only when $\bbF_{Xn}$ first-order stochastically dominates $\bbF_{(x - Y)n}$ and the functions are equal at at least one argument.  Therefore the smallest $x$ in the support of $\bbL_n$ is the smallest $x$ that shifts the $\bbF_{(x-Y)n}$ distribution function enough so that it is dominated by $\bbF_{Xn}$.  This is can be estimated by computing the minimum that is, $\min_k \{\hat{Q}_{Xn}(\tau_k) - \hat{Q}_{(-Y)n}(\tau_k)\}$ over a common set of quantiles $\{\tau_k\}_k$.  Similar logic can also be applied to restrict the size of the grid to as small a search area as possible for the one-sided test experiment below.

We test the null hypothesis that the lower bound of the treatment effect distribution corresponds to the bound associated with two standard normal marginal distributions, which is $L_0(x) = 2 \Phi(x / 2) - 1$ for $x > 0$ and zero otherwise, where $\Phi$ is the CDF of the standard normal distribution.

We used this design to choose a sequence $\{a_n\}$ used to estimate the set of $\epsilon$-maximizers in the estimation of the derivative $\psi'$ and $\{b_n\}$ used to estimate contact sets in one-sided derivative estimates.  We look for the tuning parameters $a_n$ and $b_n$ that provide size as close as possible to the nominal 5\% rejection probability.  We chose among $a_n = c_a \sqrt{\log(\log (n)) / n}$ and $b_n = c_b \log (\log (n)) / \sqrt{n}$ where $n = n_0 + n_1$ and $c_a, c_b \in \{0.5, 1.0, 1.5, \ldots, 5\}$.  Empirical rejection probabilities are reported in Tables~\ref{size_05_tab} and~\ref{size_10_tab}.

\begin{table}[!tbp]
{\footnotesize
\begin{center}
\begin{tabular}{lrrrrrrrrrr}
\hline\hline
\multicolumn{1}{l}{tab}&\multicolumn{1}{c}{$c_a$ = 0.5}&\multicolumn{1}{c}{$c_a$ = 1}&\multicolumn{1}{c}{$c_a$ = 1.5}&\multicolumn{1}{c}{$c_a$ = 2}&\multicolumn{1}{c}{$c_a$ = 2.5}&\multicolumn{1}{c}{$c_a$ = 3}&\multicolumn{1}{c}{$c_a$ = 3.5}&\multicolumn{1}{c}{$c_a$ = 4}&\multicolumn{1}{c}{$c_a$ = 4.5}&\multicolumn{1}{c}{$c_a$ = 5}\tabularnewline
\hline
{\bfseries s. size 100}&&&&&&&&&&\tabularnewline
~~$c_b$ = 0.5&$ 0.8$&$-1.0$&$-1.0$&$-1.3$&$-0.8$&$-1.0$&$-1.4$&$-2.2$&$-1.6$&$-1.8$\tabularnewline
~~$c_b$ = 1&$ 0.8$&$ 0.5$&$-0.8$&$-1.8$&$-1.3$&$-1.9$&$-1.6$&$-2.0$&$-1.4$&$-1.5$\tabularnewline
~~$c_b$ = 1.5&$ 1.6$&$-1.3$&$-1.3$&$-1.7$&$-2.0$&$-2.2$&$-1.7$&$-2.6$&$-1.9$&$-2.1$\tabularnewline
~~$c_b$ = 2&$ 0.6$&$-0.5$&$ 0.2$&$-1.2$&$-2.1$&$-2.0$&$-2.0$&$-0.9$&$-1.2$&$-2.7$\tabularnewline
~~$c_b$ = 2.5&$ 0.7$&$-0.5$&$-1.7$&$-1.5$&$-1.3$&$-2.2$&$-1.8$&$-2.0$&$-1.3$&$-2.4$\tabularnewline
~~$c_b$ = 3&$ 2.0$&$ 1.4$&$-1.7$&$-0.1$&$-1.0$&$-2.2$&$-1.6$&$-2.3$&$-1.7$&$-2.0$\tabularnewline
~~$c_b$ = 3.5&$ 1.2$&$-0.7$&$ 0.1$&$-1.3$&$-0.8$&$-1.4$&$-1.5$&$-1.3$&$-2.0$&$-1.6$\tabularnewline
~~$c_b$ = 4&$ 0.6$&$-1.5$&$-0.3$&$-1.5$&$-2.0$&$-1.9$&$-1.6$&$-2.7$&$-2.0$&$-2.2$\tabularnewline
~~$c_b$ = 4.5&$ 1.0$&$-1.0$&$-1.5$&$-1.4$&$-2.6$&$-2.0$&$-1.8$&$-1.5$&$-2.2$&$-1.9$\tabularnewline
~~$c_b$ = 5&$ 0.4$&$-0.5$&$-0.8$&$-2.4$&$-1.7$&$-0.6$&$-1.7$&$-1.3$&$-1.4$&$-2.8$\tabularnewline
\hline
{\bfseries s. size 500}&&&&&&&&&&\tabularnewline
~~$c_b$ = 0.5&$ 1.4$&$-0.4$&$-1.1$&$-2.1$&$-1.7$&$-2.0$&$-0.6$&$-2.4$&$-2.5$&$-1.7$\tabularnewline
~~$c_b$ = 1&$ 0.4$&$-1.2$&$-0.9$&$-1.1$&$-1.6$&$-1.9$&$-1.8$&$-2.3$&$-1.5$&$-2.4$\tabularnewline
~~$c_b$ = 1.5&$ 1.2$&$-1.8$&$-2.5$&$-1.4$&$-1.1$&$-2.5$&$-1.6$&$-1.4$&$-2.1$&$-1.2$\tabularnewline
~~$c_b$ = 2&$-0.4$&$-1.5$&$-1.8$&$-2.3$&$-2.4$&$-2.8$&$-2.3$&$-2.5$&$-1.8$&$-3.0$\tabularnewline
~~$c_b$ = 2.5&$ 0.9$&$-0.4$&$-1.9$&$-2.4$&$-2.2$&$-1.2$&$-2.4$&$-1.9$&$-2.1$&$-2.2$\tabularnewline
~~$c_b$ = 3&$ 0.5$&$-1.4$&$-1.6$&$-2.6$&$-1.5$&$-2.4$&$-1.8$&$-2.1$&$-2.5$&$-1.8$\tabularnewline
~~$c_b$ = 3.5&$-0.1$&$-0.3$&$-1.3$&$-1.7$&$-1.7$&$-1.5$&$-1.4$&$-1.2$&$-1.8$&$-2.1$\tabularnewline
~~$c_b$ = 4&$ 1.1$&$-0.8$&$-1.4$&$-1.7$&$-1.6$&$-2.6$&$-1.6$&$-2.2$&$-1.8$&$-1.5$\tabularnewline
~~$c_b$ = 4.5&$ 0.7$&$-2.1$&$-0.4$&$-1.7$&$-2.2$&$-1.7$&$-3.0$&$-2.4$&$-1.2$&$-2.7$\tabularnewline
~~$c_b$ = 5&$ 0.5$&$-1.2$&$-0.7$&$-1.4$&$-2.4$&$-2.6$&$-1.8$&$-2.5$&$-1.7$&$-2.4$\tabularnewline
\hline
{\bfseries s. size 1000}&&&&&&&&&&\tabularnewline
~~$c_b$ = 0.5&$ 0.0$&$-1.5$&$-0.8$&$-1.6$&$-1.8$&$-2.3$&$-2.3$&$-2.4$&$-1.7$&$-2.5$\tabularnewline
~~$c_b$ = 1&$ 0.3$&$-0.4$&$-0.9$&$-2.3$&$-2.4$&$-1.5$&$-1.7$&$-1.1$&$-1.3$&$-1.9$\tabularnewline
~~$c_b$ = 1.5&$-0.9$&$-1.2$&$-1.3$&$-1.5$&$-2.1$&$-1.7$&$-1.8$&$-2.4$&$-1.6$&$-2.6$\tabularnewline
~~$c_b$ = 2&$-1.7$&$-0.8$&$-2.3$&$-0.8$&$-2.6$&$-1.7$&$-0.9$&$-2.4$&$-2.6$&$-2.4$\tabularnewline
~~$c_b$ = 2.5&$ 0.7$&$-1.8$&$-1.4$&$-1.9$&$-2.9$&$-2.1$&$-1.8$&$-1.1$&$-2.8$&$-2.5$\tabularnewline
~~$c_b$ = 3&$ 0.1$&$-1.2$&$-1.8$&$-2.2$&$-2.1$&$-1.9$&$-2.2$&$-2.1$&$-1.7$&$-1.6$\tabularnewline
~~$c_b$ = 3.5&$-0.1$&$-1.4$&$-1.2$&$-1.2$&$-0.8$&$-1.9$&$-2.5$&$-2.8$&$-1.2$&$-2.0$\tabularnewline
~~$c_b$ = 4&$-0.1$&$-1.5$&$-1.8$&$-2.2$&$-1.2$&$-2.2$&$-1.8$&$-1.7$&$-2.4$&$-2.0$\tabularnewline
~~$c_b$ = 4.5&$ 1.1$&$-1.2$&$-2.3$&$-2.1$&$-1.5$&$-0.7$&$-3.2$&$-2.0$&$-2.1$&$-1.6$\tabularnewline
~~$c_b$ = 5&$ 0.0$&$-2.0$&$-1.4$&$-3.1$&$-2.1$&$-1.6$&$-2.4$&$-2.4$&$-2.2$&$-1.9$\tabularnewline
\hline
\end{tabular}
\caption{Empirical minus nominal 5\% rejection frequency for a test that $L_{X+Y} \equiv L_0$ for known $L_0$.  Numbers in the table closer to zero are better.  Rows describe tuning parameters for contact set estimation and columns for epsilon-maximizer estimation.\label{size_05_tab}}\end{center}}
\end{table}

\begin{table}[!tbp]
{\footnotesize
\begin{center}
\begin{tabular}{lrrrrrrrrrr}
\hline\hline
\multicolumn{1}{l}{tab}&\multicolumn{1}{c}{$c_a$ = 0.5}&\multicolumn{1}{c}{$c_a$ = 1}&\multicolumn{1}{c}{$c_a$ = 1.5}&\multicolumn{1}{c}{$c_a$ = 2}&\multicolumn{1}{c}{$c_a$ = 2.5}&\multicolumn{1}{c}{$c_a$ = 3}&\multicolumn{1}{c}{$c_a$ = 3.5}&\multicolumn{1}{c}{$c_a$ = 4}&\multicolumn{1}{c}{$c_a$ = 4.5}&\multicolumn{1}{c}{$c_a$ = 5}\tabularnewline
\hline
{\bfseries s. size 100}&&&&&&&&&&\tabularnewline
~~$c_b$ = 0.5&$ 1.1$&$-2.3$&$-1.7$&$-2.8$&$-2.9$&$-2.2$&$-2.3$&$-4.6$&$-2.3$&$-4.0$\tabularnewline
~~$c_b$ = 1&$ 1.8$&$-0.1$&$-1.6$&$-3.2$&$-2.6$&$-3.4$&$-4.0$&$-4.1$&$-3.2$&$-3.2$\tabularnewline
~~$c_b$ = 1.5&$ 1.3$&$-1.9$&$-1.7$&$-2.5$&$-1.8$&$-2.7$&$-2.8$&$-4.6$&$-3.1$&$-4.0$\tabularnewline
~~$c_b$ = 2&$ 0.9$&$-1.6$&$-1.2$&$-2.6$&$-3.2$&$-3.5$&$-2.3$&$-4.1$&$-3.0$&$-5.3$\tabularnewline
~~$c_b$ = 2.5&$ 0.8$&$-1.5$&$-1.8$&$-1.5$&$-2.4$&$-4.1$&$-3.3$&$-3.3$&$-1.9$&$-4.4$\tabularnewline
~~$c_b$ = 3&$ 3.0$&$ 0.9$&$-2.8$&$-2.9$&$-2.4$&$-3.5$&$-3.6$&$-4.5$&$-3.4$&$-4.2$\tabularnewline
~~$c_b$ = 3.5&$ 1.2$&$-1.5$&$-1.9$&$-3.4$&$-1.0$&$-2.4$&$-3.9$&$-3.1$&$-3.1$&$-3.2$\tabularnewline
~~$c_b$ = 4&$ 1.1$&$-1.0$&$-1.6$&$-2.5$&$-3.4$&$-4.6$&$-3.2$&$-5.0$&$-3.8$&$-3.3$\tabularnewline
~~$c_b$ = 4.5&$ 2.1$&$-0.4$&$-3.2$&$-3.0$&$-4.6$&$-3.3$&$-4.0$&$-2.4$&$-4.8$&$-3.3$\tabularnewline
~~$c_b$ = 5&$-0.8$&$-1.5$&$-2.1$&$-4.5$&$-3.2$&$-1.5$&$-3.0$&$-2.4$&$-2.8$&$-4.4$\tabularnewline
\hline
{\bfseries s. size 500}&&&&&&&&&&\tabularnewline
~~$c_b$ = 0.5&$-0.1$&$-0.6$&$-2.3$&$-2.5$&$-4.3$&$-4.2$&$-2.5$&$-4.3$&$-4.2$&$-3.4$\tabularnewline
~~$c_b$ = 1&$ 0.3$&$-2.6$&$-2.2$&$-1.7$&$-2.7$&$-3.2$&$-4.6$&$-4.0$&$-3.3$&$-4.9$\tabularnewline
~~$c_b$ = 1.5&$ 1.4$&$-2.6$&$-2.8$&$-2.1$&$-2.7$&$-3.7$&$-2.9$&$-3.4$&$-4.5$&$-3.4$\tabularnewline
~~$c_b$ = 2&$ 0.5$&$-2.0$&$-2.7$&$-3.5$&$-3.6$&$-4.3$&$-3.2$&$-4.5$&$-3.4$&$-4.3$\tabularnewline
~~$c_b$ = 2.5&$ 1.3$&$-0.5$&$-3.4$&$-3.4$&$-4.2$&$-2.7$&$-4.3$&$-3.2$&$-4.2$&$-3.9$\tabularnewline
~~$c_b$ = 3&$-0.3$&$-3.6$&$-2.4$&$-5.1$&$-3.1$&$-4.0$&$-3.5$&$-2.5$&$-4.0$&$-3.8$\tabularnewline
~~$c_b$ = 3.5&$-0.1$&$-0.6$&$-2.6$&$-3.8$&$-2.6$&$-3.1$&$-2.7$&$-2.6$&$-4.0$&$-3.4$\tabularnewline
~~$c_b$ = 4&$ 1.3$&$-2.2$&$-3.2$&$-2.4$&$-3.2$&$-4.3$&$-3.6$&$-5.2$&$-3.8$&$-3.9$\tabularnewline
~~$c_b$ = 4.5&$ 1.9$&$-3.7$&$-2.0$&$-2.8$&$-4.5$&$-3.7$&$-4.5$&$-3.9$&$-2.5$&$-4.3$\tabularnewline
~~$c_b$ = 5&$ 0.8$&$-2.2$&$-1.6$&$-2.9$&$-4.3$&$-4.1$&$-3.2$&$-4.1$&$-2.9$&$-4.1$\tabularnewline
\hline
{\bfseries s. size 1000}&&&&&&&&&&\tabularnewline
~~$c_b$ = 0.5&$-1.2$&$-3.2$&$-2.0$&$-3.9$&$-3.5$&$-4.3$&$-2.8$&$-4.1$&$-3.9$&$-4.8$\tabularnewline
~~$c_b$ = 1&$ 1.0$&$-1.7$&$-1.7$&$-4.9$&$-3.4$&$-2.8$&$-2.6$&$-3.6$&$-2.1$&$-2.9$\tabularnewline
~~$c_b$ = 1.5&$-1.0$&$-1.9$&$-3.5$&$-3.1$&$-2.9$&$-3.1$&$-3.0$&$-4.1$&$-3.5$&$-5.0$\tabularnewline
~~$c_b$ = 2&$-1.6$&$-2.0$&$-4.3$&$-2.3$&$-3.5$&$-3.6$&$-2.3$&$-5.0$&$-3.9$&$-4.6$\tabularnewline
~~$c_b$ = 2.5&$ 0.5$&$-2.5$&$-2.0$&$-3.0$&$-4.6$&$-4.7$&$-3.0$&$-3.5$&$-5.9$&$-4.6$\tabularnewline
~~$c_b$ = 3&$ 0.5$&$-2.6$&$-3.0$&$-3.5$&$-3.9$&$-3.6$&$-4.0$&$-4.1$&$-2.5$&$-3.7$\tabularnewline
~~$c_b$ = 3.5&$-0.2$&$-1.9$&$-2.0$&$-1.8$&$-1.6$&$-3.8$&$-4.7$&$-4.3$&$-3.0$&$-4.4$\tabularnewline
~~$c_b$ = 4&$-0.5$&$-1.9$&$-2.5$&$-3.6$&$-2.5$&$-3.3$&$-3.5$&$-3.4$&$-3.6$&$-3.6$\tabularnewline
~~$c_b$ = 4.5&$ 1.9$&$-2.2$&$-4.0$&$-3.3$&$-3.3$&$-2.1$&$-4.5$&$-3.8$&$-4.4$&$-2.5$\tabularnewline
~~$c_b$ = 5&$-0.1$&$-2.8$&$-1.8$&$-5.0$&$-3.2$&$-3.0$&$-3.6$&$-4.6$&$-4.5$&$-4.4$\tabularnewline
\hline
\end{tabular}
\caption{Empirical minus nominal 10\% rejection frequency for a test that $L_{X+Y} \equiv L_0$ for known $L_0$.  Numbers in the table closer to zero are better.  Rows describe tuning parameters for contact set estimation and columns for epsilon-maximizer estimation.\label{size_10_tab}}\end{center}}
\end{table}

We decided on $a_n = 0.5 \log(\log(n)) / \sqrt{n}$, since in general smaller $c_a$ resulted in more accurate size.  For contact set estimation we chose $b_n = 3.5 \log(\log(n)) / \sqrt{n}$, which is similar to the sequence used in \citet{LintonSongWhang10}, who concentrated on estimating the contact set in a stochastic dominance experiment.\footnote{They used sequences of the form $c \log\log(\bar{n}) / \sqrt{\bar{n}}$, where $\bar{n} = (n_1 + n_2) / 2$, where they suggested $c$ between 3 and 4.  Using $c = 4$ but adjusting the formula to depend on $n$, the sum of the two samples, we have $4 / \sqrt{2} \approx 2.8$, so there is some justification for setting the constant for $b_n$ lower.  In our simulations there was a slight improvement when using constant 3.5 instead of 3.}

\subsection{Two-sided tests} \label{sec:twosides}

We verify that our proposed method provides accurate coverage probability for uniform confidence bands and power against local alternatives.  We test the null hypothesis that the lower bound of the treatment effect distribution corresponds to the bound associated with two standard normal marginal distributions, which is $L_0(x) = 2 \Phi(x / 2) - 1$ for $x > 0$ and zero otherwise, where $\Phi$ is the CDF of the standard normal distribution.  The upper bound is symmetric so it is sufficient to examine only the lower bound test.  Alternatives are local location-shift alternatives with the mean $\mu_X = 0$ while the mean $\mu_Y$ changes.  In previous experiments we found good power against alternatives with $\mu_Y = k/\sqrt{n}$ for varying $k$, but thanks to referee suggestions we also consider a wider range of local alternative specifications: corresponding to those alternatives used in \citet{HongLi18}, we let $\mu_Y \in \{0, n^{-1}, n^{-2/3}, n^{-1/2}, n^{-1/3}, n^{-1/6}, n^{-1/10}, 2\}$, and also consider these means multiplied by negative 1.  These are tests that should reject for any alternative $\mu_Y \neq 0$, but nontrivial power against all local alternatives cannot be guaranteed for them.  Results from this simulation experiment are collected in Tables~\ref{cband_tab100}, \ref{cband_tab500} and \ref{cband_tab1000}.

\citet{HongLi18} suggest estimating the directional derivative nonparametrically using
\begin{equation*}
  \hat{\lambda}'_{HL} = \frac{\lambda(\bbF_n + \epsilon_n \sqrt{n}(\bbF_n^* - \bbF_n)) - \lambda(\bbF_n)}{\epsilon_n},
\end{equation*}
where $\epsilon_n \rightarrow 0$ and $n^{1/2} \epsilon_n \rightarrow \infty$.  However, it is difficult to use this estimate here, because it is difficult to impose a null hypothesis on the form of the derivative in this experiment.  Specifically, the derivative~\eqref{cband_derivative} relies on $L_0$.  Although $L_0$ does not need to be known in the analytic construction of a derivative estimate, it is needed to estimate the derivative numerically using the above formula.  We addressed this problem in two ways.  First, we imposed the hypothesis that $L_{X+Y} = L_0$, that is, using statistic~\eqref{twosided_ex} exactly as written with the hypothesized $L_0$ in $\hat{\lambda}'_{HL}$.  Second, we used a plug-in estimate of $L_{X+Y}$.  Recalling~\eqref{dep_stat_def}, let
\begin{align*}
  \tilde{\lambda}'_{HL} &= \epsilon_n^{-1} \Bigg( \sup_{x \in \R} |\max\{ \sup_{u \in \R} \Pi^+(\bbF_n + \epsilon_n \sqrt{n}(\bbF_n^* - \bbF_n))(u, x), 0\} - \bbL_n(x)| \\
  {} &\phantom{=} \qquad \qquad \qquad \qquad \qquad \qquad - \sup_{x \in \R} |\max\{ \sup_{u \in \R} \Pi^+(\bbF_n)(u, x), 0\} - \bbL_n(x)| \Bigg) \\
  {} &= \epsilon_n^{-1} \sup_{x \in \R} |\max\{ \sup_{u \in \R} \Pi^+(\bbF_n + \epsilon_n \sqrt{n}(\bbF_n^* - \bbF_n))(u, x), 0\} - \bbL_n(x)|,
\end{align*}
where the second equality arises from the fact that $\bbL_n(x) = \max\{ \sup_{u \in \R} \Pi^+(\bbF_n)(u, x), 0\}$ by definition.  Neither of these options is very satisfactory, since $L_0$ is unlikely to be known and using $\bbL_n$ as in $\tilde{\lambda}_{HL}$ results in a statistic that appears to estimate the derivative inside the $L_\infty$ norm.  The method proposed in \citet{HongLi18} requires $\epsilon_n$ to be larger than $n^{-1/2}$, so size distortions for smaller choices of $\epsilon_n$ may be expected.  However, those $\epsilon_n$ were also considered in their paper so we included them here.

\begin{landscape}\begin{table}[!tbp]
{\footnotesize
\begin{center}
\begin{tabular}{lrcrrrrcrrrr}
\hline\hline
\multicolumn{1}{l}{\bfseries Location}&\multicolumn{1}{c}{\bfseries Analytic}&\multicolumn{1}{c}{\bfseries }&\multicolumn{4}{c}{\bfseries Numeric (null known)}&\multicolumn{1}{c}{\bfseries }&\multicolumn{4}{c}{\bfseries Numeric (null estimated)}\tabularnewline
\cline{2-2} \cline{4-7} \cline{9-12}
\multicolumn{1}{l}{}&\multicolumn{1}{c}{}&\multicolumn{1}{c}{}&\multicolumn{1}{c}{$\epsilon_n = n^{-1/6}$}&\multicolumn{1}{c}{$\epsilon_n = n^{-1/3}$}&\multicolumn{1}{c}{$\epsilon_n = n^{-1/2}$}&\multicolumn{1}{c}{$\epsilon_n = n^{-1}$}&\multicolumn{1}{c}{}&\multicolumn{1}{c}{$\epsilon_n = n^{-1/6}$}&\multicolumn{1}{c}{$\epsilon_n = n^{-1/3}$}&\multicolumn{1}{c}{$\epsilon_n = n^{-1/2}$}&\multicolumn{1}{c}{$\epsilon_n = n^{-1}$}\tabularnewline
\hline
{\bfseries }&&&&&&&&&&&\tabularnewline
~~$-2$&$1.000$&&$1.000$&$1.000$&$1.000$&$1.000$&&$1.000$&$1.000$&$1.000$&$1.000$\tabularnewline
~~$-n^{-1/10}$&$0.968$&&$0.976$&$0.984$&$0.992$&$0.997$&&$0.967$&$0.977$&$0.981$&$0.973$\tabularnewline
~~$-n^{-1/6}$&$0.816$&&$0.841$&$0.893$&$0.938$&$0.981$&&$0.806$&$0.842$&$0.860$&$0.824$\tabularnewline
~~$-n^{-1/3}$&$0.301$&&$0.335$&$0.441$&$0.568$&$0.739$&&$0.283$&$0.330$&$0.367$&$0.297$\tabularnewline
~~$-n^{-1/2}$&$0.130$&&$0.155$&$0.245$&$0.363$&$0.570$&&$0.118$&$0.151$&$0.173$&$0.131$\tabularnewline
~~$-n^{-2/3}$&$0.093$&&$0.110$&$0.180$&$0.294$&$0.500$&&$0.081$&$0.108$&$0.120$&$0.092$\tabularnewline
~~$-n^{-1}$&$0.072$&&$0.085$&$0.161$&$0.266$&$0.482$&&$0.063$&$0.088$&$0.098$&$0.069$\tabularnewline
\hline
{\bfseries Null is true}&&&&&&&&&&&\tabularnewline
~~$0$&$0.058$&&$0.072$&$0.141$&$0.273$&$0.483$&&$0.053$&$0.069$&$0.077$&$0.059$\tabularnewline
\hline
{\bfseries }&&&&&&&&&&&\tabularnewline
~~$n^{-1}$&$0.048$&&$0.061$&$0.120$&$0.246$&$0.433$&&$0.043$&$0.057$&$0.066$&$0.049$\tabularnewline
~~$n^{-2/3}$&$0.037$&&$0.052$&$0.118$&$0.237$&$0.432$&&$0.030$&$0.046$&$0.055$&$0.035$\tabularnewline
~~$n^{-1/2}$&$0.030$&&$0.050$&$0.132$&$0.266$&$0.427$&&$0.025$&$0.036$&$0.045$&$0.030$\tabularnewline
~~$n^{-1/3}$&$0.017$&&$0.057$&$0.200$&$0.322$&$0.419$&&$0.016$&$0.020$&$0.024$&$0.019$\tabularnewline
~~$n^{-1/6}$&$0.242$&&$0.497$&$0.750$&$0.850$&$0.836$&&$0.222$&$0.283$&$0.320$&$0.254$\tabularnewline
~~$n^{-1/10}$&$0.662$&&$0.863$&$0.947$&$0.972$&$0.964$&&$0.646$&$0.714$&$0.740$&$0.665$\tabularnewline
~~$2$&$1.000$&&$1.000$&$1.000$&$1.000$&$1.000$&&$1.000$&$1.000$&$1.000$&$1.000$\tabularnewline
\hline
\end{tabular}
\caption{Empirical rejection frequencies for a test that $L_{X+Y} \equiv L_0$ for known $L_0$.  Analytic derivative estimates are described in the text, while Numeric derivative estimates use the method suggested by \citet{HongLi18}.  Each row label gives the location parameter of one of the distributions, while the null assumes both parameters are zero.  Samples of size 100, 499 bootstrap repetitions per test, $1000$ simulation repetitions.\label{cband_tab100}}\end{center}}
\end{table}\end{landscape}

\begin{landscape}\begin{table}[!tbp]
{\footnotesize
\begin{center}
\begin{tabular}{lrcrrrrcrrrr}
\hline\hline
\multicolumn{1}{l}{\bfseries Location}&\multicolumn{1}{c}{\bfseries Analytic}&\multicolumn{1}{c}{\bfseries }&\multicolumn{4}{c}{\bfseries Numeric (null known)}&\multicolumn{1}{c}{\bfseries }&\multicolumn{4}{c}{\bfseries Numeric (null estimated)}\tabularnewline
\cline{2-2} \cline{4-7} \cline{9-12}
\multicolumn{1}{l}{}&\multicolumn{1}{c}{}&\multicolumn{1}{c}{}&\multicolumn{1}{c}{$\epsilon_n = n^{-1/6}$}&\multicolumn{1}{c}{$\epsilon_n = n^{-1/3}$}&\multicolumn{1}{c}{$\epsilon_n = n^{-1/2}$}&\multicolumn{1}{c}{$\epsilon_n = n^{-1}$}&\multicolumn{1}{c}{}&\multicolumn{1}{c}{$\epsilon_n = n^{-1/6}$}&\multicolumn{1}{c}{$\epsilon_n = n^{-1/3}$}&\multicolumn{1}{c}{$\epsilon_n = n^{-1/2}$}&\multicolumn{1}{c}{$\epsilon_n = n^{-1}$}\tabularnewline
\hline
{\bfseries }&&&&&&&&&&&\tabularnewline
~~$-2$&$1.000$&&$1.000$&$1.000$&$1.000$&$1.000$&&$1.000$&$1.000$&$1.000$&$1.000$\tabularnewline
~~$-n^{-1/10}$&$1.000$&&$1.000$&$1.000$&$1.000$&$1.000$&&$1.000$&$1.000$&$1.000$&$1.000$\tabularnewline
~~$-n^{-1/6}$&$0.997$&&$0.997$&$0.998$&$1.000$&$1.000$&&$0.995$&$0.997$&$0.998$&$0.997$\tabularnewline
~~$-n^{-1/3}$&$0.383$&&$0.388$&$0.483$&$0.655$&$0.813$&&$0.356$&$0.416$&$0.438$&$0.383$\tabularnewline
~~$-n^{-1/2}$&$0.106$&&$0.114$&$0.182$&$0.339$&$0.549$&&$0.096$&$0.122$&$0.139$&$0.108$\tabularnewline
~~$-n^{-2/3}$&$0.050$&&$0.056$&$0.103$&$0.256$&$0.465$&&$0.041$&$0.064$&$0.073$&$0.051$\tabularnewline
~~$-n^{-1}$&$0.061$&&$0.065$&$0.124$&$0.272$&$0.452$&&$0.054$&$0.067$&$0.074$&$0.060$\tabularnewline
\hline
{\bfseries Null is true}&&&&&&&&&&&\tabularnewline
~~$0$&$0.058$&&$0.060$&$0.108$&$0.257$&$0.453$&&$0.051$&$0.066$&$0.072$&$0.055$\tabularnewline
\hline
{\bfseries }&&&&&&&&&&&\tabularnewline
~~$n^{-1}$&$0.050$&&$0.056$&$0.111$&$0.240$&$0.430$&&$0.042$&$0.057$&$0.069$&$0.050$\tabularnewline
~~$n^{-2/3}$&$0.040$&&$0.041$&$0.100$&$0.215$&$0.397$&&$0.034$&$0.049$&$0.059$&$0.035$\tabularnewline
~~$n^{-1/2}$&$0.035$&&$0.042$&$0.113$&$0.241$&$0.362$&&$0.027$&$0.039$&$0.044$&$0.034$\tabularnewline
~~$n^{-1/3}$&$0.064$&&$0.116$&$0.281$&$0.455$&$0.494$&&$0.057$&$0.074$&$0.088$&$0.067$\tabularnewline
~~$n^{-1/6}$&$0.904$&&$0.953$&$0.989$&$0.996$&$0.991$&&$0.896$&$0.924$&$0.933$&$0.909$\tabularnewline
~~$n^{-1/10}$&$1.000$&&$1.000$&$1.000$&$1.000$&$1.000$&&$1.000$&$1.000$&$1.000$&$1.000$\tabularnewline
~~$2$&$1.000$&&$1.000$&$1.000$&$1.000$&$1.000$&&$1.000$&$1.000$&$1.000$&$1.000$\tabularnewline
\hline
\end{tabular}
\caption{Empirical rejection frequencies for a test that $L_{X+Y} \equiv L_0$ for known $L_0$.  Analytic derivative estimates are described in the text, while Numeric derivative estimates use the method suggested by \citet{HongLi18}.  Each row label gives the location parameter of one of the distributions, while the null assumes both parameters are zero.  Samples of size 500, 999 bootstrap repetitions per test, $1000$ simulation repetitions.\label{cband_tab500}}\end{center}}
\end{table}\end{landscape}

\begin{landscape}\begin{table}[!tbp]
{\footnotesize
\begin{center}
\begin{tabular}{lrcrrrrcrrrr}
\hline\hline
\multicolumn{1}{l}{\bfseries Location}&\multicolumn{1}{c}{\bfseries Analytic}&\multicolumn{1}{c}{\bfseries }&\multicolumn{4}{c}{\bfseries Numeric (null known)}&\multicolumn{1}{c}{\bfseries }&\multicolumn{4}{c}{\bfseries Numeric (null estimated)}\tabularnewline
\cline{2-2} \cline{4-7} \cline{9-12}
\multicolumn{1}{l}{}&\multicolumn{1}{c}{}&\multicolumn{1}{c}{}&\multicolumn{1}{c}{$\epsilon_n = n^{-1/6}$}&\multicolumn{1}{c}{$\epsilon_n = n^{-1/3}$}&\multicolumn{1}{c}{$\epsilon_n = n^{-1/2}$}&\multicolumn{1}{c}{$\epsilon_n = n^{-1}$}&\multicolumn{1}{c}{}&\multicolumn{1}{c}{$\epsilon_n = n^{-1/6}$}&\multicolumn{1}{c}{$\epsilon_n = n^{-1/3}$}&\multicolumn{1}{c}{$\epsilon_n = n^{-1/2}$}&\multicolumn{1}{c}{$\epsilon_n = n^{-1}$}\tabularnewline
\hline
{\bfseries }&&&&&&&&&&&\tabularnewline
~~$-2$&$1.000$&&$1.000$&$1.000$&$1.000$&$1.000$&&$1.000$&$1.000$&$1.000$&$1.000$\tabularnewline
~~$-n^{-1/10}$&$1.000$&&$1.000$&$1.000$&$1.000$&$1.000$&&$1.000$&$1.000$&$1.000$&$1.000$\tabularnewline
~~$-n^{-1/6}$&$1.000$&&$1.000$&$1.000$&$1.000$&$1.000$&&$1.000$&$1.000$&$1.000$&$1.000$\tabularnewline
~~$-n^{-1/3}$&$0.432$&&$0.434$&$0.555$&$0.708$&$0.852$&&$0.411$&$0.460$&$0.506$&$0.438$\tabularnewline
~~$-n^{-1/2}$&$0.108$&&$0.107$&$0.172$&$0.349$&$0.536$&&$0.095$&$0.115$&$0.132$&$0.107$\tabularnewline
~~$-n^{-2/3}$&$0.072$&&$0.074$&$0.135$&$0.271$&$0.450$&&$0.061$&$0.079$&$0.088$&$0.070$\tabularnewline
~~$-n^{-1}$&$0.041$&&$0.044$&$0.098$&$0.267$&$0.428$&&$0.037$&$0.048$&$0.057$&$0.041$\tabularnewline
\hline
{\bfseries Null is true}&&&&&&&&&&&\tabularnewline
~~$0$&$0.046$&&$0.045$&$0.107$&$0.245$&$0.435$&&$0.037$&$0.053$&$0.060$&$0.045$\tabularnewline
\hline
{\bfseries }&&&&&&&&&&&\tabularnewline
~~$n^{-1}$&$0.051$&&$0.053$&$0.116$&$0.259$&$0.419$&&$0.046$&$0.058$&$0.068$&$0.050$\tabularnewline
~~$n^{-2/3}$&$0.046$&&$0.049$&$0.108$&$0.239$&$0.408$&&$0.039$&$0.050$&$0.060$&$0.041$\tabularnewline
~~$n^{-1/2}$&$0.027$&&$0.030$&$0.103$&$0.264$&$0.387$&&$0.024$&$0.032$&$0.035$&$0.026$\tabularnewline
~~$n^{-1/3}$&$0.089$&&$0.131$&$0.328$&$0.520$&$0.539$&&$0.074$&$0.099$&$0.115$&$0.087$\tabularnewline
~~$n^{-1/6}$&$0.994$&&$0.998$&$0.999$&$1.000$&$0.999$&&$0.993$&$0.995$&$0.996$&$0.995$\tabularnewline
~~$n^{-1/10}$&$1.000$&&$1.000$&$1.000$&$1.000$&$1.000$&&$1.000$&$1.000$&$1.000$&$1.000$\tabularnewline
~~$2$&$1.000$&&$1.000$&$1.000$&$1.000$&$1.000$&&$1.000$&$1.000$&$1.000$&$1.000$\tabularnewline
\hline
\end{tabular}
\caption{Empirical rejection frequencies for a test that $L_{X+Y} \equiv L_0$ for known $L_0$.  Analytic derivative estimates are described in the text, while Numeric derivative estimates use the method suggested by \citet{HongLi18}.  Each row label gives the location parameter of one of the distributions, while the null assumes both parameters are zero.  Samples of size 1000, 1999 bootstrap repetitions per test, $1000$ simulation repetitions.\label{cband_tab1000}}\end{center}}
\end{table}\end{landscape}

The tables indicate power loss only for alternatives that are very small and positive (for locations of size $n^{-2/3}$ and $n^{1/2}$).  The numerically estimated reference distribution with known $L_0$ was usually outperformed by the analytically estimated distribution and that using $\bbL_n$ in the place of $L_0$.  When using a known $L_0$, the choice of $\epsilon_n$ is crucial.  On the other hand, the numerical derivatives do well when $L$ is estimated and $\epsilon_n$ converges either slowly or very quickly to zero~--- surprisingly, the last column, with estimated null and $\epsilon_n = n^{-1}$ is closest to the behavior of the test with analytically estimated reference distribution.  The analytic estimate does not require knowledge of $L_0$ and performs as well as the best-performing numeric estimates.

\subsection{One-sided tests} \label{sec:oneside}

In a second experiment we simulate uniformly distributed samples and test the condition that implies stochastic dominance of $F_{\Delta_A}$ over $F_{\Delta_B}$ using statistic~\eqref{eq:sd_test_stat}.  A control sample is made up of uniform observations on the unit interval, as is treatment $B$, while treatment $A$ is distributed uniformly on $[\mu, \mu + 1]$.  We test the hypothesis that $L_A - U_B \leq 0$.  When $\mu = -1$, $L_A - U_B \equiv 0$ so that value of $\mu = -1$ represents a least-favorable null hypothesis.\footnote{It can be verified that if $X \sim \text{Unif}[\mu, \mu + 1]$ and $Y \sim \text{Unif}[0, 1]$, then $L_{X-Y}$ is the CDF of a $\text{Unif}[\mu, \mu+1]$ random variable and $U_{X-Y}$ is the CDF of a $\text{Unif}[\mu-1, \mu]$ random variable.}  When $\mu > -1$, the null hypothesis is satisfied with a strict inequality and should not be rejected, while for $\mu < -1$ the null should be rejected.  We examine local alternatives around the central value $\mu = -1$, which is normalized to zero in Tables~\ref{sd_tab100}-\ref{sd_tab1000}.  We test size and power against local alternatives for samples of size 100, 500 and 1000 with respectively 499, 999 and 1999 bootstrap repetitions for each sample size and a grid with increment size $0.02$.  1,000 simulation repetitions were used for each sample size.  In this experiment the method of \citet{HongLi18} can be implemented without issue.  We compare the analytically-estimated derivative using the same $a_n$ and $b_n$ and numerically-estimated derivatives using the same choices for $\epsilon_n$ as in the previous experiment.

\begin{table}[!tbp]
\begin{center}
\begin{tabular}{lrcrrrr}
\hline\hline
\multicolumn{1}{l}{\bfseries Location}&\multicolumn{1}{c}{\bfseries Analytic}&\multicolumn{1}{c}{\bfseries }&\multicolumn{4}{c}{\bfseries Numeric}\tabularnewline
\cline{2-2} \cline{4-7}
\multicolumn{1}{l}{}&\multicolumn{1}{c}{}&\multicolumn{1}{c}{}&\multicolumn{1}{c}{$\epsilon_n = n^{-1/6}$}&\multicolumn{1}{c}{$\epsilon_n = n^{-1/3}$}&\multicolumn{1}{c}{$\epsilon_n = n^{-1/2}$}&\multicolumn{1}{c}{$\epsilon_n = n^{-1}$}\tabularnewline
\hline
{\bfseries }&&&&&&\tabularnewline
~~$-2$&$1.000$&&$1.000$&$1.000$&$1.000$&$1.000$\tabularnewline
~~$-n^{-1/10}$&$1.000$&&$1.000$&$1.000$&$1.000$&$1.000$\tabularnewline
~~$-n^{-1/6}$&$1.000$&&$1.000$&$1.000$&$1.000$&$1.000$\tabularnewline
~~$-n^{-1/3}$&$1.000$&&$1.000$&$1.000$&$1.000$&$1.000$\tabularnewline
~~$-n^{-1/2}$&$0.645$&&$0.479$&$0.819$&$0.984$&$1.000$\tabularnewline
~~$-n^{-2/3}$&$0.266$&&$0.179$&$0.410$&$0.721$&$0.961$\tabularnewline
~~$-n^{-1}$&$0.114$&&$0.070$&$0.199$&$0.438$&$0.812$\tabularnewline
\hline
{\bfseries Null bndry.}&&&&&&\tabularnewline
~~$0$&$0.100$&&$0.059$&$0.193$&$0.407$&$0.778$\tabularnewline
\hline
{\bfseries }&&&&&&\tabularnewline
~~$n^{-1}$&$0.070$&&$0.043$&$0.129$&$0.355$&$0.708$\tabularnewline
~~$n^{-2/3}$&$0.032$&&$0.018$&$0.059$&$0.177$&$0.432$\tabularnewline
~~$n^{-1/2}$&$0.005$&&$0.004$&$0.011$&$0.044$&$0.128$\tabularnewline
~~$n^{-1/3}$&$0.000$&&$0.000$&$0.000$&$0.000$&$0.003$\tabularnewline
~~$n^{-1/6}$&$0.000$&&$0.000$&$0.000$&$0.000$&$0.000$\tabularnewline
~~$n^{-1/10}$&$0.000$&&$0.000$&$0.000$&$0.000$&$0.000$\tabularnewline
~~$2$&$0.000$&&$0.000$&$0.000$&$0.000$&$0.000$\tabularnewline
\hline
\end{tabular}
\caption{Empirical rejection frequencies for a test that $L_A \leq U_B$.  Analytic derivative estimates are described in the text, while Numeric derivative estimates use the method suggested by \citet{HongLi18}.  Each row label gives the location parameter of one of the distributions, while the null assumes both parameters are zero.  Samples of size 100, 499 bootstrap repetitions per test, $1000$ simulation repetitions.\label{sd_tab100}}\end{center}
\end{table}

\begin{table}[!tbp]
\begin{center}
\begin{tabular}{lrcrrrr}
\hline\hline
\multicolumn{1}{l}{\bfseries Location}&\multicolumn{1}{c}{\bfseries Analytic}&\multicolumn{1}{c}{\bfseries }&\multicolumn{4}{c}{\bfseries Numeric}\tabularnewline
\cline{2-2} \cline{4-7}
\multicolumn{1}{l}{}&\multicolumn{1}{c}{}&\multicolumn{1}{c}{}&\multicolumn{1}{c}{$\epsilon_n = n^{-1/6}$}&\multicolumn{1}{c}{$\epsilon_n = n^{-1/3}$}&\multicolumn{1}{c}{$\epsilon_n = n^{-1/2}$}&\multicolumn{1}{c}{$\epsilon_n = n^{-1}$}\tabularnewline
\hline
{\bfseries }&&&&&&\tabularnewline
~~$-2$&$1.000$&&$1.000$&$1.000$&$1.000$&$1.000$\tabularnewline
~~$-n^{-1/10}$&$1.000$&&$1.000$&$1.000$&$1.000$&$1.000$\tabularnewline
~~$-n^{-1/6}$&$1.000$&&$1.000$&$1.000$&$1.000$&$1.000$\tabularnewline
~~$-n^{-1/3}$&$1.000$&&$1.000$&$1.000$&$1.000$&$1.000$\tabularnewline
~~$-n^{-1/2}$&$0.668$&&$0.526$&$0.781$&$0.983$&$1.000$\tabularnewline
~~$-n^{-2/3}$&$0.172$&&$0.099$&$0.245$&$0.570$&$0.943$\tabularnewline
~~$-n^{-1}$&$0.090$&&$0.057$&$0.140$&$0.363$&$0.802$\tabularnewline
\hline
{\bfseries Null bndry.}&&&&&&\tabularnewline
~~$0$&$0.085$&&$0.049$&$0.146$&$0.358$&$0.777$\tabularnewline
\hline
{\bfseries }&&&&&&\tabularnewline
~~$n^{-1}$&$0.063$&&$0.037$&$0.137$&$0.340$&$0.751$\tabularnewline
~~$n^{-2/3}$&$0.029$&&$0.013$&$0.053$&$0.162$&$0.482$\tabularnewline
~~$n^{-1/2}$&$0.004$&&$0.002$&$0.005$&$0.029$&$0.116$\tabularnewline
~~$n^{-1/3}$&$0.000$&&$0.000$&$0.000$&$0.000$&$0.000$\tabularnewline
~~$n^{-1/6}$&$0.000$&&$0.000$&$0.000$&$0.000$&$0.000$\tabularnewline
~~$n^{-1/10}$&$0.000$&&$0.000$&$0.000$&$0.000$&$0.000$\tabularnewline
~~$2$&$0.000$&&$0.000$&$0.000$&$0.000$&$0.000$\tabularnewline
\hline
\end{tabular}
\caption{Empirical rejection frequencies for a test that $L_A \leq U_B$.  Analytic derivative estimates are described in the text, while Numeric derivative estimates use the method suggested by \citet{HongLi18}.  Each row label gives the location parameter of one of the distributions, while the null assumes both parameters are zero.  Samples of size 500, 999 bootstrap repetitions per test, $1000$ simulation repetitions.\label{sd_tab500}}\end{center}
\end{table}

\begin{table}[!tbp]
\begin{center}
\begin{tabular}{lrcrrrr}
\hline\hline
\multicolumn{1}{l}{\bfseries Location}&\multicolumn{1}{c}{\bfseries Analytic}&\multicolumn{1}{c}{\bfseries }&\multicolumn{4}{c}{\bfseries Numeric}\tabularnewline
\cline{2-2} \cline{4-7}
\multicolumn{1}{l}{}&\multicolumn{1}{c}{}&\multicolumn{1}{c}{}&\multicolumn{1}{c}{$\epsilon_n = n^{-1/6}$}&\multicolumn{1}{c}{$\epsilon_n = n^{-1/3}$}&\multicolumn{1}{c}{$\epsilon_n = n^{-1/2}$}&\multicolumn{1}{c}{$\epsilon_n = n^{-1}$}\tabularnewline
\hline
{\bfseries }&&&&&&\tabularnewline
~~$-2$&$1.000$&&$1.000$&$1.000$&$1.000$&$1.000$\tabularnewline
~~$-n^{-1/10}$&$1.000$&&$1.000$&$1.000$&$1.000$&$1.000$\tabularnewline
~~$-n^{-1/6}$&$1.000$&&$1.000$&$1.000$&$1.000$&$1.000$\tabularnewline
~~$-n^{-1/3}$&$1.000$&&$1.000$&$1.000$&$1.000$&$1.000$\tabularnewline
~~$-n^{-1/2}$&$0.653$&&$0.528$&$0.776$&$0.978$&$1.000$\tabularnewline
~~$-n^{-2/3}$&$0.169$&&$0.104$&$0.261$&$0.559$&$0.932$\tabularnewline
~~$-n^{-1}$&$0.084$&&$0.054$&$0.135$&$0.355$&$0.793$\tabularnewline
\hline
{\bfseries Null bndry.}&&&&&&\tabularnewline
~~$0$&$0.067$&&$0.031$&$0.107$&$0.334$&$0.730$\tabularnewline
\hline
{\bfseries }&&&&&&\tabularnewline
~~$n^{-1}$&$0.069$&&$0.031$&$0.110$&$0.305$&$0.729$\tabularnewline
~~$n^{-2/3}$&$0.033$&&$0.017$&$0.054$&$0.167$&$0.518$\tabularnewline
~~$n^{-1/2}$&$0.002$&&$0.001$&$0.005$&$0.026$&$0.102$\tabularnewline
~~$n^{-1/3}$&$0.000$&&$0.000$&$0.000$&$0.000$&$0.000$\tabularnewline
~~$n^{-1/6}$&$0.000$&&$0.000$&$0.000$&$0.000$&$0.000$\tabularnewline
~~$n^{-1/10}$&$0.000$&&$0.000$&$0.000$&$0.000$&$0.000$\tabularnewline
~~$2$&$0.000$&&$0.000$&$0.000$&$0.000$&$0.000$\tabularnewline
\hline
\end{tabular}
\caption{Empirical rejection frequencies for a test that $L_A \leq U_B$.  Analytic derivative estimates are described in the text, while Numeric derivative estimates use the method suggested by \citet{HongLi18}.  Each row label gives the location parameter of one of the distributions, while the null assumes both parameters are zero.  Samples of size 1000, 1999 bootstrap repetitions per test, $1000$ simulation repetitions.\label{sd_tab1000}}\end{center}
\end{table}

The reference distribution estimated analytically performs well in terms of size and power, although not as well as in the previous experiment~--- the empirical rejection probability is higher than expected at the boundary of the null region.  The largest choice of $\epsilon_n$ appears to work best, although this leads to a conservative test for the largest sample size.  The analytic test with sample size 1,000 show higher rejection rate for a very small location shift towards the interior of the null region (location $n^{-1}$ in the table) than on the boundary of the null region.  Theorem~\ref{thm:uniform_size_control} predicts that this should disappear asymptotically.

\newpage
\section{Proofs} \label{appn}

\begin{lem} \label{lem:saddle}
  Suppose that $U \subseteq \R^{d_U}$ and $X \subseteq \R^{d_X}$ and $f, h \in \ell^\infty(\gr A)$.  Let $A: X \rightrightarrows U$ be non-empty-valued and define $\sigma: \ell^\infty(\gr A) \rightarrow \R$ by $\sigma(f) = \sup_{x \in X} \inf_{u \in A(x)} f(u, x)$.  Then $\sigma$ is Hadamard directionally differentiable and
  \begin{equation*} 
    \sigma_f'(h) = \lim_{\delta \rightarrow 0^+} \sup_{x \in X_f^\sigma(\delta)} \lim_{\epsilon \rightarrow 0^+} \inf_{u \in U_{(-f)}(x, \epsilon)} h(u, x),
  \end{equation*}
  where $U_f$ and $X_f^\sigma$ are defined in~\eqref{marginal_epsmax_def} and~\eqref{Xfsigma_def}.
\end{lem}

\begin{proof}[Proof of Lemma~\ref{lem:saddle}]
  It can be verified using that $|\sigma(f) - \sigma(g)| \leq \|f - g\|_\infty$, so we can focus on G\^ateaux differentiability for a fixed $f$ and $h$ \citep[Proposition 3.5]{Shapiro90}.  Fix some $f, h \in \ell^\infty(\gr A)$ and $t_n \rightarrow 0^+$.

  Start with an upper bound for the scaled difference $(\sigma(f + t_nh) - \sigma(f)) / t_n$, which, using $s_n = t_n^{-1}$ may be rewritten as $\sigma(s_nf + h) - s_n\sigma(f)$.  Consider the inner optimization problem in $u$.  For any $x \in X$ and $\epsilon > 0$, by definition of $U_{(-f)}$ (which collects near-minimizers in $u$ of $f$) there exists a $u_\epsilon \in U_{(-f)}(x, \epsilon)$ such that $h(u_\epsilon, x) \leq \inf_{u \in U_{(-f)}(x, \epsilon)} h(u, x) + \epsilon$ and $f(u_\epsilon, x) \leq \inf_{u \in A(x)} f(u, x) + \epsilon$ and therefore
  \begin{align*}
    \inf_{u \in U_{(-f)}(x, \epsilon)} h(u, x) &\geq h(u_\epsilon, x) - \epsilon \\
    {} &= s_n f(u_\epsilon, x) + h(u_\epsilon, x) - s_n f(u_\epsilon, x) - \epsilon \\
    {} &\geq \inf_{u \in A(x)} (s_n f + h)(u, x) - s_n \inf_{u \in A(x)} f(u, x) + s_n \left( \inf_{u \in A(x)} f(u, x) - f(u_\epsilon, x) \right) - \epsilon \\
    {} &\geq \inf_{u \in A(x)} (s_n f + h)(u, x) - s_n \inf_{u \in A(x)} f(u, x) - (s_n + 1) \epsilon.
  \end{align*}
  Thus, for each $n$, for any $x \in X$,
  \begin{equation} \label{ineq_inside}
    \inf_{u \in A(x)} (s_n f + h)(u, x) - s_n \inf_{u \in A(x)} f(u, x) \leq \lim_{\epsilon \rightarrow 0^+} \inf_{u \in U_{(-f)}(x, \epsilon)} h(u, x).
  \end{equation}
  Next consider the outer optimization problem in $x$: for each $x \in X$, we may write
  \begin{multline*}
    \inf_{u \in A(x)} (s_n f + h)(u, x) - s_n \sigma(f) = \left( \inf_{u \in A(x)} (s_n f + h)(u, x) - s_n \inf_{u \in A(x)} f(u, x) \right) \\
    {} + s_n \left( \inf_{u \in A(x)} f(u, x) - \sigma(f) \right) \\
    {} \leq \lim_{\epsilon \rightarrow 0^+} \inf_{u \in U_{(-f)}(x, \epsilon)} h(u, x) + s_n \left( \inf_{u \in A(x)} f(u, x) - \sigma(f) \right),
  \end{multline*}
  where the inequality comes from~\eqref{ineq_inside}.  Consider this inequality for some $x' \notin X_f^\sigma(\delta)$, for any $\delta > 0$:
  \begin{align*}
    \inf_{u \in A(x')} (s_n f + h)(u, x') - s_n \sigma(f) &\leq \lim_{\epsilon \rightarrow 0^+} \inf_{u \in U_{(-f)}(x', \epsilon)} h(u, x') + s_n \left( \inf_{u \in A(x')} f(u, x') - \sigma(f) \right) \\
    {} &\leq \lim_{\epsilon \rightarrow 0^+} \inf_{u \in U_{(-f)}(x', \epsilon)} h(u, x') - s_n \delta.
  \end{align*}
  Recall that $s_n = t_n^{-1}$ so that $s_n$ diverges as $n \rightarrow \infty$, and that $x'$ that are not $\delta$-maximinimizers cannot be optimal.  Then for any $\delta > 0$,
  \begin{align*}
    \limsup_{n \rightarrow \infty} \left( \sigma (s_n f + h) - s_n \sigma(f) \right) &= \limsup_{n \rightarrow \infty} \left( \sup_{x \in X_f^\sigma(\delta)} \inf_{u \in A(x)} (s_n f + h)(u, x) - s_n \sigma(f) \right) \\
    {} &\leq \sup_{x \in X_f^\sigma(\delta)} \left( \lim_{\epsilon \rightarrow 0^+} \inf_{u \in U_{(-f)}(x, \epsilon)} h(u, x) \right)
  \end{align*}
  and therefore this inequality holds as $\delta \rightarrow 0$.

  To obtain a lower bound, start again with the inner problem.  For any $x \in X$, choose an $\epsilon > 0$ and note that for any $u' \in A(x) \backslash U_{(-f)}(x, \epsilon)$,
  \begin{equation*}
    s_n f(u', x) + h(u', x) - s_n \inf_{u \in A(x)} f(u, x) \geq \inf_{u \in A(x)} h(u, x) + s_n \epsilon.
  \end{equation*}
  Therefore, again recalling that $s_n$ diverges with $n$, we may restrict attention to the $\epsilon$-minimizers (in $u$) for each $x \in X$:
  \begin{align}
    \liminf_{n \rightarrow \infty} \bigg( \inf_{u \in A(x)} (s_n f + h)(u, x) &- s_n \inf_{u \in A(x)} f(u, x) \bigg) \notag \\
    {} &= \liminf_{n \rightarrow \infty} \left( \inf_{u \in U_{(-f)}(x, \epsilon)} (s_n f + h)(u, x) - s_n \inf_{u \in A(x)} f(u, x) \right) \label{UisAinlimit} \\
    {} &\geq \inf_{u \in U_{(-f)}(x, \epsilon)} h(u, x) \label{lowerbnd_inner}
  \end{align}
  and this holds as $\epsilon \rightarrow 0^+$ as well.

  Consider again the outer maximization problem.  For any $\delta > 0$, define
  \begin{equation*}
    \tilde{h}(\delta) = \sup_{x \in X_f^\sigma(\delta)} \lim_{\epsilon \rightarrow 0^+} \inf_{u \in U_{(-f)}(x, \epsilon)} h(u, x).
  \end{equation*}
  For each $\delta$ there is an $x_\delta \in X_f^\sigma(\delta)$ such that
  \begin{equation*}
    \lim_{\epsilon \rightarrow 0^+} \inf_{u \in U_{(-f)}(x_\delta, \epsilon)} h(u, x_\delta) \geq \tilde{h}(\delta) - \delta
  \end{equation*}
  and by construction,
  \begin{align*}
    \lim_{\epsilon \rightarrow 0^+} \inf_{u \in U_{(-f)}(x_\delta, \epsilon)} f(u, x_\delta) &\geq \inf_{u \in A(x_\delta)} f(u, x_\delta) \\
    {} &\geq \sup_{x \in X} \inf_{u \in A(x)} f(u, x) - \delta = \sigma(f) - \delta,
  \end{align*}
  where the second inequality follows from the definition of $X_f^\sigma(\delta)$.  Then for each $n$,
  \begin{align*}
    \tilde{h}(\delta) &\leq \lim_{\epsilon \rightarrow 0^+} \inf_{u \in U_{(-f)}(x_\delta, \epsilon)} h(u, x_\delta) + \delta \\
    {} &= s_n \lim_{\epsilon \rightarrow 0^+} \inf_{u \in U_{(-f)}(x_\delta, \epsilon)} f(u, x_\delta) + \lim_{\epsilon \rightarrow 0^+} \inf_{u \in U_{(-f)}(x_\delta, \epsilon)} h(u, x_\delta) - s_n \lim_{\epsilon \rightarrow 0^+} \inf_{u \in U_{(-f)}(x_\delta, \epsilon)} f(u, x_\delta) + \delta \\
    {} &\leq \lim_{\epsilon \rightarrow 0^+} \inf_{u \in U_{(-f)}(x_\delta, \epsilon)} (s_n f + h)(u, x_\delta) - s_n \lim_{\epsilon \rightarrow 0^+} \inf_{u \in U_{(-f)}(x_\delta, \epsilon)} f(u, x_\delta) + \delta \\
    {} &= \left( \lim_{\epsilon \rightarrow 0^+} \inf_{u \in U_{(-f)}(x_\delta, \epsilon)} (s_n f + h)(u, x_\delta) - s_n \sigma(f) \right) - s_n \left( \lim_{\epsilon \rightarrow 0^+} \inf_{u \in U_{(-f)}(x_\delta, \epsilon)} f(u, x_\delta) - \sigma(f) \right) + \delta \\
    {} &\leq \left( \lim_{\epsilon \rightarrow 0^+} \inf_{u \in U_{(-f)}(x_\delta, \epsilon)} (s_n f + h)(u, x_\delta) - s_n \sigma(f) \right) + (s_n + 1) \delta
  \end{align*}
  Then letting $\delta \rightarrow 0$ and using~\eqref{UisAinlimit} and~\eqref{lowerbnd_inner} we have
  \begin{equation*}
    \liminf_{n \rightarrow \infty} \left( \sigma(s_n f + h) - s_n \sigma(f) \right) \geq \lim_{\delta \rightarrow 0^+} \tilde{h}(\delta).
  \end{equation*}
\end{proof}

\begin{proof}[Proof of Theorem~\ref{thm:supnorm_stats}]
  The function $\gamma: \R^2 \rightarrow \R$ defined by $\gamma(x, y) = \max\{x, y\}$ has Hadamard directional derivative
  \begin{equation*}
    \gamma'_{x,y}(h, k) = \begin{cases} h & x > y \\ \max\{h, k\} & x = y \\ k & x < y \end{cases}.
  \end{equation*}

  Use the equivalence $\sup |f| = \max\{ \sup f, \sup (-f)\}$ to rewrite the difference for any $n$ as
  \begin{multline*}
    \frac{1}{t_n} \left( \sup_{x \in X} \left| \sup_{u \in A(x)} \left( f(u, x) + t_n h_n(u, x) \right) \right| - \sup_{x \in X} \left| \sup_{u \in A(x)} f(u, x) \right| \right) = \\
    \frac{1}{t_n} \Bigg( \max \left\{ \sup_{(u, x) \in \gr A} \left( f(u, x) + t_n h_n(u, x) \right), \sup_{x \in X} \left( -\sup_{u \in A(x)} \left( f(u, x) + t_n h_n(u, x) \right) \right) \right\} - \\
    - \max \left\{ \sup_{(u, x) \in \gr A} f(u, x), \sup_{x \in X} \left( - \sup_{u \in A(x)} f(u, x) \right) \right\} \Bigg) \\
    = \frac{1}{t_n} \left( \max \left\{ \mu \left( f + t_n h_n \right), \sigma \left( -f - t_n h_n \right) \right\} - \max \left\{ \mu(f), \sigma(-f) \right\} \right).
  \end{multline*}

  Defining $\tilde{\sigma}(f) = \sigma(-f)$, Lemma~\ref{lem:saddle} implies that $\tilde{\sigma}'_f(h) = \sigma_{(-f)}'(-h)$.  Then the chain rule implies
  \begin{equation*}
    (\gamma(\mu(f), \tilde{\sigma}(f)))'(h) = \gamma'_{(\mu(f), \tilde{\sigma}(f))}(\mu_f'(h), \tilde{\sigma}_f'(h)),
  \end{equation*}
  and writing the derivatives out is the first result.  The condition $\|\psi(f)\|_\infty = 0$ is equivalent to $\mu(f) = \sigma(-f) = 0$.  Specializing the general derivative so that $\psi(f) \equiv 0$ implies $A_f^\mu(\epsilon) = U_f(X, \epsilon)$.  Lemma~\ref{limits}, the continuity of $\max\{x, y\}$ and the identity $|x| = \max\{-x, x\}$ imply the second result.

  Next consider the functional $\lambda_2$.  Write $\sup_x \max\{ \sup_u f(u, x), 0\} = \max\{ \sup_{(u, x)} f(u, x), 0 \}$ and apply the chain rule with $\gamma(x, 0)$.  Then the general result for $\lambda_2(f)$ results from reversing the order of the maximum and supremum again.  When $\|[\psi(f)]_+\|_\infty = 0$, the derivative is only nonzero if $\sup_{\gr A} f = 0$ so that $A_f^\mu(\epsilon) = U_f(X_0, \epsilon)$, and similar calculations lead to the final result.
\end{proof}

\begin{lem} \label{limits}
  Let $A_f^\mu(\epsilon)$ and $U_f(\cdot, \epsilon)$ be defined by~\eqref{Afmu_def} and~\eqref{marginal_epsmax_def} in the text.  Let
  \begin{equation*}
    X_f(\epsilon) = \left\{ x \in X : \sup_{u \in A(x)} f(u, x) \geq \sup_{(u, x) \in \gr A} f(u, x) - \epsilon \right\}.
  \end{equation*}
  Then
  \begin{equation*}
    \lim_{\epsilon \rightarrow 0^+} \sup_{(u, x) \in A_f^\mu(\epsilon)} h(u, x) = \lim_{\epsilon \rightarrow 0^+} \lim_{\delta \rightarrow 0^+} \sup_{x \in X_f(\epsilon)} \sup_{u \in U_f(x, \delta)} h(u, x).
  \end{equation*}
\end{lem}

\begin{proof}[Proof of Lemma~\ref{limits}]
  Note that for any $\epsilon > 0$, $A_f^\mu(\epsilon) = \cup_{x \in X_f(\epsilon)} U_f(x, \epsilon)$.  Define
  \begin{equation*}
    \calA = \lim_{\epsilon \rightarrow 0^+} \sup_{(u, x) \in A_f^\mu(\epsilon)} h(u, x), \quad \calB = \lim_{\epsilon \rightarrow 0^+} \lim_{\delta \rightarrow 0^+} \sup_{x \in X_f(\epsilon)} \sup_{u \in U_f(x, \delta)} h(u, x).
  \end{equation*}
  As families of sets indexed by $\epsilon$, $A_f^\mu(\epsilon)$, $X_f(\epsilon)$ and $U_f(x, \epsilon)$ (for any $x$) are nonincreasing in diameter as $\epsilon \rightarrow 0$.  Therefore suprema evaluated over any of these sets are nonincreasing as $\epsilon \rightarrow 0$.

  For any $\epsilon, \delta > 0$, $\calA \leq \sup_{x \in X_f(\epsilon)} \sup_{u \in U_f(x, \delta)} h(u, x)$, and by letting $\epsilon, \delta \rightarrow 0$ we have $\calA \leq \calB$.  Since $h$ is bounded, $\calA < \infty$ and for all $\eta > 0$ there exist $\epsilon^*, \delta^* > 0$ such that
  \begin{equation*}
    \sup_{x \in X_f(\epsilon^*)} \sup_{u \in U_f(x, \delta^*)} h(u, f) \leq \calA + \eta.
  \end{equation*}
  Then
  \begin{equation*}
    \calB \leq \lim_{\epsilon \rightarrow 0^+} \sup_{x \in X_f(\epsilon)} \sup_{u \in U_f(x, \delta^*)} h(u, f) \leq \sup_{x \in X_f(\epsilon^*)} \sup_{u \in U_f(x, \delta^*)} h(u, f) \leq \calA + \eta.
  \end{equation*}
  By letting $\eta \rightarrow 0$ we find $\calB \leq \calA$.
\end{proof}

\begin{proof}[Proof of Theorem~\ref{thm:lpnorm_stats}]
  The result of this theorem follows from several applications of the chain rule for the function evaluated at (almost) every point in $X$, along with dominated convergence to move from pointwise convergence to convergence in the $L_p$ norm.  The maps and their Hadamard directional derivatives are discussed first and then composed.

  First, we note that the condition $f \in \ell^\infty(\gr A)$ and $m(X) < \infty$ imply $\|\psi(f)\|_p < \infty$: since $\| \psi(f) \|_p^p \leq \|f\|_\infty^p m(X)$, the $L_p$ norm of $\psi(f)$ is well-defined.

  Let $f$ and $h$ satisfy the assumptions in the statement of the theorem and choose $t_n \rightarrow 0^+$ and $h_n$ such that $\|h_n - h\|_\infty \rightarrow 0$.  It can be verified from the definition that:
  \begin{enumerate}
    \item For $x \in \R$ and $a \in \R$, the derivative of $x \mapsto |x|$ in direction $a$ is $\sgn(x) \cdot a$ when $x \neq 0$ and $|a|$ when $x = 0$.
    \item For $x \in \R$ and $b \in \R$, the derivative of $x \mapsto [x]_+$ in direction $b$ is $b$ when $x > 0$, $[b]_+$ when $x = 0$ and $0$ when $x < 0$. \label{pos_part}
    \item For $x \geq 0$, $\alpha > 0$ and $c \in \R$, the derivative of $x \mapsto x^\alpha$ in direction $c$ is $\alpha cx^{\alpha-1}$.
  \end{enumerate}

  Suppose that $\|\psi(f)\|_p \neq 0$.  For almost all $x \in X$ the chain rule implies that the derivative of $f(\cdot, x) \mapsto |\psi(f)(x)|^p$ in direction $h(\cdot, x)$ is
  \begin{equation*}
    \left\{ \begin{aligned} &p \sgn(\psi(f)(x)) |\psi(f)(x)|^{p-1} \psi'_f(h)(x), & \psi(f)(x) \neq 0 \\ &0, & \psi(f)(x) = 0 \end{aligned} \right\} = p\sgn(\psi(f)(x)) |\psi(f)(x)|^{p-1} \psi_f'(h)(x).
  \end{equation*}
  Similarly, for almost all $x \in X$ the derivative of $f(\cdot, x) \mapsto [\psi(f)(x)]_+^p$ in direction $h(\cdot, x)$ is
  \begin{equation*}
    p [\psi(f)(x)]_+^{p-1} \times \left\{ \begin{aligned} &\psi'_f(h)(x), & \psi(f)(x) > 0 \\ &[\psi'_f(h)(x)]_+, & \psi(f)(x) = 0 \\ &0, & \psi(f)(x) < 0 \end{aligned} \right\} = p [\psi(f)(x)]_+^{p-1} \psi_f'(h)(x).
  \end{equation*}
  Then the assumed $p$-integrability of these functions, implied by Assumption~\ref{A:estimator} and Proposition 6.10 of~\citet{Folland99}, and dominated convergence imply that when $\|\psi(f)\|_p \neq 0$,
  \begin{multline*}
    \lim_{n \rightarrow \infty} \frac{1}{t_n} \left( \int_X |\psi(f + t_n h_n)(x)|^p \ud m(x) - \int_X |\psi(f)(x)|^p \ud m(x) \right) \\
    = \int_X p\sgn(\psi(f)(x)) |\psi(f)(x)|^{p-1} \psi_f'(h)(x) \ud m(x)
  \end{multline*}
  and similarly, when $\|[\psi(f)]_+\|_p \neq 0$,
  \begin{equation*}
    \lim_{n \rightarrow \infty} \frac{1}{t_n} \left( \int_X [\psi(f + t_n h_n)(x)]_+^p \ud m(x) - \int_X [\psi(f)(x)]_+^p \ud m(x) \right) = \int_X p [\psi(f)(x)]_+^{p-1} \psi_f'(h)(x) \ud m(x).
  \end{equation*}
  When considering the $1/p$-th power of these integrals, one more application of the chain rule and the third basic derivative in the above list imply
  \begin{multline*}
    \lim_{n \rightarrow \infty} \frac{\lambda_3(f + t_nh_n) - \lambda_3(f)}{t_n} \\
    = \left( \int_X |\psi(f)(x)|^p \ud m(x) \right)^{(1-p)/p} \int_X \sgn(\psi(f)(x)) |\psi(f)(x)|^{p-1} \psi_f'(h)(x) \ud m(x) \\
  \end{multline*}
  and
  \begin{equation*}
    \lim_{n \rightarrow \infty} \frac{\lambda_4(f + t_nh_n) - \lambda_4(f)}{t_n} = \left( \int_X [\psi(f)(x)]_+^p \ud m(x) \right)^{(1-p)/p} \int_X [\psi(f)(x)]_+^{p-1} \psi_f'(h)(x) \ud m(x).
  \end{equation*}

  On the other hand, when $\|\psi(f)\|_p = 0$, one can calculate directly that
  \begin{align*}
    \lim_{n \rightarrow \infty} \frac{1}{t_n} \Bigg( \left( \int_X |\psi(f + t_nh_n)(x)|^p \ud m(x) \right)^{1/p} &- \left( \int_X |\psi(f)(x)|^p \ud m(x) \right)^{1/p} \Bigg) \\
    {} &= \left( \int_X \left| \lim_{n \rightarrow \infty} \frac{\psi(f + t_nh_n)(x) - 0}{t_n} \right|^p \ud m(x) \right)^{1/p} \\
    {} &= \left( \int_X \left| \lim_{\epsilon \rightarrow 0^+} \sup_{\{u : f(u, x) \geq -\epsilon\}} h(u, x) \right|^p \ud m(x) \right)^{1/p}.
  \end{align*}
  The case is slightly different for $\lambda_4$ because $\|[\psi(f)]_+\|_p = 0$ only implies $\psi(f)(x) \leq 0$ for almost all $x \in X$.  Rewrite the difference as in the above display.  Rule~\ref{pos_part} in the above list of derivative rules implies that only $X_0$ will contribute asymptotically to the inner integral, and in the limit, using $p$-integrability and dominated convergence we have the final statement.
\end{proof}

\begin{lem} \label{lem:other_sup}
  Let $\lambda$ be the map defined in~\eqref{dep_stat_def}.  Then $\lambda$ is Hadamard directionally differentiable and under the hypothesis that $\max\{ \Pi^+(F)(x), 0 \} = L_0(x)$ for all $x \in \R$, its derivative is defined in~\eqref{cband_derivative}.
\end{lem}

\begin{proof}[Proof of Lemma~\ref{lem:other_sup}]
  To save space, we have used the notation $\sup_A f = \sup_{x \in A} f(x)$ and $a \vee b = \max\{a, b\}$ in this proof.  The map $\Pi^+$ has a Hadamard derivative at any $(u, x)$ equal to $h_X(u) + h_Y(x - u)$.  Therefore we calculate, for $f \in \ell^\infty(\R^2)$, in the direction of some $h \in \ell^\infty(\R^2)$, the directional derivative
  \begin{equation} \label{derivative_to_solve}
    \varphi_f'(h) := \lim_{t \rightarrow 0^+} \frac{1}{t} \left( \sup_X \left| \psi(f + th) \vee 0 - L_0 \right| - \sup_X \left| \psi(f) \vee 0 - L_0 \right| \right).
  \end{equation}

  If $\psi(f)(x) < 0$ then for small enough $t$, $\max\{ \psi(f + th)(x), 0 \} = 0$.  Therefore we may focus on $\{x \in \R: \psi(f)(x) \geq 0\}$.  Break the domain into the two subsets $X_0 = \{x \in \R: \psi(f)(x) = 0\}$ and $X_+ = \{x \in \R: \psi(f)(x) > 0\}$.  Write
  \begin{equation*}
    \sup_X \left| \psi(f + th) \vee 0 - L_0 \right| = \max \left\{ \sup_{X_0} \left| \psi(f + th) \vee 0 - L_0 \right|, \sup_{X_+} \left| \psi(f + th) \vee 0 - L_0 \right| \right\},
  \end{equation*}
  and analogously for the second term in the difference~\eqref{derivative_to_solve}.  This indicates that the easiest way to find the derivative is via the chain rule.  We take each one of these two suprema in turn.

  Under the null hypothesis that $\psi(f) \equiv L_0$ and on $X_0$, where both are zero,
  \begin{align}
    \lim_{t \rightarrow 0^+} t^{-1} \bigg( \sup_{X_0} | \psi(f + th) \vee 0 - L_0 | - \sup_{X_0} | \psi(f) \vee 0 - L_0 | \bigg) &= \lim_{t \rightarrow 0^+} t^{-1} \sup_{X_0} \psi(f + th) \vee 0 \notag \\
    {} &= \lim_{\epsilon \rightarrow 0^+} \sup_{x \in X_0} \sup_{u \in U_f(x, \epsilon)} [h(u, x)]_+. \label{muzero_def}
  \end{align}
  Meanwhile, on $X_+$,
  \begin{multline*}
    \lim_{t \rightarrow 0^+} t^{-1} \left( \sup_{X_+} | \psi(f + th) \vee 0 - L_0 | - \sup_{X_+} | \psi(f) \vee 0 - L_0 | \right) \\
    = \lim_{t \rightarrow 0^+} t^{-1} \left( \sup_{X_+} | \psi(f + th) - L_0 | - \sup_{X_+} |\psi(f) - L_0| \right).
  \end{multline*}
  Apply Theorem~\ref{thm:supnorm_stats} using the function $\tilde{f}(u, x) = f(u, x) - L_0(x)$.  Note that $U_{\tilde{f}}(x, \epsilon) = U_f(x, \epsilon)$ for all $x \in X_+$.  Since the second term in the above expression is zero under the null hypothesis that $\psi(f) \equiv L_0$, Theorem~\ref{thm:supnorm_stats} implies
  \begin{equation} \label{muplus_def}
    \lim_{t \rightarrow 0^+} t^{-1} \left( \sup_{X_+} | \psi(f + th) - L_0 | - \sup_{X_+} | \psi(f) - L_0 | \right) = \lim_{\epsilon \rightarrow 0^+} \sup_{x \in X_+} \left| \sup_{u \in U_{f}(x, \epsilon)} h(u, x) \right|.
  \end{equation}

  Let
  \begin{equation*}
    \mu_0(f) = \sup_{X_0} |\psi(f) \vee 0 - L_0|, \qquad \mu_+(f) = \sup_{X_+} |\psi(f) \vee 0 - L_0|.
  \end{equation*}
  Then define the derivatives $\mu_{0f}'$ and $\mu_{+f}'$ by~\eqref{muzero_def} and~\eqref{muplus_def}.  Now, via the chain rule, under the hypothesis that $L_0$ has been correctly chosen, $\mu_0(f) = \mu_+(f) = 0$.  Then using the derivative of the map $\gamma$ defined in the proof of Theorem~\ref{thm:supnorm_stats}, we have a solution to~\eqref{derivative_to_solve}:
  \begin{equation} \label{varphi_pr_final}
    \varphi_f'(h) = \max \left\{ \mu_{0f}'(h), \mu_{+f}'(h) \right\}.
  \end{equation}
  To translate this back to a result on the $\lambda$ defined in~\eqref{dep_stat_def}, use the definitions that precede~\eqref{cband_derivative} in~\eqref{varphi_pr_final} and let $h(u, x) = h_X(u) + h_Y(x - u)$.
\end{proof}

\begin{proof}[Proof of Theorem~\ref{thm:asymptotic}]
  The limit statements follow directly from Theorem~\ref{thm:supnorm_stats}, Theorem~\ref{thm:lpnorm_stats} above combined with Theorem 2.1 of \citet{FangSantos19}.
\end{proof}

\begin{proof}[Proof of Theorem~\ref{thm:bootstrap_consistency_teststats}]
  We will apply Theorem 3.2 of \citet{FangSantos19} to show the result by verifying that the maps involved satisfy the appropriate regularity conditions.  First, Theorems~\ref{thm:supnorm_stats} and~\ref{thm:lpnorm_stats} show that the $\lambda_j$ are Hadamard directionally differentiable maps, implying their Assumption 1 is satisfied.  Our assumptions~\ref{A:sequences} and~\ref{A:bootstrapf} imply their Assumptions 2 and 3.

  Next, we show that each map is Lipschitz and consistent for each fixed direction $h$, so that we may apply Lemma S.3.6 of \citet{FangSantos19}, which implies their Assumption 4 is satisfied.  Start with $\lambda_1$.  Fix $f \in \ell^\infty(\gr A)$ and suppose that $h, k \in \ell^\infty(\gr A)$.  For any $n$,
  \begin{equation*}
    \left| \hat{\lambda}_{1n}'(h) - \hat{\lambda}_{1n}'(k) \right| \leq \sup_{(u, x) \in U_{f_n}(X, a_n)} \left| h(u, x) - k(u, x) \right| \leq \| h - k \|_\infty.
  \end{equation*}
  Next we show that for each $h$, $|\hat{\lambda}_{1n}'(h) - \lambda_{1f}'(h)| = o_p(1)$.  For any $\eta > 0$,
  \begin{align*}
    \prob{ \left| \hat{\lambda}_{1n}'(h) - \lambda_{1f}'(h) \right| > \eta } &\leq \prob{ \left| \sup_{(u, x) \in U_{f_n}(X, a_n)} h(u, x) - \lim_{\epsilon \rightarrow 0^+} \sup_{(u, x) \in U_{f}(X, \epsilon)} h(u, x) \right| > \eta }.
  \end{align*}
  If $(u_n, x_n) \in U_{f_n}(X, a_n)$, then $f_n(u_n, x_n) \geq \sup_{(u, x) \in gr A} f_n(u, x) - a_n$.  This is equivalent to
  \begin{align*}
    f(u_n, x_n) + (f_n(u_n, x_n) - f(u_n, x_n)) &\geq \sup_{(u, x) \in \gr A} f(u, x) + \left( \sup_{(u, x) \in \gr A} f_n(u, x) - \sup_{(u, x) \in \gr A} f(u, x) \right) - a_n \\
    \Rightarrow f(u_n, x_n) &\geq \sup_{(u, x) \in \gr A} f(u, x) - a_n - 2 \|f_n - f\|_\infty.
  \end{align*}
  That is, $U_{f_n}(X, a_n) \subseteq U_f(X, a_n + 2\|f_n - f\|_\infty)$.  Therefore as $n \rightarrow \infty$,
  \begin{multline*}
    \prob{ \sup_{(u, x) \in U_{f_n}(X, a_n)} h(u, x) > \lim_{\epsilon \rightarrow 0^+} \sup_{(u, x) \in U_{f}(X, \epsilon)} h(u, x) + \eta } \\
    \leq \prob{ \sup_{(u, x) \in U_{f}(X, a_n + 2\|f_n - f\|_\infty)} h(u, x) > \lim_{\epsilon \rightarrow 0^+} \sup_{(u, x) \in U_{f}(X, \epsilon)} h(u, x) + \eta } \rightarrow 0
  \end{multline*}
  It can similarly be shown that $U_f(X, a_n) \subseteq U_{f_n}(X, a_n + 2\|f_n - f\|_\infty)$.  This implies that
  \begin{equation*}
    \lim_{\epsilon \rightarrow 0^+} \sup_{(u, x) \in U_{f}(X, \epsilon)} h(u, x) \leq \sup_{(u, x) \in U_{f}(X, a_n)} h(u, x) \leq \sup_{(u, x) \in U_{f}(X, a_n + 2\|f_n - f\|_\infty)} h(u, x).
  \end{equation*}
  Then as $n \rightarrow \infty$,
  \begin{multline*}
    \prob{ \lim_{\epsilon \rightarrow 0^+} \sup_{(u, x) \in U_{f}(X, \epsilon)} h(u, x) > \sup_{(u, x) \in U_{f_n}(X, a_n)} h(u, x) + \eta } \\
    \leq \prob{ \sup_{(u, x) \in U_{f_n}(X, a_n + 2\|f_n - f\|_\infty)} h(u, x) > \sup_{(u, x) \in U_{f_n}(X, a_n)} h(u, x) + \eta } \rightarrow 0.
  \end{multline*}

  Next consider $\lambda_2$.  Using the fact that $x \mapsto [x]_+$ is Lipschitz with constant 1,
  \begin{equation*}
    \left| \hat{\lambda}_{2n}'(h) - \hat{\lambda}_{2n}'(k) \right| \leq \sup_{(u, x) \in U_{f_n}(X, a_n)} \left| [h(u, x)]_+ - [k(u, x)]_+ \right| \leq \| h - k \|_\infty.
  \end{equation*}
  In order to show consistency for fixed $h$, recall that under the hypothesis that $P \in \calP_{00}^I$ there is only one form for $\lambda_{2f}'$ (that is, the trivial case where it is equal to zero is ruled out by assumption).  Note that for any $\eta > 0$
  \begin{align*}
    \prob{ \left| \hat{\lambda}_{2n}'(h) - \lambda_{2f}'(h) \right| > \eta } \leq \prob{ \left| \sup_{(u, x) \in U_{f_n}(\hat{X}_0, a_n)} h(u, x) - \lim_{\epsilon \rightarrow 0^+} \sup_{(u, x) \in U_{f}(X_0, \epsilon)} h(u, x) \right| > \eta } \\
    \leq \prob{ \left| \sup_{(u, x) \in U_{f_n}(X_0 \cap \hat{X}_0, a_n)} h(u, x) - \lim_{\epsilon \rightarrow 0^+} \sup_{(u, x) \in U_{f}(X_0 \cap \hat{X}_0, \epsilon)} h(u, x) \right| > \eta } + \prob{ \hat{X}_0 \triangle X_0 \neq \varnothing }.
  \end{align*}
  On $X_0 \cap \hat{X}_0$, consistency is guaranteed by the argument used above for $\hat{\lambda}_{1n}'$.  Consider the probability that $X_0 \triangle \hat{X}_0$ is non-empty.  Theorem~\ref{thm:asymptotic} applied to $\lambda_1$ implies that $\sup_{x \in X} r_n |\psi(f_n)(x) - \psi(f)(x)|$ is asymptotically tight.  Therefore if $b_n$ satisfies the condition in Assumption~\ref{A:sequences}, for all $x \in X$, $| \psi(f_n)(x') - \psi(f)(x') | = O_p(r_n) = o_p(b_n)$.  First, if $x \in X_0$ then $\psi(f)(x) = 0$ and
  \begin{equation*}
    \prob{x \in X_0 \cap \hat{X}_0^{\mathsf{c}}} = \prob{ |\psi(f_n)(x)| > b_n} = \prob{ |\psi(f_n)(x) - \psi(f)(x)| > b_n } \rightarrow 0.
  \end{equation*}
  On the other hand, if $x' \notin X_0$ then $\psi(f)(x') \neq 0$ and $\prob{|\psi(f)(x')| \leq b_n} \rightarrow 0$.  Then
  \begin{equation*}
    \prob{x \in X_0^{\mathsf{c}} \cap \hat{X}_0} = \prob{ |\psi(f_n)(x')| \leq b_n } \leq \prob{ | \psi(f_n)(x') - \psi(f)(x') | + |\psi(f)(x')| \leq b_n } \rightarrow 0.
  \end{equation*}
  This implies that for all $\epsilon > 0$, $\prob{d(\hat{X}_0, X_0) > \epsilon} \rightarrow 0$.

  For the case of $\lambda_3$, suppose in addition that $m(X) < \infty$.  Then
  \begin{align*}
    \left| \hat{\lambda}_{3n}'(h) - \hat{\lambda}_{3n}'(k) \right| &\leq \left( \int_X \left| \sup_{u \in U_{f_n}(x, a_n)} h(u, x) - \sup_{u \in U_{f_n}(x, a_n)} k(u, x) \right|^p \ud m(x) \right)^{1/p} \\
    {} &\leq \left( \int_X \left| \sup_{u \in A(x)} | h(u, x) - k(u, x) | \right|^p \ud m(x) \right)^{1/p} \\
    {} &\leq m(X)^{1/p} \| h - k \|_\infty.
  \end{align*}
  The above inequalities follow from Minkowski's inequality (which implies a reverse triangle inequality) and a bound on the $L_p$ norm assuming $m(X)$ is finite.  Similarly, for any fixed $\eta > 0$, defining
  \begin{align*}
    X_1 &= \left\{ x \in X : \sup_{u \in U_{f_n}(x, a_n)} h(u, x) > \lim_{\epsilon \rightarrow 0^+} \sup_{u \in U_f(x, \epsilon)} h(u, x) \right\}, \\
    X_2 &= \left\{ x \in X : \sup_{u \in U_{f_n}(x, a_n)} h(u, x) < \lim_{\epsilon \rightarrow 0^+} \sup_{u \in U_f(x, \epsilon)} h(u, x) \right\},
  \end{align*}
  we may write
  \begin{multline*}
    \prob{ \left| \hat{\lambda}_{3n}'(h) - \lambda_{3f}'(h) \right| > \eta } \\
    \leq \prob{ \left( \int_{X_1} \left( \sup_{u \in U_{f_n}(x, a_n)} h(u, x) - \lim_{\epsilon \rightarrow 0^+} \sup_{u \in U_{f}(x, \epsilon)} h(u, x) \right)^p \ud m(x) \right)^{1/p} > \eta } \\
    + \prob{ \left( \int_{X_2} \left( \lim_{\epsilon \rightarrow 0^+} \sup_{u \in U_{f}(x, \epsilon)} h(u, x) - \sup_{u \in U_{f_n}(x, a_n)} h(u, x) \right)^p \ud m(x) \right)^{1/p} > \eta }
  \end{multline*}
  Using an analogous argument as for the $\lambda_1$ case, this sum is bounded by
  \begin{multline*}
    \prob{ \left( \int_{X_1} \left( \sup_{u \in U_{f}(x, a_n + 2\|f_n - f\|_\infty)} h(u, x) - \lim_{\epsilon \rightarrow 0^+} \sup_{u \in U_{f}(x, \epsilon)} h(u, x) \right)^p \ud m(x) \right)^{1/p} > \eta } \\
    + \prob{ \left( \int_{X_2} \left( \sup_{u \in U_{f_n}(x, a_n + 2\|f_n - f\|_\infty)} h(u, x) - \sup_{u \in U_{f_n}(x, a_n)} h(u, x) \right)^p \ud m(x) \right)^{1/p} > \eta } \rightarrow 0
  \end{multline*}
  by monotone convergence and the dominated convergence theorem.

  Finally, because $x \mapsto [x]_+$ is Lipschitz we find
  \begin{equation*}
    | \hat{\lambda}_{4n}'(h) - \hat{\lambda}_{4n}'(k) | \leq \left( \int_X \left| \sup_{u \in A(x)} | h(u, x) - k(u, x) | \right|^p \ud m(x) \right)^{1/p} \leq m(X)^{1/p} \|h - k\|_\infty.
  \end{equation*}
  Consistency for fixed $h$ follows from the argument for the $\lambda_3$ case.  Specifically, subdivide the domain of integration further into $X_1 \cap \hat{X}_0 \cap X_0$, $X_2 \cap \hat{X}_0 \cap X_0$ and $\hat{X}_0 \triangle X_0$.  The first two subsets can be dealt with as in the $\lambda_3$ case, and the last can be dealt with as in the argument for the $\lambda_2$ case.

  Bootstrap consistency for $\lambda_j$ for all $j$ then follows from Theorem 3.2 and Lemma S.3.6 of \citet{FangSantos19}.
\end{proof}

\begin{proof}[Proof of Theorem~\ref{thm:size_control}]
  For $j = 1$ or $j = 3$, under the assumption that $P \in \calP_{00}^E$, there are no interior distributions and so for any sequence $P_n \rightarrow P$ such that $P_n \in \calP_{00}^E$, $\lambda_j(f(P_n)) = 0$ for all $n$.  Therefore $\lambda_{jf}'(f'(c)) = \lim_{n \rightarrow \infty} r_n(\lambda_j(f(P_n)) - \lambda_j(f(P_0))) = 0$.  Now consider $\lambda_j$ for $j \in \{1, 3\}$.  Theorem 3.3 of \citet{FangSantos19} implies that, using $Z_n \stackrel{L_c}{\cw} Z$ to denote weak convergence under $\bigotimes_{i=1}^n P_{c/r_n}$,
  \begin{equation*}
    r_n(\lambda_j(f_n) - \lambda_j(f(P_n))) \stackrel{L_{c}}{\cw} \lambda_{jf}'(\calG_{P_0} + f'(c)) - \lambda_{jf}'(f'(c)) \sim \lambda_{jf}'(\calG_{P_0} + f'(c)),
  \end{equation*}


  The result for the other two statistics follows from Corollary 3.2 of \citet{FangSantos19} if we show the convexity of the derivatives $\lambda_{jf}'(h)$ when $\lambda_j(f) = 0$ for $j = 2$ or $j = 4$.  The case that $P \in \calP_0^I \backslash \calP_{00}^I$ is trivially convex, so we restrict our attention to $P \in \calP_{00}^I$.  It can be verified that for $a, b \in \R$, $[a + b]_+ \leq [a]_+ + [b]_+$.  Then for $h, k \in \ell^\infty(U \times X)$,
  \begin{equation*}
    \lambda_{2f}'(\alpha h + (1 - \alpha) k) \leq \left[ \alpha \lim_{\epsilon \rightarrow 0^+} \sup_{U_f(X, \epsilon)} h + (1 - \alpha) \lim_{\epsilon \rightarrow 0^+} \sup_{U_f(X, \epsilon)} k \right]_+ \leq \alpha \lambda_{2f}'(h) + (1 - \alpha) \lambda_{2f}'(k).
  \end{equation*}
  Similarly,
  \begin{align*}
    \lambda_{4f}'(\alpha h + (1 - \alpha) k) &\leq \left( \int_{X_0} \left| \alpha [\psi'_f(h)(x)]_+ + (1 - \alpha) [\psi'_f(k)(x)]_+ \right|^p \ud m(x) \right)^{1/p} \\
    &\leq \alpha \lambda_{4f}'(h) + (1 - \alpha) \lambda_{4f}'(k).
  \end{align*}
\end{proof}

The next lemma confirms that under assumptions that ensure convergence of the data occurs uniformly over $P \in \calP$, the results of Theorems~\ref{thm:asymptotic} and~\ref{thm:bootstrap_consistency_teststats} also hold uniformly over $P \in \calP$.  The second result is limited to $\calP_{00}^I$ because the estimates $\hat{\lambda}_{jn}'$ are constructed under the assumption that $P \in \calP_{00}^I$.
\begin{lem} \label{lem:uniform_weakconv}
  Under Assumptions~\ref{A:estimator}-\ref{A:bootstrapf} and \ref{A:unif_convf}-\ref{A:unif_boot}, for $j \in \{2, 4\}$,
  \begin{equation*}
    \limsup_{n \rightarrow \infty} \sup_{P \in \calP} \sup_{g \in BL_1(\R))} \left| \ex{g(r_n (\lambda_j(f_n) - \lambda_j(f))} - \ex{g(\lambda_{jf}'(\calG_P))} \right| = 0
  \end{equation*}
  and for any $\epsilon > 0$,
  \begin{equation*}
    \limsup_{n \rightarrow \infty} \sup_{P \in \calP_{00}^I} \prob{ \sup_{g \in BL_1(\R)} \left| \ex{g \left( \hat{\lambda}_{jn}'(r_n (f_n^\ast - f_n)) \right) | \{Z_i\}_{i=1}^n } - \ex{g(\lambda_{jf}'(\calG_P))} \right| > \epsilon } = 0.
  \end{equation*}
\end{lem}

\begin{proof}[Proof of Lemma~\ref{lem:uniform_weakconv}]
  First note that under Assumptions~\ref{A:unif_convf} and~\ref{A:unif_tight} (along with Theorems~\ref{thm:supnorm_stats} and~\ref{thm:lpnorm_stats}), the convergence-in-distribution result of Theorem 2.1 of \citet{FangSantos19} may be rewritten to hold uniformly over $P \in \calP$.  Specifically, use the uniform continuous mapping theorem developed in~\citet{LintonSongWhang10} in the proof of Theorem 2.1 of \citet{FangSantos19} (in the place of their appeal to Theorem 1.11.1 of \citet{vanderVaartWellner96}) and note that under the null hypothesis, $\lambda_j(f_n) = \lambda_j(f_n) - \lambda_j(f)$.

  The second result is similar to Theorem 3.2 of \citet{FangSantos19} but holds uniformly over $P \in \calP_{00}^I$ for the specific statistics considered in the lemma.  For convenience let $\calG_n^* = r_n(f_n^* - f_n)$ and $Z = \{Z_i\}_{i=1}^n$.

  Under Assumptions~\ref{A:bootstrapf}, \ref{A:unif_tight} and~\ref{A:unif_boot}, Lemma S.3.1 of \citet{FangSantos19} can be extended to hold uniformly over $P \in \calP$ to show that uniformly in $P \in \calP$, $\calG_n^* \cw \calG_P$.  This implies that for any closed set $F \in \mathcal{C}(\gr A)$,
  \begin{equation} \label{portmanteau}
    \limsup_{n \rightarrow \infty} \sup_{P \in \calP} \prob{ \calG_n^* \in F } \leq \sup_{P \in \calP} \prob{ \calG_P \in F }.
  \end{equation}
  This can be seen by writing
  \begin{equation*}
    \limsup_{n \rightarrow \infty} \sup_{P \in \calP} \prob{ \calG_n^* \in F } \leq \limsup_{n \rightarrow \infty} \sup_{P \in \calP} \left( \prob{ \calG_n^* \in F } - \prob{ \calG_P \in F } \right) + \sup_{P \in \calP} \prob{ \calG_P \in F }.
  \end{equation*}
  If the first term on the right-hand side were strictly positive, then as $n \rightarrow \infty$, for some $\epsilon > 0$ it occurs infinitely often that there is a $P \in \calP$ such that $\prob{ \calG_n^* \in F } - \prob{ \calG_P \in F } \geq \epsilon$.  Construct a sequence of such measures $\{P_k\}_{k=1}^\infty$.  But then it is possible to find functions that are multiples of functions in $BL_1(\gr A)$ that violate the weak convergence assumption: use functions analogous (the functions should dominate indicator functions here) to those in Addendum 1.12.3 and Theorem 1.12.2 of \citet{vanderVaartWellner96}.  This contradiction implies that the first term on the right-hand side must be nonpositive, implying~\eqref{portmanteau}.

  Next we show that for fixed $h$, for every $\epsilon > 0$,
  \begin{equation} \label{unif_in_P_cons}
    \limsup_{n \rightarrow \infty} \sup_{P \in \calP_{00}^I} \prob{ \| \hat{\lambda}_{jn}'(h) - \lambda_{jf}'(h) \|_\infty > \epsilon} = 0.
  \end{equation}
  Consider $\lambda_{2f}'$.  As in the proof of Theorem~\ref{thm:bootstrap_consistency_teststats}, for any $\eta > 0$,
  \begin{equation*}
    \sup_{P \in \calP_{00}^I} \prob{ \left| \hat{\lambda}_{2n}'(h) - \lambda_{2f}'(h) \right| > \eta } \leq \sup_{P \in \calP_{00}^I} \prob{ \left| \sup_{(u, x) \in U_{f_n}(X, a_n)} h(u, x) - \lim_{\epsilon \rightarrow 0^+} \sup_{(u, x) \in U_{f}(X, \epsilon)} h(u, x) \right| > \eta }.
  \end{equation*}
  As shown in the proof of that theorem, $U_{f_n}(X, a_n) \subseteq U_f(X, a_n + 2\|f_n - f\|_\infty)$ and $U_f(X, a_n) \subseteq U_{f_n}(X, a_n + 2\|f_n - f\|_\infty)$.  Under Assumptions~\ref{A:unif_convf} and~\ref{A:unif_tight}, the other arguments of that theorem and Lemma~\ref{limits} imply that
  \begin{equation*}
    \sup_{P \in \calP_{00}^I} \prob{ \left| \sup_{(u, x) \in U_{f_n}(X, a_n)} h(u, x) - \lim_{\epsilon \rightarrow 0^+} \sup_{(u, x) \in U_{f}(X, \epsilon)} h(u, x) \right| > \eta } \rightarrow 0,
  \end{equation*}
  implying~\eqref{unif_in_P_cons} holds for $\lambda_{2f}'$.  Arguments like those used in Theorem~\ref{thm:bootstrap_consistency_teststats} for $\lambda_{4f}'$ hold uniformly over $\calP_{00}^I$ as well under Assumption~\ref{A:unif_convf}, implying~\eqref{unif_in_P_cons} holds for $\lambda_{4f}'$ as well (note that under Assumption~\ref{A:unif_convf}, the contact set consistency argument made at the end of the proof of Theorem~\ref{thm:bootstrap_consistency_teststats} holds uniformly over $P \in \calP_{00}^I$ as well).  In the proof of Theorem~\ref{thm:bootstrap_consistency_teststats} it is shown that $\hat{\lambda}_{jn}'$ are Lipschitz with bounded Lipschitz constants.  Lemma~\ref{lem:FS_S36} shows that this and~\eqref{unif_in_P_cons} imply that for any compact set $K \subset \ell^\infty(\gr A)$, letting $K^\delta$ be its $\delta$-enlargement for any $\delta > 0$,
  \begin{equation} \label{uniform_consistency}
    \lim_{\delta \rightarrow 0^+} \limsup_{n \rightarrow \infty} \sup_{P \in \calP_{00}^I} \prob{\sup_{h \in K^\delta} \| \hat{\lambda}_{jn}'(h) - \lambda_{jf}'(h) \|_\infty > \epsilon} = 0.
  \end{equation}

  Now fix arbitrary $\epsilon > 0$ and $\eta > 0$.  By Assumption~\ref{A:unif_tight}, there is a compact set $K_0 \in \mathcal{C}(\gr A)$ such that $\sup_{P \in \calP_{00}^I} \prob{ \calG_P \notin K_0 } < \epsilon \eta / 2$.  Equation~\eqref{portmanteau} implies that for any $\delta > 0$,
  \begin{equation}
    \limsup_{n \rightarrow \infty} \sup_{P \in \calP_{00}^I} \prob{ \calG_n^* \notin K_0^\delta } \leq \sup_{P \in \calP_{00}^I} \prob{ \calG_P \notin K_0^\delta } \leq \sup_{P \in \calP_{00}^I} \prob{ \calG_P \notin K_0 } \leq \epsilon \eta / 2.
  \end{equation}

  Then for each $P \in \calP_{00}^I$ and any $\delta_0 > 0$,
  \begin{align}
    \sup_{g \in BL_1(\R)} &\phantom{=} \left| \ex{ g(\hat{\lambda}_{jn}'(\calG_n^*)) | Z } - \ex{ g(\lambda_{jf}'(\calG_n^*)) | Z } \right| \notag \\
    {} &\leq \sup_{g \in BL_1(\R)} \ex{ \left| g(\hat{\lambda}_{jn}'(\calG_n^*)) - g(\lambda_{jf}'(\calG_n^*)) \right| \Big| Z } \notag \\
    {} &\leq \ex{ 2 I(\calG_n^* \notin K_0^{\delta_0}) + \sup_{h \in K_0^{\delta_0}} \| \hat{\lambda}_{jn}'(h) - \lambda_{jf}'(h) \| \Big| Z } \notag \\
    {} &\leq 2 \prob{ \calG_n^* \notin K_0^{\delta_0} | Z } + \sup_{h \in K_0^{\delta_0}} \| \hat{\lambda}_{jn}'(h) - \lambda_{jf}'(h) \|, \label{big_ineq}
  \end{align}
  where the last term is not conditional on the data $Z$ because $\hat{\lambda}_{jn}'$ only depends on $Z$.

  Markov's inequality and Fubini's theorem \citep[Lemma 1.2.7]{vanderVaartWellner96} imply that
  \begin{equation} \label{tail_small}
    \limsup_{n \rightarrow \infty} \sup_{P \in \calP} 2 \prob{ \prob{ \calG_n^* \notin K_0^{\delta_0} | Z } > \epsilon } \leq \limsup_{n \rightarrow \infty} \sup_{P \in \calP} 2 \epsilon^{-1} \prob{\calG_n^* \notin K_0^{\delta_0}} < \eta.
  \end{equation}
  Furthermore, \eqref{uniform_consistency} implies that $\delta_0$ can be chosen such that
  \begin{equation} \label{est_small}
    \limsup_{n \rightarrow \infty} \sup_{P \in \calP_{00}^I} \prob{\sup_{h \in K^{\delta_0}} \| \hat{\lambda}_{jn}'(h) - \lambda_{jf}'(h) \|_\infty > \epsilon} < \eta.
  \end{equation}
  Finally, using Assumptions~\ref{A:bootstrapf}, \ref{A:regular}, \ref{A:unif_boot} and the Lipschitz continuity of $\lambda_{jf}'$, which can be shown as in the proof of Theorem~\ref{thm:bootstrap_consistency_teststats}, Lemma A.1 of \citet{LintonSongWhang10} implies
  \begin{equation} \label{other_weakc}
    \limsup_{n \rightarrow \infty} \sup_{P \in \calP} \prob{ \sup_{g \in BL_1(\R)} \left| \ex{ g(\lambda_{jf}'(\calG_n^*)) | Z } - \ex{g(\lambda_{jf}'(\calG_P))} \right| > \epsilon } = 0.
  \end{equation}

  By combining~\eqref{big_ineq}, \eqref{est_small}, \eqref{tail_small} and~\eqref{other_weakc}, we have
  \begin{equation*}
    \limsup_{n \rightarrow \infty} \sup_{P \in \calP_{00}^I} \prob{ \sup_{g \in BL_1(\R)} \left| \ex{ g(\hat{\lambda}_{jn}'(\calG_n^*)) | Z } - \ex{ g(\lambda_{jf}'(\calG_P)) } \right| > 3\epsilon } < 3\eta.
  \end{equation*}
  Since the constants are arbitrary we have the second result.
\end{proof}

\begin{lem}{(\citet{FangSantos19})} \label{lem:FS_S36}
  Suppose that a Hadamard directional derivative $\lambda'_f: \mathbb{D} \rightarrow \mathbb{E}$ has estimator $\hat{\lambda}'_n: \mathbb{D} \rightarrow \mathbb{E}$ such that for some $\kappa > 0$ and $C_n \in \R_+$, $\|\hat{\lambda}_n'(h) - \hat{\lambda}_n'(k)\|_{\mathbb{E}} \leq C_n \|h - k\|_{\mathbb{D}}^\kappa$.  Let $\calP$ denote a set of probability distributions.  Suppose that (a) $C_n = O_P(1)$ uniformly over $P \in \calP$, i.e., for all $\epsilon > 0$ there is an $M$ such that $\sup_{P \in \calP} \prob{C_n > M} < \epsilon$, and (b) for all $\epsilon > 0$, $\limsup_{n \rightarrow \infty} \sup_{P \in \calP} \prob{ \| \hat{\lambda}_n'(h) - \lambda_f'(h) \|_{\mathbb{E}} > \epsilon} = 0$.  Then for all compact $K \in \mathbb{D}$, letting $K^\delta$ for any $\delta > 0$ denote the $\delta$-enlargement of $K$, for every $\epsilon > 0$ we have
  \begin{equation*}
    \lim_{\delta \rightarrow 0^+} \limsup_{n \rightarrow \infty} \sup_{P \in \calP} \prob{ \sup_{h \in K^\delta} \| \hat{\lambda}_n'(h) - \lambda_f'(h) \|_{\mathbb{E}} > \epsilon } = 0.
  \end{equation*}
\end{lem}

\begin{proof}[Proof of Lemma~\ref{lem:FS_S36}]
  This is a straightforward adaptation of Lemma S.3.6 of \citet{FangSantos19}, and only needs a few minor modifications.  First, the assumption on $C_n$ must be strengthened to hold uniformly over $P \in \calP$.  Most of the statements are analytic in nature (including the appeal to their Lemma S.3.4) and assumed to hold outer almost surely, and the probabilistic inequalities still hold when taking suprema over $P \in \calP$.
\end{proof}

\begin{proof}[Proof of Theorem~\ref{thm:uniform_size_control}]
  For convenience let $\calG_n^* = r_n(f_n^* - f_n)$ and let $Z = \{Z_i\}_{i=1}^n$.  Consider the subset $\calP_{00}^I$, which is the collection of $P \in \calP_0^I$ such that $X_0 \neq \varnothing$.  Let $\hat{F}_{n,P}(s) = \prob{ \hat{\lambda}_{jn}'(\calG_n^*) \leq s | Z }$ and $F_P(s) = \prob{ \lambda_{jf}'(\calG_P) \leq s }$ be the distribution functions of bootstrapped and asymptotic statistics.  We note that $\hat{q}_{\lambda^*}(1-\alpha) = \inf\{q \in \R : \hat{F}_{n,P}(q) \geq 1 - \alpha\}$, and define $q_{\lambda'}^P(1-\alpha) = \inf\{q \in \R : F_P(q) \geq 1 - \alpha\}$.  Lemma 10.11 of~\citet{Kosorok08} can be extended to hold uniformly over $P \in \calP_{00}^I$, so that given the second result of Lemma~\ref{lem:uniform_weakconv}, under Assumption~\ref{A:regular}, for any closed $S \subset \R$ with $\min S > 0$,
  \begin{equation*}
    \lim_{n \rightarrow \infty} \sup_{P \in \calP_{00}^I} \prob{ \sup_{s \in S} | \hat{F}_{n,P}(s) - F_P(s) | > \epsilon } = 0.
  \end{equation*}
  As in the proof of Theorem S.1.1 of \citet{FangSantos19}, we find that for any $\epsilon > 0$,
  \begin{equation} \label{cv_consistent}
    \sup_{P \in \calP_{00}^I} \prob{ | \hat{q}_{\lambda^*}(1-\alpha) - q_{\lambda'}^P(1-\alpha) | > \epsilon } \rightarrow 0,
  \end{equation}
  so that the critical values are consistent uniformly over $P \in \calP_{00}^I$.  Under the null hypothesis, $r_n \lambda_j(f_n) = r_n( \lambda_j(f_n) - \lambda_j(f))$, and Lemma~\ref{lem:uniform_weakconv}, Assumption~\ref{A:regular} and~\eqref{cv_consistent} imply
  \begin{equation*}
    \limsup_{n \rightarrow \infty} \sup_{P \in \calP_{00}^I} \left| \prob{ r_n \lambda_j(f_n) > \hat{q}_{\lambda^*}(1-\alpha) } - \prob{ \lambda_{jf}'(\calG_P) > q_{\lambda'}^P(1-\alpha) } \right| = 0.
  \end{equation*}
  Since the second probability inside the absolute value is less than or equal to $\alpha$, we have
  \begin{equation} \label{limit2}
    \limsup_{n \rightarrow \infty} \sup_{P \in \calP_{00}^I} \prob{ r_n \lambda_j(f_n) > \hat{q}_{\lambda^*}(1-\alpha) } \leq \alpha.
  \end{equation}

  Now consider the subset $\calP_0^I \backslash \calP_{00}^I$, which implies that $X_0 = \varnothing$, that is, such that $\psi(f)(x) < 0$ for all $x \in X$.  The definition of $\hat{X}_0$ ensures that under the assumption that $X_0 = \varnothing$,
  \begin{align}
    \prob{ \hat{X}_0 \neq X } &= \prob{ \sup_{x \in X} |\psi(f_n)(x)| \leq b_n } \notag \\
    {} &\leq \prob{ \sup_{x \in X} |r_n(\psi(f_n)(x) - \psi(f)(x))| + \sup_{x \in X} |r_n \psi(f)(x)| \leq r_n b_n } \notag \\
    {} &\rightarrow 0 \label{X_est}
  \end{align}
  uniformly over $P \in \calP_0^I \backslash \calP_{00}^I$, where the inequality follows from the fact that $\sup_{x \in X} |\psi(f)(x)| > 0$ and $b_n \rightarrow 0$ and by Assumptions~\ref{A:unif_convf} and~\ref{A:unif_tight}.  Let $\tilde{\lambda}_{jn}'$ be defined like their asymptotic counterparts $\tilde{\lambda}_{jf}'$ that were defined in Assumption~\ref{A:regular}:
  \begin{equation*}
    \tilde{\lambda}_{2n}'(h) = \sup_{(u, x) \in U_{f_n}(X, a_n)} [h(u, x)]_+, \quad \tilde{\lambda}_{4n}'(h) = \left( \int_X \sup_{u \in U_{f_n}(x, a_n)} [h(u, x)]_+^p \ud m(x) \right)^p.
  \end{equation*}
  A small modification of the proof of the second part of Lemma~\ref{lem:uniform_weakconv} (that is, without estimation of $X_0$) implies that $\tilde{\lambda}_{jn}'(r_n(f_n^* - f_n)) \cw \tilde{\lambda}_{jf}'(\calG_P)$ in probability uniformly over $P \in \calP_0^I$.  Since also~\eqref{X_est} implies that for all $\epsilon > 0$, $\limsup_{n \rightarrow \infty} \sup_{P \in \calP_0^I \backslash \calP_{00}^I} \prob{ \left| \hat{\lambda}_{jn}'(\calG_n^*) - \tilde{\lambda}_{jn}'(\calG_n^*) \right| > \epsilon | Z } = 0$, we know also that $\hat{\lambda}_{jn}'(\calG_n^*) \cw \tilde{\lambda}_{jf}'(\calG_P)$ uniformly over $P \in \calP_0^I \backslash \calP_{00}^I$.  Under Assumption~\ref{A:regular}, there is some $\underline{q} > 0$ such that $\inf_{P \in \calP_0^I \backslash \calP_{00}^I} \inf \left\{ q \in \R : \prob{ \tilde{\lambda}_{jf}'(\calG_P) \leq q } \geq 1 - \alpha \right\} \geq \underline{q}$.  Using this fact with another small modification of the second part of Lemma~\ref{lem:uniform_weakconv} implies that there is a $\underline{q} > 0$ such that
  \begin{equation} \label{bootcv_case1}
    \liminf_{n \rightarrow \infty} \inf_{P \in \calP_0^I \backslash \calP_{00}^I} \prob{ \hat{q}_{\lambda^*}(1 - \alpha) \geq \underline{q} } = 1.
  \end{equation}
  On the other hand, the first part of Lemma~\ref{lem:uniform_weakconv} implies, as in Corollary~\ref{cor:null_asymptotics} in the main text but uniformly over $P \in \calP_0^I \backslash \calP_{00}^I$, that for all $\epsilon > 0$, $\limsup_{n \rightarrow \infty} \sup_{P \in \calP_0^I \backslash \calP_{00}^I} \prob{ r_n \lambda_j(f_n) > \epsilon } = 0$.  Combining this fact with~\eqref{bootcv_case1} implies that
  \begin{equation} \label{limit1}
    \limsup_{n \rightarrow \infty} \sup_{P \in \calP_0^I \backslash \calP_{00}^I} \prob{ r_n \lambda_j(f_n) > \hat{q}_{\lambda^*}(1-\alpha) } = 0.
  \end{equation}
  Then~\eqref{limit2} and~\eqref{limit1} together imply the result.
\end{proof}

\bibliographystyle{econometrica}
\bibliography{bib_bound}

\end{document}